\documentclass[11pt]{article}
\pdfoutput=1
\usepackage{bm} 
\usepackage[table,xcdraw,svgnames, dvipsnames]{xcolor}

\usepackage{microtype}
\usepackage[utf8]{inputenc}
\usepackage{graphicx,fullpage,paralist}
\usepackage{amsmath,amssymb,amsthm}
\usepackage{comment,hyperref}
\usepackage{subcaption}
\usepackage{thmtools}
\usepackage{xspace}

\usepackage{nicefrac}
\usepackage{enumitem}
\usepackage[font=footnotesize]{caption}
\usepackage{bbold}
\usepackage[sort,nocompress]{cite}
\usepackage{booktabs}
\usepackage{mathtools}
\usepackage{thm-restate}
\usepackage{cleveref}
\usepackage[rightcaption]{sidecap} 
\usepackage{layout}

\newcommand{\calF}{\mathcal{F}}

\newcommand{\cL}{\mathcal{L}}
\newcommand{\calA}{\mathcal{A}}

\newcommand{\R}{\mathbb{R}}
\newcommand{\RR}{\mathbb{R}}

\newcommand{\N}{\mathbb{N}}
\newcommand{\ZZ}{\mathbb{Z}}

\newcommand{\EE}{\mathbb{E}}
\newcommand{\E}{\mathbb{E}}

\newcommand{\Bin}{\textrm{Bin}}

\DeclareMathOperator{\OPT}{\normalfont{OPT}}

\DeclareMathOperator{\ALG}{\normalfont{ALG}}

\DeclareMathOperator*{\argmax}{arg\,max}

\DeclareMathOperator{\Rel}{Rel}

\DeclareMathOperator{\toprel}{top-rel}

\DeclareMathOperator{\Greedy}{Greedy}
\DeclareMathOperator{\Hull}{Hull}
\DeclareMathOperator{\Core}{Core}
\DeclareMathOperator{\PHull}{PHull}
\DeclareMathOperator{\PCore}{PCore}

\def\set#1{\left\{ #1 \right\}}
\def\abs#1{\left| #1 \right|}

\def\prn#1{\left( #1 \right)}
\def\brk#1{\left[ #1 \right]}

\def\midd{\:\middle|\:}
\def\eps{\varepsilon}
\def\cM{\mathcal{M}}
\def\cF{\mathcal{F}}
\def\cN{\mathcal{N}}
\def\cL{\mathcal{L}}
\def\cI{\mathcal{I}}
\def\cJ{\mathcal{J}}

\def\cB{\mathcal{B}}
\def\cE{\mathcal{E}}

\def\bigO#1{\operatorname{O}\prn{#1}}

\def\bigOm#1{\operatorname{\Omega}\prn{#1}}

\theoremstyle{plain}
\newtheorem{theorem}{Theorem}
\newtheorem{lemma}{Lemma}
\newtheorem{corollary}{Corollary}
\newtheorem{proposition}[theorem]{Proposition}
\newtheorem{claim}{Claim}
{\bfseries}{\itshape} 
\newtheorem*{objective*}{Objective}
\newtheorem{observation}{Observation}

\theoremstyle{definition}
\newtheorem{definition}{Definition}
\newtheorem{example}{Example}

\newtheorem*{reg}{\emph{Running example}}

\usepackage{tikz}
\usepackage{pgfplots}
\pgfplotsset{compat=1.10}
\usepgfplotslibrary{fillbetween}
\usetikzlibrary{backgrounds, patterns, positioning}

\usepackage[textsize=small,textwidth=2cm]{todonotes}

\usepackage{multirow}
\usepackage[ruled,linesnumbered]{algorithm2e}

\hypersetup{
	colorlinks,
	linkcolor={red!50!black},
	citecolor={blue!50!black},
	urlcolor={blue!80!black}
}
\SetKwInput{KwData}{Input}
\SetKwInput{KwResult}{Output}

\definecolor{TODOcolor}{cmyk}{0.05,0,0,0}
\definecolor{TODOtxtcolor}{cmyk}{0,1,1,0}
\definecolor{todocolor}{RGB}{220,236,255}
\definecolor{todotxtcolor}{RGB}{44,131,236}
\definecolor{reviewcolor}{RGB}{159,230,112}
\definecolor{reviewtxtcolor}{RGB}{68, 95, 184}

\newenvironment{todo-box}
   {\fboxrule1.5pt\begin{center}
     \begin{lrbox}{\@tempboxa}
     \begin{minipage}{0.9\columnwidth}
     \small\sffamily\parindent1.5em\color{todotxtcolor}}
   {\end{minipage}
    \end{lrbox}
    \fcolorbox{todotxtcolor}{todocolor}{\usebox{\@tempboxa}}\end{center}}

\newlength{\bibitemsep}
\newlength{\bibparskip}
\let\oldthebibliography\thebibliography
\renewcommand\thebibliography[1]{%
  \oldthebibliography{#1}%
  \setlength{\parskip}{\bibitemsep}%
  \setlength{\itemsep}{\bibparskip}%
}

\makeatletter
\renewcommand{\paragraph}{%
  \@startsection{paragraph}{4}%
  {\z@}{1.6ex \@plus 1ex \@minus .2ex}{-0.2em}%
  {\normalfont\normalsize\bfseries}%
}
\makeatother

\newenvironment{symenum} {\enumerate[label=(\noexpand\thisenumsymbol),align=parleft,labelindent=0pt,itemindent=0pt,labelsep=35pt,leftmargin=*]}
 {\endenumerate}
\newcommand\thisenumsymbol{}
\newcommand\itemsymbol[1]{%
  \renewcommand{\thisenumsymbol}{#1}%
  \item
}

\makeatletter
 \newcommand{\linkdest}[1]{\Hy@raisedlink{\hypertarget{#1}{}}}
\makeatother

\begin{document}

\title{Free-Order Online Selection for $k$-Systems
\vspace{.7cm}}

\author{
Krist\'{o}f B\'{e}rczi
\thanks{MTA-ELTE Matroid Optimization Research Group and HUN-REN--ELTE Egerv\'{a}ry Research Group, Department of Operations Research, E\"{o}tv\"{o}s Lor\'{a}nd University, Budapest, Hungary {\tt (kristof.berczi@ttk.elte.hu)}}
\and
Vasilis Livanos
\thanks{Center for Mathematical Modeling, Universidad de Chile, Santiago, Chile {\tt (vas.livanos@gmail.com)}}
\and
Jos\'{e} A.~Soto
\thanks{Department of Mathematical Engineering and Center for Mathematical Modeling, Universidad de Chile, Santiago, Chile {\tt (jsoto@dim.uchile.cl)}}
\and
Victor Verdugo
\thanks{Institute for Mathematical and Computational Engineering, and Department of Industrial and Systems Engineering, PUC Chile, Santiago, Chile {\tt (victor.verdugo@uc.cl)}}
}
\date{}

\thispagestyle{empty}
\maketitle
\begin{abstract}

The \emph{Matroid Secretary Problem} is a central question in online optimization, modeling sequential decision-making under combinatorial constraints. We introduce a \emph{bipartite graph framework} that unifies and extends several known formulations, including bipartite matching, matroid intersection, and matroid secretary problems. In this model, agents and items form a bipartite graph, and the goal is to select a feasible matching subject to independence constraints on both sides.

We first study the \emph{free-order} setting under edge-arrivals. For \emph{$k$-matroid intersection}, we leverage a core lemma by (Feldman, Svensson and Zenklusen, 2022) to design an $\operatorname{\Omega}(1/k^2)$-competitive algorithm, extending known results for single matroids. Building on this, we introduce \emph{$k$-growth systems} -- a new class of independence systems that lie properly between $k$-matchoids and $k$-extendible systems and may be of independent combinatorial interest. We establish a generalized core lemma for $k$-growth systems, showing that a suitably defined set of \emph{critical elements} retains a $\operatorname{\Omega}(1/k^2)$ fraction of the optimal weight. Using this lemma, we extend our $\operatorname{\Omega}(1/k^2)$-competitive algorithm to $k$-growth systems.

We then study the \emph{agent-arrival} model, which presents unique challenges to our framework. We extend the core lemma to this model and then apply it to obtain an $\operatorname{\Omega}(\beta/k^2)$-competitive algorithm for $k$-growth systems, where $\beta$ denotes the competitiveness of an appropriate type of order-oblivious algorithm for the item-side constraint. Finally, we extend our results to the case of \emph{multiple item selection}, and obtain constant-competitive algorithms for fundamental cases such as partition matroids and $k$-matching constraints.

We also study the structural role of \emph{$k$-growth systems} within the hierarchy of \emph{$k$-systems}. We analyze their closure under key operations, such as parallel extension, restriction and contraction. Our analysis further yields two characterizations: one of \emph{$k$-extendible systems} via contractions of \emph{$k$-systems}, and one of \emph{$k$-circuit bounded systems} via a natural generalization of the matroid circuit axiom.
\end{abstract}
\thispagestyle{empty}
\newpage

\pagenumbering{roman}
\tableofcontents

\newpage
\pagenumbering{arabic}
\setcounter{page}{1}

\section{Introduction}\label{sec:intro}

The \emph{secretary problem} is one of the classical examples in optimal stopping theory. 
In the original setup, one must select the best among $n$ candidates arriving in a uniformly random order. 
No information about the candidates is known in advance. 
Still, the decision-maker can rank the candidates observed so far and make an immediate and irrevocable decision to accept or reject the current one. 
The well-known strategy of observing the first $1/e$ fraction of candidates without accepting any, and then selecting the first one who is better than all previously seen candidates, yields an optimal $1/e$ probability of choosing the best one \cite{lindley,dynkin,historical-secretary}.

Recent years have seen a renewed surge of interest in the secretary problem and its extensions, mainly driven by their deep connections to problems in online mechanism design, where agents arrive sequentially and seek to acquire items or services. 
Beyond this setting, the secretary problem is also closely related to online matching, which encompasses numerous real-world applications, including ride-sharing platforms, online labor markets, advertising auctions, and cloud resource allocation. 
In most applications, the original setup is generalized to a combinatorial framework, where each arriving element $e$ reveals a nonnegative weight $w(e)$, and the objective is to maximize the expected total weight of the selected elements, subject to the chosen set being independent with respect to a given feasibility constraint.

Typical constraints that appear frequently, both in applications and in the literature, include uniform matroids\footnote{For a formal description of matroids see Definition~\ref{def:matroid}.} 
(also known as cardinality constraints) and partition matroids \cite{klein-secr}, matching constraints \cite{korula-pal-graphic,kesselheim-secretary-matching,Ezra-general-graph-arrival-matching-secretary,marinkovic-soto-verdugo}, and knapsack constraints \cite{BabaioffKnapsack,KesselheimPrimalDual,NaoriKnapsack,Abels-KnapsackBoosting,Klimm-Knack-generalknapsack-randomorder}, among others. 
Perhaps due to their rich combinatorial structure and the strong performance guarantees they provide for greedy algorithms, matroid constraints have received particular attention in the literature on the secretary problem. 
Although the \emph{Matroid Secretary Conjecture}\cite{mat-sec} -- which asks whether there exists a constant-competitive algorithm for every matroid -- remains unsolved (the best known competitive ratio is $\Omega(1/\log \log r)$ where $r$ is the rank of the underlying matroid\cite{loglogrank-matroid-sec1, loglogrank-matroid-sec2}), several important subclasses of matroids admit constant-factor competitive ratios, see, e.g., \cite{korula-pal-graphic,soto-secretary,kesselheim-secretary-matching,kristof-liv-jose-victor-label-msp,soto2021strong}.

In many important practical settings, however, the feasibility constraint cannot be captured by a single matroid. Instead, the chosen set must be independent with respect to multiple matroids defined over the same ground set -- in other words, a \emph{matroid intersection} constraint. In this model, Feldman, Svensson, and Zenklusen \cite{fsz-secretary-framework-matroid-intersection} showed that, if each individual matroid admits an \emph{order-oblivious}\footnote{An \emph{order-oblivious} secretary algorithm proceeds in two phases: a sampling phase that draws a random subset of elements, observes their weights, and selects none; and a selection phase that processes all remaining elements, with an arrival order that may be adversarial and even adaptive to the observed sample.} constant-competitive algorithm, then one can combine these to obtain a constant-competitive algorithm for the intersection, provided the number of matroids is constant.
Moving beyond the classical random-order formulation -- where the general matroid secretary conjecture remains unresolved -- one interesting variant is the \emph{free-order} model. In this setting, the algorithm is allowed to adaptively choose the next element to reveal its weight at each step. For this model, Jaillet, Soto, and Zenklusen \cite{free-order-secretary,jaillet-soto-zenklusen-arxiv} devised an algorithm that selects every element in the optimal-weight basis of a matroid with probability $1/4$ -- see Figure~\ref{fig:sampling-matroid}.

Motivated by both theoretical considerations and real-world applications, we introduce a new \emph{bipartite graph model} for the free-order setting. In this model, we are given a bipartite graph consisting of \emph{agents}~$A$ and \emph{items}~$B$, and the objective is to select a matching that satisfies two feasibility constraints: the set of matched agents must belong to a given family~$\cF_A$, and the set of matched items must belong to a given family~$\cF_B$. This formulation generalizes bipartite matching secretary, which corresponds to the case where both~$\cF_A$ and~$\cF_B$ are free matroids\footnote{In a \emph{free matroid}, all subsets of the ground set are independent.}, but is considerably more expressive. By varying~$\cF_A$ and~$\cF_B$, the model is able to capture a broad range of constraint families, including matroid and matroid intersection constraints, and even the random-order matroid secretary problem within a free-order framework. Moreover, approaching free-order selection through this two-sided matching perspective reveals a close relationship with the classical \emph{online bipartite matching} setting since, for instance, our model can naturally accommodate \emph{agent-arrival} settings.

\begin{reg}
To motivate our bipartite graph model, consider the following example. Agents represent facilities providing different services, such as water, electricity, or internet, while items correspond to clients. The weight of an edge $\set{a,b}$ represents the utility obtained if facility~$a$ services client~$b$. We note that this example is simplified: in realistic scenarios, a user may connect to multiple providers and providers may serve many clients. Nevertheless, it highlights the essential features of the free-order model and clarifies the behavior of our algorithms. When both $\cF_A$ and $\cF_B$ are free, the model reduces to a standard bipartite matching problem; later we show how our framework allows us to handle more complex variants of this setting.
\end{reg}

In what follows, we introduce our models, recall standard notions and prior results, and explain how our framework relates to previous work, to give readers a comprehensive understanding of the context in which our results operate. Readers already familiar with the literature may safely skip the discussion of prior work.

\subsection{Overview of our Contributions and Techniques}

While the matroid secretary problem has been extensively studied, much less is known when both sides impose feasibility constraints, even in the free-order model. Existing techniques focus on a single independence system or combine guarantees via order-oblivious algorithms, and do not extend naturally to this setting. With constraints on both sides, even basic questions -- such as the existence of constant-competitive algorithms or how the two sides interact -- become more subtle. In this paper, we leverage our new \emph{bipartite graph model} to investigate the free-order secretary problem under both the \emph{edge-arrival} and \emph{agent-arrival} settings. At a high level, our aim is to build on and significantly extend recent sampling-based techniques for free-order secretary problems. Our results consist of one conceptual and three algorithmic contributions.

The first challenge is that, beyond matroids, the standard notion of span no longer controls which elements remain relevant after sampling, and interactions among multiple constraints can break the structural guarantees of existing analyses. Our first key insight is that the core lemma due to Feldman, Svensson and Zenklusen can be abstracted to a general principle: a carefully defined set of \emph{critical} elements -- determined by a \emph{span-like operator} and a notion of \emph{relevance} -- retains a constant fraction of the optimum.

This observation allows us to design a constant-competitive algorithm for the free-order secretary problem under $k$-matroid intersection constraints, when $k$ is constant. Motivated by this insight, we introduce a new family of independence systems, which we call \emph{$k$-growth systems}. In our second algorithmic contribution, we show that $k$-growth systems still satisfy a version of the Feldman--Svensson--Zenklusen core lemma for this carefully defined set of critical elements. Interestingly, $k$-growth systems are far from just a technical choice for the Feldman--Svensson--Zenklusen core lemma: they generalize several natural independence constraints such as $k$-matchoids, bounded-ratio knapsacks and stable sets of intersection graphs, they are properly contained within the hierarchy of $k$-systems (see Section~\ref{sec:higher-order-systems} for definitions), and they exhibit several useful closure properties. For these reasons, $k$-growth systems constitute our main conceptual contribution. Moreover, $k$-growth systems are closed under combination in the style of $k$-matchoids. We present our results for $k$-growth systems to illustrate both their modeling power as well as the strength of our generalized core lemma.

Finally, in our third algorithmic contribution, we extend our setting to \emph{agent arrivals}, which introduces further challenges due to correlated information and limited observability. To overcome these challenges, we further refine our notion of critical elements and use a multi-phase sampling scheme.

Below, we describe each of our contributions and techniques in detail.



\paragraph{$k$-matroid intersection.}
We begin with $k$-matroid intersection. The algorithm of Jaillet, Soto, and Zenklusen~\cite{jaillet-soto-zenklusen-arxiv} provides a natural starting point: sample half of the elements, order the optimal solution within the sample by decreasing weight, and then process the remaining elements in an order according to the span of the prefixes of this sampled optimum. Unfortunately, this approach cannot be generalized to matroid intersection directly due to the following two obstacles: first, the spans of the individual matroids may require completely different orders and second, even with a favorable order, it is a priori unclear why the intersection of the individual optimal sets retains enough weight compared to the global optimum. The latter problem can be resolved using a lemma by Feldman, Svensson and Zenklusen \cite{fsz-secretary-framework-matroid-intersection}: by sampling a large fraction of the elements, they show that the so-called \emph{greedy-relevant} elements of the second phase retain a decent fraction of the optimum weight -- we refer to this as the \emph{core lemma} for matroid intersection (Lemma~\ref{lem:fsz-original-lemma}). Equipped with this structural insight, we design an algorithm that extends the approach of~\cite{jaillet-soto-zenklusen-arxiv}, by using the union of the individual spans of the matroids. Our algorithm selects every element in the intersection of the optimal sets of the greedy-relevant elements with probability $1/4$ (Theorem~\ref{thm:alg-general}). More importantly, this algorithm serves as a versatile black-box that applies beyond $k$-matroid intersection, whenever the underlying constraint family admits an appropriate analogue of the core lemma.

\begin{reg}
Suppose now that each facility requires certain raw materials or resources supplied by a central authority in order to operate. We can represent these dependencies in a different graph, where raw material providers are sources and the facilities are sinks. For each resource, the subset of facilities that can be served simultaneously forms a gammoid, represented by this new graph that is separate from the original bipartite graph. \footnote{A \emph{gammoid} is a special type of matroid in which a set $S$ of vertices in a given graph is independent if there exist vertex-disjoint paths from $S$ to a designated sink.} Hence, the global constraint on the agent side, i.e. on the facilities, can be viewed as an intersection of gammoids. Notably, although no constant-factor guarantees are known for the secretary problem with gammoid constraints in the random-order setting; in the free-order case, our bipartite graph framework yields constant-factor guarantees even when multiple gammoids are present.
\end{reg}

\paragraph{$k$-growth systems.}
We then generalize this result in two different directions. First, we revisit the core lemma by Feldman, Svensson and Zenklusen. We identify the key structural characteristic that underpins the lemma and, motivated by this property, we introduce a new family of independence systems which we call \emph{$k$-growth systems}. This new family fits nicely in the existing hierarchy of $k$-systems: it contains $k$-matchoids and, in turn, is contained within the class of $k$-extendible systems. Intuitively, $k$-growth systems require a stronger $k$-extendability property, not just for independent sets but for arbitrary sets.

At first glance, the definition of $k$-growth systems seems quite technical, so one might wonder whether they form a natural class within the $k$-system hierarchy. Beyond the closure properties discussed above, we clarify their place in this hierarchy by showing that $k$-growth systems properly contain $k$-matchoids and are properly contained in $k$-extendible systems. The same strict containments hold for $k$-circuit-bounded systems, which naturally prompts the question of how these two classes relate to one another. We show that they are incomparable: there exist $k$-growth systems that are not $k$-circuit-bounded, as well as $k$-circuit-bounded systems that are not $k$-growth. Thus, these notions reveal that the the $k$-systems hierarchy does not form a single chain of containment, but instead contains distinct and complementary structural classes.

We further motivate their study by showing that the parameter $k$ has a natural interpretation in many well-studied constraints. Specifically, for knapsacks, stable sets of intersection graphs, capacitated intervals and sparse packing constraints, $k$ serves as a bound on the ratio of sizes associated with the elements. For example, for a knapsack constraint, $k$ corresponds to the ratio of the maximum-size over minimum-size feasible subsets. When selecting a stable set of intersection graphs, e.g. a stable set of intersecting disks on the plane $k$ serves as the ratio of maximum to minimum radii of the disks. For the case of selecting at most $c$ intervals with lengths in $[a,b]$, $k$ corresponds to $c \cdot b/a$. Notice that this also captures unit-demand unsplittable flow on paths (UFP), by viewing the path edges as the ``points'' of the interval system. Moreover, one can interpret $d$-sparse packing constraints as a combination of knapsack constraints and these can be seen as $k$-growth systems for $k = d \cdot \max_j \rho_j$, where $\rho_j$ is the ratio of the maximum-size over minimum-size feasible subsets in the knapsack corresponding to the $j$-row. Thus, $k$-growth systems capture a broad class of constraints whose obstruction is local and whose complexity is controlled by bounded ratios.

To generalize the core lemma to $k$-growth systems, we first need to generalize the notion of span from matroids to $k$-growth systems. We introduce a new ``span-like'' operator, the \emph{primitive hull}, which, for any set $X$, consists of $X$ together with all elements $e$ that form a circuit with $X$. We note that computing the primitive hull can be challenging in general, though it is efficient for many natural constraints such as matroids or knapsacks. More broadly, our black-box algorithm primarily illustrates what is achievable information-theoretically, rather than guaranteeing fully efficient procedures for all systems.

We then generalize our core lemma to $k$-growth systems: we identify a set of \emph{critical} elements, those elements $e$ that do not lie in the primitive hull of the greedy-relevant elements of higher weight. We prove that, for every $k$-growth system, the set of critical elements preserves a $\bigOm{1/k^2}$ fraction of the weight of the optimal basis (Lemma~\ref{lem:core-k-growth-general}, Theorem~\ref{thm:core-lemma-vs-opt}). Substituting our new core lemma into the black-box algorithm from Theorem~\ref{thm:alg-general} immediately yields a $\bigOm{1/k^2}$-competitive algorithm for the free-order secretary problem over any $k$-growth system (Theorem~\ref{thm:growth-systems-alg}). In fact, we establish a stronger result; we introduce the notion of \emph{combination} of systems, which generalizes the classical $k$-matchoid construction, allowing the individual components of the overall system to be arbitrary constraint families rather than only matroids. We prove that any combination of growth systems, where the $i$-th system $M_i$ is a $k_i$-growth system, is itself a $k$-growth system, with $k = \max_e \sum_{i : e \in M_i} k_i$, the natural analogue of this parameter in $k$-matchoids. See Figure~\ref{fig:sampling-growth} for an illustration of our algorithm's phases.

\begin{reg}
In addition to resource constraints, each facility may also incur an operational cost, while a central authority operates under a fixed budget for servicing clients. This naturally induces a knapsack constraint on the set of facilities. The class of growth systems is particularly appealing here, as it accommodates the combined feasibility constraint of (bounded-ratio) knapsack and matroid intersection constraints.
\end{reg}

\paragraph{Agent-arrival setting.}
The strength of the bipartite graph model becomes evident when we turn to the \emph{agent-arrival} setting. Unfortunately, the results established for edge arrivals do not translate directly, as the agent-arrival setting introduces three main challenges. The first difficulty arises from the loss of partial observability. In the edge-arrival model, sampling a fraction of edges reveals information about some edges incident to each agent. Under agent arrivals, ``calling'' an agent reveals all of its incident edges simultaneously, and the agent cannot be matched later on. To address this, we switch from reasoning about all \emph{greedy-relevant edges} to a single \emph{top-relevant edge} per agent. This is the heaviest greedy relevant edge incident to that agent. Since the matching constraint allows at most one edge per agent, we can canonically associate each agent with its top-relevant edge. This perspective lets us extend the core lemma for $k$-growth systems to the agent-arrival setting, showing that the set of top-relevant edges still retains an $\bigOm{1/k^2}$ fraction of the optimal basis weight (Lemma~\ref{lem:core-agent-arrival}).

The second obstacle arises when attempting to use our general algorithm from Theorem~\ref{thm:alg-general} in a black-box manner. The algorithm requires querying agents in an order determined by the primitive hull on the agent side, ensuring that agents whose greedy-relevant edges are critical with respect to $\calF_A$ are considered with constant probability. This order, however, provides no guarantee on the item side -- in fact, the selected edges may form a dependent set in $\calF_B$. To address this, we restrict $\calF_B$ to systems that admit a special \emph{core-selecting order-oblivious} algorithm. Such an algorithm preserves its competitive guarantee even when, after a uniformly random sample phase, every \emph{critical} element of $\calF_B$ that was not sampled is selected with constant probability, regardless of the arrival order of the elements in the sampling phase. Notice that this requirement is different from that of a classical order-oblivious secretary algorithm as, for systems beyond matroids, the critical elements need not be in the optimal solution.

This, however, reveals the third, and most subtle, obstacle in our approach: When we sample agents, we simultaneously observe all edges incident to those agents. From the perspective of the item side, this does not induce a uniformly random sample of items or even of item-adjacent edges, since edges corresponding to the same agent are inherently correlated. Notice, however, that after our first sample, each agent contributes at most one greedy-relevant edge -- their top greedy-relevant edge, if one exists. Therefore, to overcome this final obstacle, after the first sample, we restrict our ground set to the (unknown) set of greedy-relevant edges. We can then perform a second, independent sampling phase, ensuring that the sampled greedy-relevant edges form a uniform random subset of the new ground set; for more details on this, see Section~\ref{sec:agent-arrival}. Finally, applying the algorithm from Theorem~\ref{thm:alg-general} as a black box (which introduces a third and final sampling phase) yields an $\bigOm{\beta / k^2}$-competitive algorithm for the agent-arrival model, where $\beta$ is the competitive ratio of the order-oblivious algorithm for $\calF_B$ (Theorem~\ref{thm:alg-general-agent}) -- see Figure~\ref{fig:sampling-order-oblivious}. This result combines the structural insights of the edge-arrival model with the algorithmic challenges of agent arrivals, and is thus the main algorithmic result of the paper.

\begin{reg}
Suppose that the items actually represent contracts between agencies and clients, with each facility responsible for servicing one contract. When a facility arrives, we observe all contracts that it could service. The selected contracts must form a matching between facilities and clients, corresponding to the item-side constraint $\cF_B$. Because this constraint admits a core-selecting, order-oblivious algorithm, our bipartite graph framework yields constant-factor guarantees for free-order agent arrivals in this setting.
\end{reg}

\begin{figure}[t]
    \centering
    \begin{subfigure}{\textwidth}
        \centering
        \includegraphics[width=0.5\linewidth]{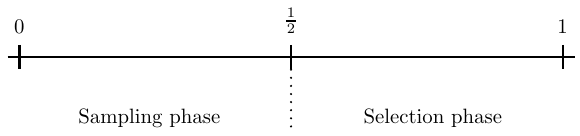}
        \caption{Illustration of the algorithm of Jaillet, Soto, and Zenklusen for a single matroid (Theorem~\ref{thm:jsz-free-order-matroid}), with a single sampling phase in $[0,1/2)$ and a selection phase in $[1/2, 1]$.}
        \label{fig:sampling-matroid}
    \end{subfigure}\\[5pt]
    \begin{subfigure}{\textwidth}
        \centering
        \includegraphics[width=0.5\linewidth]{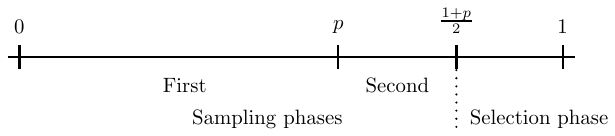}
        \caption{Illustration of our algorithm for combinations of $k$-growth systems in the edge-arrival setting (Theorem~\ref{thm:growth-systems-alg}), with a first sampling phase in $\left[0, p\right)$ to define the greedy-relevant elements in subsequent phases, a second sampling phase in $[p, (1+p)/2)$ used by our black-box algorithm (Theorem~\ref{alg:free-general}) to decide the order in which to ``call'' the greedy-relevant elements, and a selection phase in $[(1+p)/2, 1]$. Here, $p = \sqrt{1 - 1/(k+1)}$ for arbitrary $k$-growth systems and $p = 1 - 1/(2k)$ for $k$-matchoids.}
        \label{fig:sampling-growth}
    \end{subfigure}\\[5pt]
    \begin{subfigure}{\textwidth}
        \centering
        \includegraphics[width=0.5\linewidth]{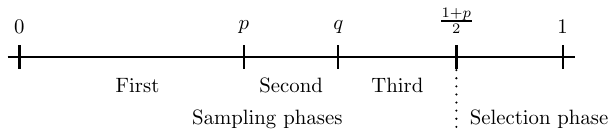}
        \caption{Illustration of our algorithm for combinations of $k$-growth systems in the agent-arrival setting (Theorem~\ref{thm:alg-general-agent}), with a first sampling phase in $[0, p)$ to define the greedy-relevant elements in subsequent phases, a second sampling phase in $[p, q)$ required by the corresponding order-oblivious algorithm for $\calF_B$, a third sampling phase in $[q, (1+q)/2)$ used by our black-box algorithm (Theorem~\ref{alg:free-general}) to decide the order in which to ``call'' the greedy-relevant agents, and a selection phase in $[(1+q)/2, 1]$. Here, $p = \sqrt{1 - 1/(k+1)}$ for arbitrary $k$-growth systems and $p = 1 - 1/(2k)$ for $k$-matchoids, and $q$ depends on the sampling probability of the order-oblivious algorithm for $\calF_B$.}
        \label{fig:sampling-order-oblivious}
    \end{subfigure}
    \caption{A visual explanation of the sampling and selection phases of each of our algorithms.}
    \label{fig:interval-sampling}
\end{figure}

\paragraph{Beyond matching constraints.}
We next consider settings in which each agent may select multiple items. Extending the core lemma to this setting remains challenging -- even for the intersection of two matroids -- so we take a different route and focus on the case of \emph{unrelated agents}. Here, each agent possesses their own, separate independence system governing the feasible subsets of its incident edges, and feasibility across agents is coupled only through a global item-side constraint. In this setting, agents arrive sequentially, reveal the weights of all adjacent edges, and the algorithm must irrevocably select a feasible subset for each agent. Interestingly, this model captures a rich class of problems in online selection, including the \emph{secretary problem with groups} of Korula and P\'{a}l \cite{korula-pal-graphic}, as well as variants where items represent edges in low-degree graphs or hypergraphs and feasible sets correspond to matchings~\cite{marinkovic-soto-verdugo,kesselheim-secretary-matching}. We design a constant-competitive algorithm (Theorem~\ref{thm:multiple-items}) for the agent-arrival setting when the item-side constraint is sufficiently simple, for example, a partition matroid or a $k$-matching constraint, demonstrating that our bipartite free-order framework extends naturally to heterogeneous and structurally diverse environments. It is worth noting that this approach leads to an $\Omega(1/k^2)$-competitive algorithm for combinations of systems admitting so-called $k$-directed certifiers~\cite{marinkovic-soto-verdugo}, which includes combinations of several classes of matroids.

\begin{reg}
Finally, when the constraints governing different facilities are independent of one another, we can allow each facility to serve multiple contracts while the contract side still satisfies a matching constraint. Our framework continues to provide constant-factor guarantees in this more general setting.
\end{reg}

\paragraph{The $k$-systems hierarchy.}
Finally, we examine the position of $k$-growth systems within the broader $k$-systems hierarchy and study their relationship to other well-known classes. As discussed earlier, we show that every $k$-matchoid is a $k$-growth system and present examples of $k$-growth systems that are not $k$-circuit bounded and of $k$-circuit bounded systems that are not $k$-growth. We also show that the class of $k$-growth systems is, in turn, properly contained within the class of $k$-extendible systems. 
We illustrate the expressiveness of $k$-growth systems through several examples, including bounded-ratio knapsack constraints and stable sets of interval graphs. Moreover, we show that all these classes, except for $k$-systems, are closed under parallel extension, restriction, and contraction. More importantly, we provide two new characterizations of well-studied classes in the $k$-systems hierarchy. Specifically, we show that $k$-extendible systems as exactly those $k$-systems whose contractions remain $k$-systems, and $k$-circuit bounded systems are the systems that satisfy a natural generalization of the circuit axiom for matroids.

\subsection{Related Work}

\paragraph{Online Matchings and Combinatorial Assignment.}
The baseline is unweighted one-sided bipartite online matching, where the \emph{Ranking} algorithm achieves the optimal ratio $1-1/e$~\cite{KVV}. Secretary-style generalizations introduce random arrival, leading to two standard models in general graphs: \emph{vertex arrival} (upon a vertex's arrival, incident edges to previously seen vertices are revealed and must be acted on immediately) and \emph{edge arrival} (edges appear one by one, each requiring an irrevocable decision). Early secretary formulations on graphs and hypergraphs are due to Korula and P\'al~\cite{korula-pal-graphic}. For general graphs, Ezra, Feldman, Gravin, and Tang~\cite{Ezra-general-graph-arrival-matching-secretary} established tight bounds of $5/12$ for vertex arrival and an exponential-time $1/4$ for edge arrival, with the latter later made polynomial-time in~\cite{marinkovic-soto-verdugo}. In bipartite weighted settings with budgets (AdWords), online allocation under random order was modeled and analyzed by Mehta, Saberi, Vazirani, and Vazirani~\cite{mehta2007adwords}, while related secretary-style approaches for weighted bipartite matching and its extensions were developed by Kesselheim, Radke, T\"{o}nnis, and V\"{o}cking~\cite{kesselheim-secretary-matching}. Beyond graphs, \emph{online submodular welfare} captures diminishing-returns valuations over arriving items and admits a $1/2$-competitive greedy algorithm, which is optimal among polynomial-time algorithms under standard assumptions~\cite{kapralov2013online}. A unifying abstraction for these problems is \emph{online combinatorial assignment} over independence systems, where agents arrive with (possibly multi-weight) preferences over a fixed feasibility family (e.g., $k$-hypergraph matchings, matroids, matroid intersections, matchoids), and the algorithm assigns at most one feasible element per agent while maintaining independence~\cite{marinkovic-soto-verdugo}.

\paragraph{Prophet Inequalities and Online Contention Resolution Schemes.}
Secretary problems are closely connected to \emph{prophet inequalities}, a model that has attracted significant attention due to its applications in mechanism design and auction theory \cite{haji,ChawlaHMS,chawla07, qiqi, BlumHolen08, Adamczyk17, Alaei14, feldman-combinatorial, Dutting2, kesselheim-mhr-ppm, Dobzinski-Comb-Auc, Dobzinski-Comb-Auct-Impos, AssadiKS, assad-singla, dutting-combinatorial,correa-cristi}. In the prophet inequality framework, the element weights $w(e)$ are drawn independently from known distributions, and the arrival order is typically adversarial. The classical single-item result of Krengel, Sucheston, and Garling~\cite{kren-such,kren-such2,sam-cahn} established the optimal $1/2$-competitive ratio, later extended to matroid constraints by Kleinberg and Weinberg~\cite{klein-wein}, who also obtained a $1/(4k-2)$-competitive ratio for intersections of $k$ matroids. Recently, \cite{WeinbergMatroidProphetLB} showed an inapproximability result that is slightly better than $O(1/\sqrt{k})$. In the IID prophet inequality setting, where all elements are drawn independently from identical distributions, the optimal competitive ratio is known to be approximately $0.7451$~\cite{correa-iid}. This hardness carries over to the free-order setting: any inapproximability result that applies to the IID case also applies to free-order algorithms, so the problem remains fundamentally constrained by the same structural limitations. Early progress in the free-order model was made by Beyhaghi, Golrezaei, Paes Leme, P\'{a}l and Sivan~\cite{beyhaghi-fo} and Peng and Tang~\cite{peng-tang}, with the current best competitive ratio of $0.7258$ due to Bubna and Chiplunkar~\cite{bubna-chiplunkar}. Determining the optimal competitive ratio in this setting remains a central open question in optimal stopping theory.

Secretary problems are also closely related to \emph{online contention resolution schemes} (OCRSs), which play a central role in Bayesian mechanism design and posted-price mechanisms~\cite{ocrs,rubin,rs,chek-liv}. OCRSs share many conceptual similarities with prophet inequalities and can be combined in a black-box way, both properties that align with several ideas developed in this paper. Feldman, Svensson, and Zenklusen~\cite{ocrs} introduced the first OCRSs for matroid, matching, and knapsack constraints, achieving a $1/(e(k+1))$ guarantee for $k$-matroid intersections. Lee and Singla~\cite{lee-singla} later obtained optimal guarantees for matroids under a slightly stronger assumption of a weaker adversary. Despite this progress, determining the optimal OCRS guarantees for matroids, matchings, and knapsacks remains an active area of research \cite{liv-tight-greedy-ocrs,efgt,Ezra-general-graph-arrival-matching-secretary,tristan,ma-rocrs-matching,NutiCRSMatching,brubach,MaCalumNutiVanishingCRS}. Recently, Dughmi~\cite{dughmi1,dughmi2,Dughmi2025Combined} established a deep and beautiful equivalence: obtaining constant-factor OCRSs for matroid constraints in the random-order setting -- for any correlated distribution for which such guarantees exist offline -- is equivalent to resolving the matroid secretary conjecture.

\paragraph{The $k$-system Hierarchy.}
The notion of a $k$-system goes back to the classical analysis of greedy algorithms for independence systems by Jenkyns~\cite{jenkyns1976efficacy} and by Korte and Hausmann~\cite{korte1978analysis}, who characterized approximation guarantees in terms of the rank quotient, that is, the ratio of the largest and smallest bases of a subset. This concept was later adopted under the name $p$-system in submodular maximization~\cite{calinescu2011maximizing}. Jenkyns~\cite{jenkyns1976efficacy} also implicitly introduced the notion later referred to as $k$-circuit-bounded systems, although no explicit name was proposed in that work. The class of $k$-extendible systems was introduced by Mestre~\cite{mestre2006greedy}, fitting into the hierarchy
\begin{align*}
\text{$k$-matroid intersection} \subseteq \text{$k$-matchoid} \subseteq \text{$k$-circuit bounded} \subseteq \text{$k$-extendible} \subseteq \text{$k$-system}.
\end{align*}
Most previous work on $k$-systems has been motivated by algorithmic applications, notably the analysis of greedy procedures and the approximation of submodular maximization over independence systems. Consequently, $k$-systems are well understood from the viewpoint of performance guarantees, while their purely combinatorial structure has received comparatively less attention.

Two closely related combinatorial optimization problems are \emph{matroid $k$-parity} and
\emph{$k$-uniform matroid matching}, introduced by Lee, Sviridenko, and
Vondrák~\cite{lee2013matroid} and studied further in subsequent work.
The special case $k=2$ recovers the classical matroid parity and matroid
matching problems, introduced by Lawler~\cite{Lawler1971} and
Lovász~\cite{Lovasz1980}. From an algorithmic viewpoint, the maximum
feasible set in either problem extends the notion of a maximum feasible set
in a $k$-matchoid; from an oracle viewpoint, these models are strictly
more general (see Huang and Ward~\cite{HuangWard2023} and references
therein). From the perspective of independence systems, both can also be
viewed as generalizations of $k$-matchoids. They are contained in the class of $k$-extendible systems, but we do not otherwise elaborate on
their precise location within the hierarchy above.


\subsection{Organization}

We start with Section~\ref{sec:model}, where we describe our bipartite graph model and present some necessary preliminaries and definitions. In Section~\ref{sec:k-matroid-intersection} we describe our algorithm for $k$-matroid intersection, presented in a black-box way as it will be used for our later results as well. In Section~\ref{sec:lemma-extension} we define our new independence class of $k$-growth systems and present our generalized Core lemma for combinations of $k$-growth systems. Section~\ref{sec:agent-arrival} contains our results for the agent-arrival model. We then extend the agent-arrival setting to allow multiple items per agent when the agents' constraints are unrelated in Section~\ref{sec:multiple-items}. Section~\ref{sec:higher-order-systems} relates $k$-growth systems to the other classes in the $k$-systems hierarchy and describes several interesting properties of the hierarchy classes. We conclude the paper by a list of open problems in Section~\ref{sec:discussion}.
\section{The Bipartite Graph Model}\label{sec:model}

In this section, we define our proposed bipartite graph model for free-order secretary problems. We begin with the basic notation; all further terminology is introduced as needed. We denote the sets of \textit{reals} and \textit{integers} by $\R$ and $\ZZ$, and add the subscripts $+$ when referring to positive values. For a positive integer $k$, we write $[k] \coloneqq \set{1,\dots,k}$. Given a ground set $S$, a subset $X\subseteq S$, and an element $y\in S$, we abbreviate $X\setminus\{y\}$ and $X\cup\{y\}$ as $X-y$ and $X+y$, respectively, and write $\{x\}$ simply as $x$ when no confusion arises. For a $p \in [0,1]$ and a set $S$ of elements, we use $S(p)$ to denote a random subset of $S$ where every element appears independently with probability $p$. For a function $w\colon S\to\R_+$, we use $w(X)\coloneqq\sum_{x\in X}w(x)$. If $G=(V,E)$ is a graph, and $S\subseteq V$ is a subset of vertices, we use $\delta(S)$ to denote the set of edges with one endpoint in $S$ and one endpoint outside~$S$.
Let us first recall the definition of independence systems and matroids.

\begin{definition}[Independence System]\label{def:indep-system}
An \emph{independence system} $\cF = (S,\cI)$ consists of a finite set $S$ and a family of \emph{independent sets} $\cI \subseteq 2^S$ satisfying the following two properties:
\begin{enumerate}[label=(I\arabic*)]\itemsep0em
    \item $\cI$ is non-empty, \label{it:nonempty}
    \item $\cI$ is downward closed, that is, $I\in\cI$ implies $J\in\cI$ for all $J\subseteq I$. \label{it:downward}
\end{enumerate}
Note that conditions~\ref{it:nonempty} and~\ref{it:downward} together imply $\emptyset\in\cI$. Any set $X \notin \cI$ is called \emph{dependent}. The inclusion-wise minimal dependent sets are called \emph{circuits}, and a circuit of size $1$ is called a \emph{loop}. For any subset $X \subseteq S$, the inclusion-wise maximal independent subsets of $X$ are called \emph{bases} of~$X$.
\end{definition}

Examples of independence systems include matroids, matchings in hypergraphs,
matchoids, $k$-extendible systems, and $k$-systems (see
Section~\ref{sec:higher-order-systems} for details and definitions). Below, we provide a few basic definitions for matroids and refer the reader to \cite{oxley-matroids} for an extensive treatment.

\begin{definition}[Matroid]\label{def:matroid}
A matroid $M=(S,\cI)$ is an independence system that satisfies the following \emph{augmentation} property:
\[
\forall\, I,J\in\cI,\ |I|<|J|
\ \Longrightarrow\ \exists\, e\in J\setminus I \text{ such that } I\cup\{e\}\in\cI.
\]
This property is equivalent to the statement that, for every $X\subseteq S$,
all maximal independent subsets of $X$ (the bases of $X$) have the same
cardinality. With this, one defines the rank of $X$, denoted by $r(X)$, as the
common size of the bases of $X$. The closure (or span) of $X$ is
\[
\mathrm{span}(X)\;=\;X\ \cup\ \{e\in S\setminus X:\ r(X\cup\{e\})=r(X)\}.
\]
\end{definition}

Throughout the paper, we consider independence systems equipped with a nonnegative weight function $w \colon S \to \mathbb{R}_+$ on the ground set, with distinct values. In matroids, the maximum-weight basis under distinct weights is unique. In general independence systems, however, the maximum-weight member of the system may not be unique even with distinct weights. One can fix uniqueness via an appropriate tie-breaking rule, such as a lexicographic rule; thus, assuming uniqueness of a maximum-weight member is standard and usually without loss of generality. This assumption eliminates tie-breaking issues, allowing us to focus on structural and algorithmic aspects of the problem. For $X \subseteq S$, let $\OPT(X)$ be a maximum-weight basis of $X$ with respect to $w$.

\begin{definition}[Weighted greedy on $X$]\label{def:greedy}
Given $X\subseteq S$, the \emph{weighted greedy algorithm} processes the
elements of $X$ in non-increasing order of $w$ and builds a set $G$ starting
from $\emptyset$, adding an element $e$ when $G\cup\{e\}\in\cI$. We denote its
output by $\Greedy(X)$.
\end{definition}
For the next proposition, recall that an independence system is a $k$-system if, within every subset of the ground set, the ratio between the sizes of any two maximal independent subsets is at most $k$; the
definition is repeated in Section~\ref{sec:higher-order-systems}. The following proposition, by Korte and Hausmann \cite{korte1978analysis}, relates the weighted greedy algorithm to $k$-systems.

\begin{proposition}\label{prop:greedy-matroid-ksystem}
For matroids, $\Greedy(X)$ is the unique maximum-weight basis of $X$ under
distinct weights, i.e., $\Greedy(X)=\OPT(X)$. More generally, an independence
system is a $k$-system if and only if, for every $X\subseteq S$ and every
nonnegative weight function $w$, the greedy algorithm yields a $k$-approximate maximum-weight basis, i.e.,
\[
w(\Greedy(X)) \ge \frac{1}{k} w(\OPT(X)).
\]
\end{proposition}

\subsection{Free-order Bipartite Secretary Model}

Let $G = (A \cup B, E)$ be an undirected bipartite graph. We refer to $A$ as the ``agent'' side and $B$ as the ``item'' side. Let $w\colon E\to\R_+$ be a nonnegative injective weight function over the edge set. Furthermore, let $\cF_A = (A, \cI_A)$ and $\cF_B = (B, \cI_B)$ denote independence systems over the vertices on the agent and item side, respectively.
In the \emph{bipartite secretary model}, we are given $G$, $\cF_A$, and $\cF_B$ as input.\footnote{Similarly to matroid algorithms, we assume that independence systems are given via an independence oracle that, for a given set $X$, answers yes or no depending on whether $X$ is independent or not.} We say that a set of edges $F\subseteq E$ is a \emph{feasible matching} if
\begin{enumerate}\itemsep0em
    \item[$(i)$] $F$ is a matching in $G$.
    \item[$(ii)$] The set of agents incident to $F$, $\set{a \in A \midd \delta(a)\cap F\neq \emptyset}$, is independent in $\cF_A$.
    \item[$(iii)$] The set of items incident to $F$, $\set{b \in B \midd \delta(b)\cap F \neq \emptyset}$ is independent in $\cF_B$.
\end{enumerate}
If an edge $\set{a,b}$ is in $F$ we say that item $b$ is assigned to agent $a$.
The weight function $w$ over the edges is initially unknown and revealed over sequential rounds. An algorithm for this problem starts with an empty set of edges $\ALG$. In round $i$, one or more edges may be revealed, and at the end of the round, the algorithm must immediately and irrevocably decide to add a subset $F_i$ of the edges revealed in that round to the current solution $\ALG$, under the condition that $\ALG$ is always a feasible matching. The goal is to select a maximum-weight feasible matching. We focus on two versions of the model, differing in how rounds are defined:
\begin{itemize}\itemsep0em
  \item[(EA)] \textbf{Edge Arrival:} at each round, a single new edge reveals its weight.
  \item[(AA)] \textbf{Agent Arrival:} at each round, the set $\delta(a)$ of all edges incident to some agent $a\in A$ reveal their weights simultaneously.
\end{itemize}

In the \emph{free-order} setting, the algorithm adaptively chooses the next
object to be revealed: the next edge in (EA) or the next agent in (AA). This formulation is remarkably powerful: it generalizes several settings
previously studied in the literature, even when $\cF_A$ and $\cF_B$ are
simple constraints such as matroids. We briefly discuss these connections below.

\subsection{Relation to Existing Models}

We begin by noting that the free-order edge-arrival variant is exactly the \emph{free-order secretary problem} over the independence system of feasible matchings induced by \((G,\cI_A,\cI_B)\). Motivated by this viewpoint, we develop our results and state existing ones for the free-order model on general independence systems. Let us focus on the case when both $\cF_A$ and $\cF_B$ are matroids. We distinguish between different cases according to their respective matroid classes.

\paragraph{$\cF_A$ and $\cF_B$ are both free matroids.}
In this setting, the edge arrival model coincides with the bipartite matching secretary problem, for which there exists a $1/4$-competitive (polynomial-time) algorithm even in random-order~\cite{marinkovic-soto-verdugo}. For the agent arrival model, there is an optimal $1/e$-competitive algorithm in random-order~\cite{kesselheim-secretary-matching}.

\paragraph{$\cF_A$ is a general matroid and $\cF_B$ is a free matroid.}
For the specific case in which the graph $G$ itself is a matching, both the edge and the agent-arrival model coincide with the free-order matroid secretary problem, for which there exists a $1/4$-competitive algorithm~\cite[Theorem~\ref{thm:jsz-free-order-matroid}]{jaillet-soto-zenklusen-arxiv}. For general graphs, neither version has been studied before. It is worth noting that this setting corresponds to the intersection of a general matroid and a partition matroid, which inherits results from the more general case of intersecting two arbitrary matroids.

\paragraph{$\cF_A$ is a free matroid and $\cF_B$ is a general matroid.} Here, the edge arrival model is the same as in the above setting where $\cF_A$ is a general matroid and $\cF_B$ is the free matroid.
Surprisingly, the agent arrival model for this case is at least as hard as the classical matroid secretary problem. To see this, let $G$ be a complete bipartite graph. The adversary can choose a random perfect matching $X$ of the graph and assign non-zero weights only to edges in $X$. Every time we observe the weights of edges incident to an agent in any order, we only observe one edge with positive weight, and we can only select it at that moment. Since all agents look the same before they arrive, the order in which they arrive does not matter, and we are forced to select them uniformly at random. Therefore, we recover the classical matroid secretary problem.

\paragraph{$\cF_A$ and $\cF_B$ are both general matroids.} Note that the set of edges we can choose at the same time actually forms an independent set in the intersection of $\cF_A$ and $\cF_B$. To see this, consider the following standard construction. Let $M_A = (E, \cJ_A)$ be the matroid where, for each agent $a \in A$, all edges in $\delta(a)$ are treated as parallel copies of the same element in $\cF_A$. This way, each agent $a$ can be matched to at most one item $e$ and the constraint for $\cF_A$ is satisfied. Using the same approach, define another matroid $M_B = (E, \cJ_B)$ using the item side. By construction, any feasible matching is an independent set that belongs to both $M_A$ and $M_B$. In the edge arrival model, no constant-competitive guarantee was known for this setting prior to our work. We obtain a $1/64$-competitive algorithm (Corollary \ref{cor:alg-guarantee-matroidintersection}) for this setting.

In the agent-arrival model, this framework generalizes the case where $\cF_A$ is the free matroid and $\cF_B$ is an arbitrary matroid, and is therefore at least as hard as the classical matroid secretary problem under a uniformly random order. Nevertheless, we will see that if the matroid constraint is simple enough (see Section \ref{sec:agent-arrival}) we can achieve a constant competitive algorithm for this setting too.

\subsection{Additional definitions}

\paragraph{Competitive ratio.}\; Let $(S,\cI)$ be an independence system and let $w:S\to\RR_+$ be injective, that is, all weights are distinct. Fix a maximum-weight basis $\OPT=\OPT(S)\in\cI$; note that for matroids with injective weights, this basis is unique, so the choice for that case is canonical.
We say an algorithm for any of the considered variants of the secretary problem is $\alpha$-\emph{probability-competitive} if, for this fixed $\OPT$,
every $e\in\OPT$ is included in the algorithm’s output $\ALG\in\cI$ with probability
at least $\alpha$. We say the algorithm is $\alpha$-\emph{utility-competitive} if the output satisfies $\E[w(\ALG)]\ge \alpha\,w(\OPT)$.
Note that
$\alpha$-probability-competitiveness implies $\alpha$-utility-competitiveness, since
\[
\E[w(\ALG\cap\OPT)]=\sum_{e\in\OPT} w(e)\Pr[e\in\ALG]\ge\alpha \cdot w(\OPT).
\]

\paragraph{Combinations of independence systems.} \; In most of our results we work with independence systems obtained from simpler ones via an operation we call \emph{combination}, which strictly extends the usual \emph{intersection}. Let $\{(S_j,\mathcal{I}_j)\}_{j=1}^M$ be independence systems. Their \emph{combination} is the system $(S,\mathcal{I})$ where
\[
S=\bigcup_{j=1}^M S_j
\qquad\text{and}\qquad
\mathcal{I}=\{\,I\subseteq S:\ \forall j\in[M],\ I\cap S_j\in\mathcal{I}_j\,\}.
\]
Note that if all $S_j$ coincide, the combination reduces to the intersection $\bigcap_{j=1}^M \mathcal{I}_j$ on that common ground set.

\paragraph{Example: Matchings and $k$-matchoids.}
For a graph $G=(V,E)$, the family of matchings $\mathcal{M}=(E,\mathcal{I})$ with
\(
\mathcal{I}=\{\,F\subseteq E:\ \forall v\in V,\ |\delta(v)\cap F|\le 1\,\}
\)
is exactly the combination of the rank-1 uniform matroids at each vertex, $\mathcal{M}_v=(\delta(v),\mathcal{I}_v)$ with
\(
\mathcal{I}_v=\{\,F_v\subseteq \delta(v):\ |F_v|\le 1\,\}.
\)
Each edge of $E$ belongs to exactly two components (those of its endpoints).
More generally, $k$-matchoids are exactly combinations of matroids in which each element belongs to at most $k$ components (see Section~\ref{sec:higher-order-systems}); these strictly generalize $k$-matroid intersections.

\paragraph{Example: Generalized Assignment Problem (GAP).}
Let $A$ be agents, $B$ jobs, and $E=A\times B$, where $(a,b)$ means assigning job $b$ to agent $a$. Each agent $a$ has time budget $T_a\ge 0$, and processing $b$ on $a$ consumes $t_a(b)\ge 0$. For each $a\in A$, define a knapsack system on $\delta(a)=\{(a,b):b\in B\}$ by $\mathcal{I}_a=\{F\subseteq \delta(a):\sum_{(a,b)\in F} t_a(b)\le T_a\}$. For each $b\in B$, define a rank-1 uniform matroid on $\delta(b)=\{(a,b):a\in A\}$ by $\mathcal{I}_b=\{F\subseteq \delta(b):|F|\le 1\}$. The feasible assignments for the GAP problem are exactly those in the \emph{combination} of the per-agent knapsacks $(\delta(a),\mathcal{I}_a)$ and the per-job rank-1 uniform matroids $(\delta(b),\mathcal{I}_b)$.

\section{Free-Order Secretary Problem on \texorpdfstring{$k$}{k}-Matroid Intersection} \label{sec:k-matroid-intersection}

In this section, we present our algorithm for $k$-matroid intersection. Before that, we briefly recall the $1/4$-competitive free-order matroid secretary algorithm of Jaillet, Soto, and Zenklusen~\cite{free-order-secretary, jaillet-soto-zenklusen-arxiv}.

\subsection{Warm-up: Free-Order Matroid Secretary}

Let $M = (S,\cI)$ be a matroid given in advance, and let $w\colon S\to\R_+$ be an unknown positive weight function, where, without loss of generality, all elements have distinct weights. For the sake of completeness and to motivate future extensions, we include a proof of the following result that appeared in~\cite{jaillet-soto-zenklusen-arxiv,free-order-secretary}; the algorithm is presented as Algorithm~\ref{alg:free-matroid}, see Figure~\ref{fig:sampling-matroid} for an illustration. 

\begin{algorithm}[t]
\DontPrintSemicolon
\KwData{A matroid $M=(S,\cI)$.}
\KwResult{An independent set $\ALG \in \cI$.}
    $\ALG \gets \emptyset$\;
    \lFor{each $e\in S$}{
        Choose $t_e$ independently and uniformly from $(0,1)$
    }
    $S_1 \gets \set{e \midd t_e \in (0, 1/2)}$ \;
    $S_2 \gets \set{e \midd t_e \in [1/2, 1)}$ \;
    Observe (without accepting) all elements in $S_1$ \\
    
    $\textbf{unseen} \gets S_2$ \tcp*{The set of elements not yet revealed}
    Sort the elements of $S_1$ in decreasing order of weights as $e_1,\dots, e_m$ \;
    \For{$j = 1$ to $m$}{
        Let $Q = \textbf{unseen} \cap \text{span}(\{e_1,\dots, e_j\})$ \;
        \For{each $f \in Q$ in arbitrary order}{
            $\textbf{unseen} \gets \textbf{unseen} - f$ \tcp*{Reveal $f$}
            \lIf{$w(f) > w(e_j)$ and $\ALG + f \in \cI$}{
                $\ALG \gets \ALG + f$
            }
        }
      }

    $Q \gets \textbf{unseen}$ \\
    \For{each $f \in Q$  in arbitrary order}{
        \lIf{$\ALG + f \in \cI$}{
            $\ALG \gets \ALG + f$
        }
    }
    
\Return $\ALG$
\caption{Free-order Matroid Secretary \cite{jaillet-soto-zenklusen-arxiv}}\label{alg:free-matroid}
\end{algorithm}

\begin{theorem}[Jaillet, Soto, Zenklusen \cite{jaillet-soto-zenklusen-arxiv,free-order-secretary}]\label{thm:jsz-free-order-matroid} 
Algorithm~\ref{alg:free-matroid} is $1/4$-probability competitive for the free-order matroid secretary problem.
\end{theorem}
\begin{proof}
Let $\OPT$ be the unique maximum weight basis of $S$ and fix $f \in \OPT$. Let $h_1$ be the heaviest element in $S_1 - f$ such that $f$ is spanned by the elements of higher weight than $h_1$ in $S_1 - f$. In other words,
\[
h_1 \coloneqq \argmax \set{w(h) \midd h \in S_1 - f\ \text{such that}\ f \in \text{span}(\set{g \in S_1 - f \midd w(g) > w(h)}}.
\]
Similarly, let $h_2$ be the heaviest element in $S_2 - f$ such that $f$ is spanned by the elements of higher weight than $h_2$ in $S_2 - f$. In other words,
\[
h_2 \coloneqq \argmax \set{w(h) \midd h \in S_2 - f\ \text{such that}\ f \in \text{span}(\set{g \in S_2 - f \midd w(g) > w(h)}}.
\]
If there are no such elements, we set $h_1 = \bot$ and $h_2 = \bot$, respectively, with $w(\bot) = 0$. 

Since $S_1 - f$ and $S_2 - f$ are identically distributed, we have that $\Pr[w(h_1) \geq w(h_2)] = 1/2$, regardless of whether one or both of $h_1$ and $h_2$ are $\bot$. Next, we condition on the independent event that $f \in S_2$, and conclude that both $w(h_1) \geq w(h_2)$ and $f \in S_2$ occur together with probability at least $1/4$. In what follows, we condition on these two events occurring. Since $f \in \OPT$, we know that $f \not\in \text{span}(\set{g \in S \midd w(g) > w(f)}$). Thus, $w(f) > w(h_1) \geq w(h_2)$.

Recall from the description of the algorithm that $S_1$ is ordered in decreasing order of weights as $e_1,\dots, e_m$. Assume that $h_1$ is the $j^*$-th highest weight element of $S_1$, that is, $h_1 = e_{j^*}$, where we set $j^* = m+1$ if $h_1 = \bot$. By construction, $f$ is revealed by the algorithm precisely in iteration $j^*$, when it first enters the set $Q$. Therefore, at the moment when $f$ is revealed, $f$ satisfies both $w(f) > w(e_{j^*}) = w(h_1)$ and $\ALG + f \in \cI$, and thus $f$ is included in $\ALG$. We thus conclude that for every $f\in \OPT$, $\Pr[f \in \ALG] \geq 1/4$.
\end{proof}

Theorem~\ref{thm:jsz-free-order-matroid} shows that, for a single matroid, it suffices to use half of the elements as a sample to determine both an order in which to reveal the remaining elements and thresholds $w(e_j)$ for those called in the $j$-th group, in order to obtain a constant-factor competitive algorithm. The algorithm then processes the elements in this order and greedily selects those whose value exceeds their threshold. Interestingly, the order in which elements are revealed depends on the \emph{span} of prefixes of the sampled set, that is, $\set{e_1,\dots,e_j}$, when the elements are sorted by weight. Our generalization to $k$-matroid intersection systems builds directly on this observation.

\subsection{Obstacles and Workarounds}

Next, we consider the $k$-matroid intersection setting. Let $M_i = (S,\cI_i)$ be a matroid for $i\in[k]$, and let $M = (S,\cI)$, with $\cI = \cap_{i=1}^k \cI_i$, be their intersection. When we try to generalize Algorithm~\ref{alg:free-matroid} to the intersection of $k$ matroids, two main obstacles arise. 
First, the spans of the prefixes $\set{e_1, \dots, e_j}$ across the matroids can be completely unrelated. In fact, the best order to reveal the non-sampled elements for one matroid may even be the reverse of the best order for another. Second, even if we had a favorable order, there is no hope for a constant-factor probability-competitive algorithm. This is because matroids satisfy a minor-optimality condition that does not extend to matroid intersection. Specifically, for a single matroid, if $f \in \OPT(S)$, then $f \in \OPT(X)$ for any subset $X \subseteq S$ containing $f$. For matroid intersection, however, this monotonicity condition can fail; $\OPT(S)$ may contain $f$, but removing a single element $e \in S-f$ can lead to $\OPT(S-e)$ omitting $f$ entirely. Thus, we have to rely on utility-competitiveness instead.

\paragraph{Union-Span and Critical Elements.}

To circumvent the first difficulty, a natural approach is to use the \emph{union-span}: for each prefix $\set{e_1, \dots, e_j}$, we ``call'' all unseen elements that are in the union of all individual matroid span functions. Notice that the elements that are in the intersection of the optimum matroid bases of all individual matroids are exactly those elements not union-spanned by elements in $S$ of higher weight; we call these the \emph{critical} elements.

A natural first idea is to use Algorithm~\ref{alg:free-matroid} to select each critical element with probability $1/4$. However, even for simple matroid intersection systems, such as bipartite matchings, the set of critical elements may have insufficient weight relative to the optimum independent set. To see this, consider the following bipartite matching system over a simple path $e_1,  \dots, e_\ell$ with weights $1 + \eps > w(e_1)  > \dots > w(e_\ell) = 1$. Note that, for each $i \geq 2$, every element $e_i$ is \emph{union-spanned} by $e_{i-1}$, so $e_1$ is the only critical element. Therefore, even though the weight of a maximum matching is roughly $\ell / 2$, the only critical element has weight $1$. This shows that, in general, the set of critical elements may have arbitrary low weight compared to the optimum basis. Given this, one might expect that the issue could be resolved by considering the set of critical elements with respect to a random subset of $S$ instead of $S$ itself. This approach fails as well, however: for instance, there may be many parallel elements with the same weight, or elements that are loops in some matroids but not in others.

\paragraph{Greedy-Relevant Elements.}

To address the second issue mentioned above, we use a different set, first defined by Feldman, Svensson and Zenklusen~\cite{fsz-secretary-framework-matroid-intersection}.

\begin{definition}[Greedy-Relevant Elements]\label{def:greedy-relevant-elements}
Let $Y \subseteq S$. We say that an element $y \in S \setminus Y$ is \emph{greedy-relevant with respect to $Y$} if $y$ belongs to the output of the weighted-greedy algorithm on $Y + y$. We use $\Rel(Y)$ to denote the set of greedy-relevant elements of $Y$. In other words,
\[
\Rel(Y) = \set{y \in S \setminus Y \midd y \in \Greedy(Y+y)}.
\]
\end{definition}
Feldman, Svensson and Zenklusen \cite{fsz-secretary-framework-matroid-intersection} showed that if every element of $S$ is contained in $Y$ with probability $p$ for some large enough $p > 0$, then the set of critical elements of $\Rel(Y)$ has high enough weight compared to the weight of the set $\Greedy(Y)$ that the weighted-greedy algorithm selects on $Y$. More precisely, they showed the following lemma.

\begin{lemma}[Core Lemma for Matroid Intersection]\label{lem:fsz-original-lemma}
Let $(S,\cI)$ be a $k$-matroid intersection system and $Y = Y(p)$ be a random set that contains each element of $S$ independently with probability $p$ for some $p\in[0,1]$. Then 
\begin{align*}
\E\brk{w\prn{\bigcap_{i=1}^k(\OPT_i(\Rel(Y))}} &\geq \E\brk{\frac{(1-(1-p)k))(1-p)}{p}w(\Greedy(Y))}\\
&\geq \frac{(1-(1-p)k)(1-p)}{k}w(\OPT(S)),
\end{align*}
where $\OPT_i$ represents the optimum weight basis in the $i$-th matroid, and $\OPT$ is the optimum weight basis of the $k$-matroid intersection system. Choosing $p = 1 - \frac{1}{2k}$ yields
\[
\E\brk{w\prn{\bigcap_{i=1}^k(\OPT_i(\Rel(Y))}} \geq \frac{1}{4k^2} \: w(\OPT(S)).
\]
\end{lemma}

\paragraph{Hull and Core.}
Given Lemma~\ref{lem:fsz-original-lemma}, we are almost ready to describe our algorithm for $k$-matroid intersection. The main difference from the algorithm for a single matroid is that we first sample a relatively large set of elements $Y$ in order to identify a good subset $\Rel(Y)$ of greedy-relevant elements, and then execute the free-order algorithm on $\Rel(Y)$. We note that this algorithm can be applied in a black-box manner to more general independence systems, assuming that one has appropriate notions of \emph{critical} elements and what it means for an element to be \emph{spanned} by a set. We will make use of this observation in later sections; for this reason, the algorithm is presented in its general form.

To describe our algorithm in its general form, we need to generalize the notion of \emph{span} from matroids and also describe a set of \emph{core} elements that will play the role of $\OPT$ for more general independence systems.

\begin{definition}[Primitive Hull]\label{def:prim-hull}
Let $(S, \cI)$ be an independence system. We define the \emph{primitive hull} function $\PHull \colon 2^S\to 2^S$ by
\[
\PHull(X) \coloneqq X \cup \set{e \in S \midd e \text{ is in a circuit of } X+e}.
\]
\end{definition}

Notice that, for a single matroid constraint, $\PHull$ is equivalent to the matroid $\text{span}$ function. In fact, $\PHull$ satisfies the following useful properties.
\begin{enumerate}[label=(H\arabic*)]\itemsep0em
    \item (Inclusion) For all $X \subseteq S, \: \: X \subseteq \PHull(X)$. \label{Hinclusion}
    \item (Monotonicity) For all $X, Y \subseteq S$ with $X\subseteq Y$, we have $\PHull(X) \subseteq \PHull(Y)$. \label{HMonotonicity}
    \item (Basis Extension) For all $X \subseteq S$ and $e \in S \setminus X$, if $e \notin \PHull(X)$ then for any $I \subseteq X$ such that $I \in \cI$, we have $I+e \in \cI$. In particular, if $X \in \cI$, then $X+e \in \cI$. \label{Hextension}
\end{enumerate}
On the other hand, $\PHull$ does not necessarily satisfy some properties of the matroid $\text{span}$ function, such as $\text{span}(\text{span}(X)) = \text{span}(X)$.

For combinations of systems, however, we use the union of the primitive hulls of each system. We call this the \emph{hull}.
\begin{definition}[Hull]\label{def:hull}
Let $\calF = (S,\cI)$ be a combination of a given set of systems $\set{(S_j,\cI_j)}_{j = 1}^M$. We define the \emph{hull} function $\Hull \colon 2^S\to 2^S$ of the pair $\calF, \set{(S_j,\cI_j)}_{j = 1}^M$ by
\[
\Hull(X) \coloneqq \bigcup_{j = 1}^M \PHull_j(X \cap S_j),
\]
where $\PHull_j$ is the primitive hull of $(S_j, \cI_j)$.
\end{definition}
It is worth mentioning that, if $\cF$ is a combination of two or more systems, $\PHull \subseteq \Hull$ but the reverse is not true, even for the intersection of two matroids. To see this, suppose that $\cF$ is the intersection of two matroids $M_1$ and $M_2$ and consider three elements $x,y$ and $a$ such that $\set{x,y,a}$ is a circuit in $M_1$ and $\set{x,y}$ is a circuit in $M_2$. Let $X= \set{x,y}$. Then, $a \in \PHull_{M_1}(X)$, which implies that $a \in \Hull_{\cF}(X)$, but $a \notin \PHull_{\cF}(X)$ since $X = \set{x,y}$ is itself a circuit in $\cF$ and thus adding $a$ to $X$ does not create any new circuits.
However, notice that, for $k$-matroid intersection, we have $\PHull_j = \text{span}_j$ for every component system $(S_j, \cI_j)$, and thus $\Hull$ is equivalent to the union-span.

We notice that the function $\Hull$ just defined satisfies all \ref{Hinclusion}, \ref{HMonotonicity} and \ref{Hextension}. 
The inclusion and monotonicity properties of $\Hull$ follow directly from the corresponding properties for the individual primitive hulls. To see the basis extension property, observe that, for $X \subseteq S$ and $e \in S \setminus X$, we have
\begin{align*}
e \notin \Hull(X) &\iff \forall j, \: \: e \notin \PHull_j(X \cap S_j) \\
&\iff \forall j, \: \: \forall I_j \subseteq X \cap S_j \text{ s.t. } I_j \in \cI_j, \: \: I_j + e  \in \cI_j \\
&\iff \forall I \subseteq X \text{ s.t. } I \in \cI, \: \: I + e  \in \cI
\end{align*}

Next, we define the set of \emph{critical} elements for larger families of systems beyond matroids. We call this the set of \emph{core} elements.
\begin{definition}[Primitive Core]\label{def:prim-core}
Let $(S, \cI)$ denote an independence system and let $w \colon S \to \R_+$ denote an arbitrary injective weight function. For any set $X \subseteq S$ and element $e$, let $A_e(X) = \set{f \in X \midd w(f) > w(e)}$ denote the set of elements with weight higher than $X$. We define the set of \emph{primitive core} elements as
\[
\PCore(X) = \set{e \in X \midd e \not\in \PHull\prn{A_e(X)}}.
\]
\end{definition}
In other words, an element $x \in X$ that is in the primitive core of $X$ cannot be the smallest element of any circuit in $X$. Equivalently, $x$ is in the primitive core of $X$ if, for every independent set $I \in \cI$ with $I \subseteq \set{y \in X \midd w(y) > w(x)}$, we have $I + x \in \cI$. Note that the last formulation implies that $x$ belongs to every basis of $\set{y \in X \midd w(y) \geq w(x)}$.
Recall that $\Greedy(X)$ denotes the set of elements selected by the weighted-greedy algorithm on $X$.
\begin{lemma}\label{lem:prim-inclusion}
For all $X\subseteq Y \subseteq S$, we have $X \cap \PCore(Y) \subseteq \PCore(X) \subseteq \Greedy(X).$
\end{lemma}
\begin{proof}
Let $e \in X \cap \PCore(Y)$. For any circuit $C$ with $C \subseteq Y$, $e$ cannot be the smallest element of $C$. In particular, the same holds for every circuit contained in $X$, and so $e \in \PCore(X)$.
Next, let $e \in \PCore(X)$, and suppose we run the weighted-greedy algorithm on $X$. Let $H$ be the set constructed by greedy just before $e$ is considered. Note that $H$ is an independent set contained in $\set{y \in X \midd w(y) > w(x)}$. Since $e \in \PCore(X)$, we have that $H + e \in \cI$, and so $e$ is added to the set constructed by the weighted-greedy algorithm.
\end{proof}

One consequence of Lemma~\ref{lem:prim-inclusion} is that $\PCore(X)$ is always independent for any $X$. However, if $\cF$ is a matroid, we have a stronger property.

\begin{lemma}\label{lem:prim-core-matroids}
Let $\cM = (S,\cI)$ be a matroid. For any $X \subseteq S$, we have
\[
\PCore(X) = \Greedy(X) = \OPT(X),
\]
where $\OPT(X)$ denotes the unique maximum weight basis of $X$.
\end{lemma}
\begin{proof}
The first equality follows since, for a matroid with distinct weights, $e \in \Greedy(X)$ if and only if $e$ is not spanned by elements of higher weight, which exactly corresponds to the definition of $e \in \PCore(X)$. The second equality holds since, for matroids, the greedy algorithm on $X$ returns the optimum basis of $X$.
\end{proof}

This equivalence fails to hold for more general independence systems. In fact, even for the intersection of two matroids, the sets $\PCore(X)$, $\Greedy(X)$, and $\OPT(X)$ may be all different.
For combinations of systems, we use the intersection of the primitive cores of each system. We call this set the \emph{core} and its elements the \emph{critical} elements.
\begin{definition}[Core and Critical Elements]\label{def:core}
Let $\calF = (S,\cI)$ be a combination of systems $\set{(S_j,\cI_j)}_{j = 1}^M$ and $w \colon S \to \R_+$ be an arbitrary weight function. We define the \emph{core} of a set $X \subseteq S$ as
\[
\Core(X) \coloneqq \bigcap_{j = 1}^M \prn{(X \setminus S_j) \cup \PCore_j(X \cap S_j)},
\]
where $\PCore_j$ is the primitive core of $(S_j, \cI_j)$. We call the elements in $\Core(X)$ the \emph{critical} elements of $X$.
\end{definition}
Notice that, for the case of $k$-matroid intersection, we have
\begin{align*}
e \in \Core(X) &\iff \forall j, \: \: e \in \PCore_j(X) \iff \forall j, \: \: e \notin \PHull_j(\set{f \in X \midd w(f) > w(e)}) \\
&\iff \forall j, \: \: e \notin \text{span}_j(\set{f \in X \midd w(f) > w(e)}) \\
&\iff \forall j, \: \: e \in \OPT_j(X).
\end{align*}
In other words, for the intersection of $k$ matroids, $\Core(X)$ is exactly $\bigcap_{j = 1}^k \OPT_j(X)$.
It is easy to see that, for a combination of systems, the relationship between the core and the hull from the component systems extends to the entire combination.
\begin{observation}\label{obs:core-hull-combination}
Let $\calF = (S,\cI)$ be a combination of systems $\set{(S_j,\cI_j)}_{j = 1}^M$ and $w \colon S \to \R_+$ be an arbitrary weight function. Then, for all $X \subseteq S$ and $e \in X$,
\[
e \in \Core(X) \iff e \notin \Hull(\set{f \in X \midd w(f) > w(e)}).
\]
\end{observation}
\begin{proof}
Let $\PHull_j$ and $\PCore_j$ denote the primitive hull and primitive core of $(S_j, \cI_j)$, respectively. We have
\begin{align*}
e \in \Core(X) &\iff \forall j, \: \: (e \notin S_j) \vee (e \in \PCore_j(X \cap S_j)) \\
&\iff \forall j, \: \: (e \notin S_j) \vee (e \notin \PHull_j(\set{f \in X \cap S_j \midd w(f) > w(e)})) \\
&\iff e \notin \Hull(\set{f \in X \midd w(f) > w(e)}).\qedhere
\end{align*}
\end{proof}

\paragraph{Algorithmic Aspects.}

Our black-box algorithm depends on the ability to solve the following problem for the given feasibility constraint: given a set $X$ and an element $e$, is $e \in \Hull(X)$? Since $\Hull$ is defined as the union of $\PHull_j$ over all system components $\cF_j$, this problem reduces to testing whether $e \in \PHull_j(X)$ for some set $X$ and $e$.

Fix a system component $\cF_j$. Whenever $\cF_j$ is a matroid, computing $\PHull_j(X) = \text{span}_j(X)$ can be done in polynomial time. For the case of knapsack constraints, we can compute $\PHull_j(X)$ by solving a knapsack subproblem, which can be done in pseudopolynomial time. However, for general systems, computing $\Hull(X)$ can be arbitrarily hard; we do not pursue this direction here, as it lies outside the scope of the present work.

\subsection{Free-Order Secretary Algorithm for \texorpdfstring{$k$}{k}-Matroid Intersection}

We now have all the necessary ingredients to design a general free-order algorithm that applies to any system, 
presented as Algorithm~\ref{alg:free-general}; see Figure~\ref{fig:sampling-growth} for an illustration. The following theorem relates the expected weight of the independent set returned by our algorithm to that of $\Core(\Rel(Y))$, where $Y$ is a random subset of $S$.

\begin{algorithm}[ht]
\DontPrintSemicolon
\KwData{An independence system $(S, \cI)$, a parameter $p$.}
\KwResult{An independent set $\ALG$.}
    $\ALG \gets \emptyset$ and $q \gets (1+p) / 2$\;
    \lFor{each $e\in S$} {
        Choose $t_e$ independently and uniformly from $(0,1)$
    }
    $Y \gets \set{e \midd t_e \in (0, p)}$ \;
    $S_1 \gets \set{e \midd t_e \in [p, (1+p)/2)}$ \;
    $S_2 \gets \set{e \midd t_e \in [(1+p)/2, 1)}$ \;
    Observe (without accepting) all elements in $Y \cup S_1$

    $\textbf{unseen} \gets S_2$ \tcp*{The set of elements not yet revealed}
    Sort the elements of $\Rel(Y) \cap S_1$ in decreasing order of weights as $e_1,\dots, e_m$ \;
    \For{$j = 1$ to $m$}{
        Let $Q = \textbf{unseen} \cap \Hull(\set{e_1,\dots, e_j})$ \;
        \For{each $f \in Q$ in arbitrary order}{
            $\textbf{unseen} \gets \textbf{unseen} - f$ \tcp*{Reveal $f$}
            \If{$f \in \Rel(Y)$ and $w(f) > w(e_j)$ and $\ALG + f \in \cI$\label{alg2:line13}}{
                $\ALG \gets \ALG + f$
            }
        }
      }

    $Q \gets \textbf{unseen}$ \\
    \For{each $f \in Q$ in arbitrary order}{
        \lIf{$f \in \Rel(Y)$ and $\ALG + f \in \cI$ \label{alg2:line20}}{
            $\ALG \gets \ALG + f$
        }
    }
    \Return $\ALG$
\caption{Free-Order General Secretary}\label{alg:free-general}
\end{algorithm}

\begin{theorem}\label{thm:alg-general}
Let $(S,\cI)$ be an independence system.
Algorithm~\ref{alg:free-general} returns an independent set $\ALG$ such that, for every element $f$
\[
\Pr\brk{f \in \ALG \midd f \in \Core(\Rel(Y))} \geq \frac{1}{4}.
\]
\end{theorem}
\begin{proof}
First, note that the algorithm always returns an independent set, since before adding any element $e$ to $\ALG$, it explicitly checks whether $\ALG + e \in \cI$. We now condition on the set $Y$. Observe that $\Rel(Y) \subseteq S_1 \cup S_2$, and let $B \coloneqq \Core(\Rel(Y))$. We next show that each element of $B$ is added to $\ALG$ with probability $1/4$. The proof of this fact follows the same argument as in Algorithm~\ref{alg:free-matroid}, but we include it here for completeness.

Fix $f \in B$ and let $h_1$ be the heaviest element in $(S_1 - f) \cap \Rel(Y)$ such that $f$ is in the hull of the elements of higher weight than $h_1$ in $(S_1 - f) \cap \Rel(Y)$. In other words,
\[
h_1 \coloneqq \argmax \set{w(h) \midd h \in (S_1 - f) \cap \Rel(Y) \text{ s.t. } f \in \Hull(\set{g \in (S_1 - f) \cap \Rel(Y) \midd w(g) \geq w(h)})}.
\]
Similarly, let $h_2$ be the heaviest element in $(S_2 - f) \cap \Rel(Y)$ such that $f$ is spanned by the elements of higher weight than $h_2$ in $(S_2 - f) \cap \Rel(Y)$. In other words,
\[
h_2 \coloneqq \argmax \set{w(h) \midd h \in (S_2 - f) \cap \Rel(Y) \text{ s.t. } f \in \Hull(\set{g \in (S_2 - f) \cap \Rel(Y) \midd w(g) \geq w(h)})}.
\]
If there are no such elements, we set $h_1 = \bot$ and $h_2 = \bot$, respectively, with $w(\bot) = 0$.

Since, even conditioned on $Y$, the sets $(S_1 - f) \cap \Rel(Y)$ and $(S_2 - f) \cap \Rel(Y)$ are identically distributed, we have that $\Pr[w(h_1) \geq w(h_2)] = 1/2$, regardless of whether one or both of $h_1$ and $h_2$ are $\bot$. Next, we condition on the independent event that $f \in S_2$, and conclude that both $w(h_1) \geq w(h_2)$ and $f \in S_2$ occur together with probability at least $1/4$. In what follows, we condition on these two events occurring.

We now verify that $f$ satisfies all the conditions for inclusion in $\ALG$ in lines~\ref{alg2:line13} and~\ref{alg2:line20} of Algorithm~\ref{alg:free-general}. Since $f \in \Rel(Y)$, the first condition is satisfied.
Next, recall that, $S_1 \cap \Rel(Y)$ is ordered in decreasing order of weights as $e_1, \dots, e_m$. Assume that $h_1$ is the $j^*$-th highest weight element of $S_1 \cap \Rel(Y)$, that is, $h_1 = e_{j^*}$,  where $j^* = m+1$ if $h_1 = \bot$. By construction, $f$ is revealed by the algorithm precisely in iteration $j^*$, when it first enters the set $Q$. Furthermore, since $f \in \Core(\Rel(Y))$, by the monotonicity of $\Hull$ we conclude that
\[
f \notin \Hull(\set{g \in \Rel(Y) \midd w(g) > w(f)}) \supseteq \Hull(\set{g\in \Rel(Y)\cap S_1 \midd w(g) > w(f)}).
\]
Thus, $w(f) > w(h_1) \geq w(h_2)$, so we satisfy the second condition of line \ref{alg2:line13}.

Let $\ALG'$ denote the solution immediately before $f$ is revealed, and note that all elements in $\ALG'$ are in $\Rel(Y) \cap S_2$ and have weights larger than $w(e_{j^*})$. There are two possibilities. First, if $h_2 \neq \bot$, since $w(e_{j^*}) = w(h_1) \geq w(h_2)$ by the definition of $h_1$ and $h_2$, we have
\[
f \notin \Hull(\{ g \in (S_2 - f)\cap \Rel(Y) \colon w(g) \geq w(e_{j^*}) \}) \supseteq \Hull(\ALG').
\]
Here, the last containment follows by the monotonicity of $\Hull$, since all elements considered by $\ALG'$ for addition had weights larger than $w(e_{j^*})$. Second, if $h_2 = \bot$ then, again by the monotonicity of $\Hull$, we have
\[
f \notin \Hull((S_2-f)\cap \Rel(Y)) \supseteq \Hull(\ALG').
\]
It follows that at the moment when $f$ is revealed, $f \notin \Hull(\ALG')$, and thus $\ALG'+f\in \cI$ by the properties of $\Hull$. We conclude that all conditions hold, and therefore $f \in \ALG$. We have shown that for every realization of $Y$ and every $f \in B = \Core(\Rel(Y))$, $\Pr\brk{f \in \ALG \midd Y} \geq \frac{1}{4}$. Since this holds for every realization of $Y$, by unconditioning, we get that $\Pr\brk{f \in \ALG} \geq \frac{1}{4}$ as stated.
\end{proof}

We now combine~Lemma \ref{lem:fsz-original-lemma} and Theorem~\ref{thm:alg-general} to get the following corollary.

\begin{corollary}\label{cor:alg-guarantee-matroidintersection}
Let $(S,\cI)$ be a $k$-matroid-intersection system. For $p = 1-1/2k$,
Algorithm~\ref{alg:free-general} returns a set $\ALG$ such that
\[
\E[w(\ALG)] \geq \frac{1}{16k^2} w(\OPT).
\]
\end{corollary}

\paragraph{Next Steps.}
In the following sections, we extend both Lemma~\ref{lem:fsz-original-lemma} and Theorem~\ref{thm:alg-general} in two different directions. First, we define and explore a new class of independence systems, called $k$-growth systems, which includes $k$-matchoids as well as knapsack constraints with bounded size ratios; this new class is already of independent combinatorial interest on its own. We then generalize the Core Lemma (Lemma~\ref{lem:fsz-original-lemma}), first to $k$-matchoids and then to general $k$-growth systems. Using the notion of $\Hull$ for both systems, we recover $\bigOm{1 / k^2}$-utility competitive guarantees for both, in the free-order setting.

Second, we turn to the agent-arrival bipartite graph model. All the results above apply directly to the edge-arrival case, provided the edge independence system belongs to one of the studied classes. However, these results do not naturally extend to the agent-arrival setting, where multiple edges are revealed simultaneously. We overcome this difficulty by establishing a new version of the Core Lemma that holds for a random subset $Y$ of agents rather than edges. Equipped with this tool, we design a new algorithm for the agent-arrival model in settings where the independence system on the \emph{item side} admits an order-oblivious algorithm.
\section{Free-Order Secretary on Combinations of Growth Systems}\label{sec:lemma-extension}

In this section, we extend Lemma~\ref{lem:fsz-original-lemma} to a new and broader class of independence systems that contains $k$-matchoids and knapsack constraints with bounded size ratios. 

\subsection{Growth Systems and Combinations}

The goal of this section is to introduce a novel class of independence systems with interesting structural properties. We call this the class of {\it $k$-growth systems}, and present the definition below; a more detailed discussion of its properties and its relationships to other classes of independence systems is provided in Section~\ref{sec:higher-order-systems}.

\begin{restatable}[$k$-growth system]{dfn}{kGrowthDefn}\label{def:k-growth-system}
Let $\cF = (S,\cI)$ be an independence system and $k$ be a positive integer. We say that $\cF$ is a \emph{$k$-growth system} if it satisfies the following $k$-Basis-Growth ($k$BG) axiom.
\begin{symenum}
\itemsep0em
    \itemsymbol{$k$BG}\label{def:kbg}\linkdest{kbg}{} $\forall X \subseteq S$, $\forall I \in \cI$, there exists a partition $(Q,Z)$ of $I\setminus X$ with $|Z| \leq k\,|X \setminus I|$ such that \\[4pt]
    $(\star1)$ for every basis $P$ of $X$, $P \cup Q \in \cI$.
\end{symenum}
\end{restatable}

Many independence systems that arise in applications can be viewed as combinations of simpler component systems. For example, $k$-matchoids are formed as the intersection of $n$ matroids -- where $n$ may be much larger than $k$ -- under the condition that each element is not a co-loop in at most $k$ of them. As it turns out, $k$-growth systems enjoy a similar combination property that will be useful for our analysis: the combination of $k_i$-growth systems is itself a growth system, with a parameter $k$ determined by the $k_i$'s of the individual components. More importantly, the union of the individual Hull functions of each component is the Hull function for the combined system. We exploit this latter property to obtain an improved result for combinations of systems whose components are matroids. The proof of the following lemma can be found in Section~\ref{sec:higher-order-systems}.

\begin{restatable}{lem}{kGrowthComb}\label{lem:k-growth-combinations}
Let $\calF = (S,\cI)$ be a combination of systems $\set{(S_j,\cI_j)}_{j = 1}^M$, where each $(S_j,\cI_j)$ is a $k_j$-growth system. Then, $\calF$ is a $k$-growth system, where
\[
k = \max_{e\in S} \sum_{ j \colon e \in S_j} k_j.
\]
In particular, for each $X \subseteq S$ and $I \in \cI$, there exists a collection of sets $Z_j \subseteq (I \setminus X) \cap S_j$ such that, for every basis $P$ of $X$, we have $P \cup \bigl((I \setminus X) \setminus \bigcup_{j=1}^{M}Z_j\bigr) \in \cI$, and, for all $j \in [M]$,
\[
|Z_j| \leq \begin{cases}
r_j((X\setminus I) \cap S_j) - |(X\cap I)\cap S_j| & \text{if } k_j = 1, \\
k_j |(X\setminus I) \cap S_j| & \text{if } k_j > 1,
\end{cases}
\]
where $r_j$ denotes the rank function of the $j$-th component $(S_j, \cI_j)$. In general, $|Z_j|\le k_j|(X\setminus I)\cap S_j|$ for $j\in[M]$, since $r_j$ is a matroid rank function when $k_j=1$.
\end{restatable}

\subsection{The Core Lemma for Combinations of Growth Systems}

We are now ready to describe the main result of this section: our generalization of the Core Lemma of \cite{fsz-secretary-framework-matroid-intersection} from intersections of matroids to combinations of growth systems.

\begin{lemma}[Core Lemma for Growth Systems]\label{lem:core-k-growth-general}
Let $\calF = (S,\cI)$ be a combination of systems $\set{(S_j,\cI_j)}_{j = 1}^M$ where each $(S_j,\cI_j)$ is a $k_j$-growth system, and let $\Core$ be the Core of $\cF$. Also, let $k = \max_{e \in S} \sum_{j \colon e \in S_j} k_j$, and $Y = Y(p)$ be a random set where every element of $S$ is in $Y$ independently with probability $p$ for some $p\in[0,1]$. Then,
\[
\E\brk{w\prn{\Core\prn{\Rel(Y)}}} \geq \frac{(p - (1 - p) k) (1 - p)}{p^2} \: \; \E\brk{w\prn{\Greedy(Y)}}.
\]
Furthermore, if all components are matroids, i.e. $\calF$ is a $k$-matchoid, then 
\[
\E\brk{w\prn{\Core\prn{\Rel(Y)}}} \geq \frac{(p - (1 - p) (k-1)) (1 - p)}{p} \: \; \E\brk{w\prn{\Greedy(Y)}}.
\]
\end{lemma}
\begin{proof} 
Recall that, by Lemma~\ref{lem:k-growth-combinations}, $\calF$ is a $k$-growth system, where $k = \max_{e\in S} \sum_{j \colon e \in S_j} k_j$. Let $e_1, \dots, e_m$ be the elements of $S$ sorted in decreasing order of weights. The idea is to split the elements into sets $R$, $N$, $R'$, and $N'$, where $R' \cup N' = \Greedy(Y)$ and $R \cup N = \Rel(Y)$, such that $R' = \Greedy(Y) \cap \Core(\Rel(Y))$, $R = \Core(\Rel(Y))$ and $N' = \Greedy(Y) \setminus R'$, $N = \Rel(Y) \setminus \Core(\Rel(Y))$. Then, the proof of the lemma will follow from the following four inequalities:
\begin{itemize}\itemsep0em
    \item $\E\bigl[|N'|\bigr] \geq \frac{p}{1-p} \: \E[|N|]$, $\hfill \refstepcounter{equation}(\theequation) \label{eq:simulation-1}$
    \item $\E[|R|] + \E[|N|] \geq \frac{1-p}{p} (\E[|R'|] + \E[|N'|])$, $\hfill \refstepcounter{equation}(\theequation) \label{eq:simulation-2}$
    \item $\E[|N'|] \leq k \: (\E[|R|] + \E[|N|])$, $\hfill \refstepcounter{equation}(\theequation) \label{eq:simulation-3}$
    \item $\E[|N'|] \leq k \: \E[|R|] + (k-1) \: \E[|N|]$ if all components of $\calF$ are matroids. $\hfill \refstepcounter{equation}(\theequation) \label{eq:simulation-4}$
\end{itemize}
To construct $R$, $N$, $R'$ and $N'$, we use Algorithm~\ref{alg:simulation} which is an offline simulation algorithm.

\begin{algorithm}[t]
\DontPrintSemicolon
$R, N, R', N' \gets \emptyset$ \\
\For{$i = 1$ to $m$}{
    Toss a coin $c_i$ that is heads with probability $p$ \tcp*{If heads, $e_i$ is in the sampling phase}
    \If{$(R' \cup N') + e_i \in \cI$}{
        \eIf{$e_i \in \Core((R \cup N) + e_i)$}{
            \leIf{$c_i$ is heads}{
                $R' \gets R' + e_i$
            }{
                $R \gets R + e_i$
            }
        }{
            \leIf{$c_i$ is heads}{
                $N' \gets N' + e_i$
            }{
                $N \gets N + e_i$
            }
        }
    }
}
\caption{Simulation algorithm}\label{alg:simulation}
\end{algorithm}

Let $R_i$, $N_i$, $R'_i$ and $N'_i$ be the sets $R$, $N$, $R'$ and $N'$, respectively, at the end of the $i$-th iteration, where we set $R_0 = N_0 = R'_0 = N'_0 = \emptyset$. Note that $Y$ has the same distribution as the set of edges $e_i$ for which $c_i$ is heads, so we let $Y$ be that set. Let $E_i \coloneqq \set{e_1, \dots, e_i}$ and $Y_i \coloneqq Y \cap E_i$. Also note that $R'_i \cup N'_i = \Greedy(Y_i)$ and $R_i \cup N_i = \Rel(Y_i) \cap E_i = \Rel(Y) \cap E_i$, and that $N_i$, $R_i$, $N'_i$ and $R'_i$ are pairwise disjoint. Let $I_i \coloneqq R'_i \cup N'_i$ and $X_i \coloneqq R_i \cup N_i$.

First, notice that, at each iteration, every $e_i$ that satisfies both $(R' \cup N') + e_i \in \cI$ and $e_i \in \Core((R \cup N) + e_i)$ is added to $R'$ with probability $p$ and to $R$ with probability $1-p$. Similarly, every $e_i$ that satisfies both $(R' \cup N') + e_i \in \cI$ and $e_i \notin \Core((R \cup N) + e_i)$ is added to $N'$ with probability $p$ and to $N$ with probability $1-p$. From these, it follows that, for every $i$,
\[
(1-p) \: \E[|R'_i|] = p \: \E[|R_i|] \quad \text{ and } \quad (1-p) \: \E[|N'_i|] = p \: \E[|N_i|],
\]
and thus \eqref{eq:simulation-1} and \eqref{eq:simulation-2} follow directly from the above. Next, we show that \eqref{eq:simulation-3} holds as well. First, we verify that $R$, $N$, $R'$, and $N'$ are the sets as claimed.

\begin{claim}\label{clm:simulation-r-n-prime}
For all $i\in[m]$, $\: R'_i \cup N'_i = \Greedy(Y_i)$.
\end{claim}
\begin{proof}
Observe that, for $\ell \leq i$, $R'_i \cup N'_i$ corresponds to the set of elements $e_\ell$ from $E_i$ which satisfied $(R'_{\ell-1} \cup N'_{\ell-1}) + e_\ell \in \cI$ and for which the coin $c_\ell$ came up heads. In other words, $R'_i \cup N'_i$ is the subset of $Y_i$ of elements $e_\ell$ which satisfied $(R'_{\ell-1} \cup N'_{\ell-1}) + e_\ell \in \cI$, which is exactly $\Greedy(Y_i)$.
\end{proof}

\begin{claim}\label{clm:simulation-r-n}
For all $i\in[m]$, $\: R_i \cup N_i = \Rel(Y) \cap E_i$.
\end{claim}
\begin{proof}
Similarly to the previous claim, for $\ell \leq i$, $R_i \cup N_i$ corresponds to the set of elements $e_\ell$ from $E_i$ which satisfied $(R'_{\ell-1} \cup N'_{\ell-1}) + e_\ell \in \cI$ and for which the coin $c_\ell$ came up tails, and thus, by the definition of greedy-relevant elements, we have $R_i \cup N_i = \Rel(Y) \cap E_i$.
\end{proof}

\begin{claim}\label{clm:simulation-r}
For all $i\in[m]$, $\: R_i = \Core(\Rel(Y)) \cap E_i$.
\end{claim}
\begin{proof}
Observe that, for $\ell \leq i$, $R_i$ corresponds to the set of elements $e_\ell$ from $E_i$ which satisfied $(R'_{\ell-1} \cup N'_{\ell-1}) + e_\ell \in \cI$, $e_\ell \in \Core((R_{\ell-1} \cup N_{\ell-1}) + e_\ell)$ and for which the coin $c_\ell$ came up tails. By Claim~\ref{clm:simulation-r-n}, this is equivalent to $e_\ell \in \Core((\Rel(Y) \cap E_{\ell-1}) + e_\ell)$ and, by Observation~\ref{obs:core-hull-combination}, to $e_\ell \notin \Hull(\Rel(Y) \cap E_{\ell-1})$, since $E_{\ell-1}$ contains exactly the elements of $S$ with weight higher than that of $e_\ell$. This is precisely the condition that $e_\ell \in \Core(\Rel(Y))$. Since $e_\ell \in E_i$, the claim follows.
\end{proof}

\begin{claim}\label{clm:simulation-r-prime}
For all $i\in[m]$, $\: R'_i = \Greedy(Y_i) \cap \Core(\Rel(Y))$.
\end{claim}
\begin{proof}
By Claim~\ref{clm:simulation-r-n-prime}, we have that $R'_i \subseteq \Greedy(Y_i)$. In particular, for $\ell \leq i$, $R'_i$ corresponds to the set of elements $e_\ell$ from $\Greedy(Y_i)$ which also satisfied $e_\ell \in \Core((R_{\ell-1} \cup N_{\ell-1}) + e_\ell)$. By Claim~\ref{clm:simulation-r-n}, this is equivalent to $e_\ell \in \Core((\Rel(Y) \cap E_{\ell-1}) + e_\ell)$ and, by Observation~\ref{obs:core-hull-combination}, to $e_\ell \notin \Hull(\Rel(Y) \cap E_{\ell-1})$, since $E_{\ell-1}$ contains exactly the elements of $S$ with weight higher than that of $e_\ell$. This is precisely the condition that $e_\ell \in \Core(\Rel(Y))$. Thus, the claim follows.
\end{proof}

Next, we use Lemma~\ref{lem:k-growth-combinations} on the disjoint sets $X_i$ and $I_i$ to obtain sets $Z_{i,j} \subseteq I_i$ for each $j \in [M]$ with the bounds given by the lemma, and let $Z_i = \bigcup_{j = 1}^M Z_{i,j}$ and $Q_i = (I_i \setminus X_i) \setminus Z_i = I_i \setminus Z_i$. We show that $N'_i \subseteq Z_i$ for all $i$.

To see this, notice that, Claims~\ref{clm:simulation-r-n-prime} and~\ref{clm:simulation-r-prime} and the definition of $N'$, we get that every element $e_\ell$ with $\ell \leq i$ that enters $N'_i$ satisfies $e_\ell \notin \Core((N_{\ell-1} \cup R_{\ell-1})+e_\ell)$. In particular, there exists an index $j$ such that $e_\ell \in S_j$ and $e_\ell \notin \Core_j((N_{\ell-1}\cup R_{\ell-1})\cap S_j)$, and thus there exists a circuit $C$ of the $j$-th component such that $e_\ell \in C \cap S_j$. Therefore, the set $P_j \coloneqq C - e_\ell \subseteq (N_{\ell-1}\cup R_{\ell-1}) \cap S_j$ satisfies $P_j \in \cI_j$ but $P_j + e_\ell \notin \cI_j$. This implies that $e_\ell \notin Q_{i,j} \coloneqq (I_i \setminus Z_{i,j}) \cap S_j$. Therefore, $e_\ell$ has to belong to $Z_{i,j} \subseteq Z_i$, and thus $N'_i \subseteq Z_i$.

Now, we condition on the coins $c_\ell$ for $\ell \leq i$ and we apply the basis-growth axiom \ref{def:kbg} to the disjoint sets $X_i$ and $I_i\in \cI$, yielding
\[
\abs{N'_i} \leq \abs{\bigcup_{j = 1}^M Z_{i,j}} \leq \sum_{j = 1}^M k_j \abs{X_i \cap S_j} = \sum_{e \in X_i} \sum_{j \colon e \in S_j}  k_j = k \abs{X_i} = k (|R_i|+|N_i|).
\]
By unconditioning, we obtain
\begin{equation}\label{eq:n-prime-upper-bound-1}
\E[|N'_i|] \leq k \: (\E[|R_i|] + \E[|N_i|]),
\end{equation}
and thus \eqref{eq:simulation-3} follows. For the particular case when all components are matroids and so $\calF$ is a $k$-matchoid, we use the following stronger inequality.
\begin{align*}
\abs{N'_i} &\leq \abs{\bigcup_{j = 1}^M Z_{i,j}} \leq \sum_{j = 1}^M r_j(X_i \cap S_j) = \sum_{\ell = 1}^i \sum_{j = 1}^M r_j(X_\ell \cap S_j) - r_j(X_{\ell-1} \cap S_j) \\
&\leq \sum_{\ell = 1}^i \sum_{e_\ell \in R_\ell} k + \sum_{\ell = 1}^i \sum_{e_\ell \in N_\ell} (k-1)= k \: |R_i| + (k-1) \: |N_i|,
\end{align*}
where the last inequality follows from the fact that if $e_\ell \in R_\ell$, then for every matroid $j$ in which $e_\ell \in S_j$, $e_\ell$ does not close any circuits in $X_{\ell}$. Therefore, for all such $\ell$, we have $r_j(X_\ell\cap S_j) - r_j(X_{\ell-1}\cap S_j) = 1$. Since $e_\ell$ participates in at most $k$ matroids, we have $\sum_{j = 1}^M  (r_j(X_\ell \cap S_j) - r_j(X_{\ell-1} \cap S_j)) \leq k$. On the other hand, if $e_\ell \in N_\ell$, then there must be an index $j$ with $e_\ell \in S_j$ such that $e_\ell$ closes a circuit in $X_\ell$. Therefore, the sum above is at most $k-1$. Once again, by unconditioning, we obtain
\begin{equation}\label{eq:n-prime-upper-bound-2}
\E[|N'_i|] \leq k \: \E[|R_i|] + (k-1) \: \E[|N_i|],
\end{equation}
and thus \eqref{eq:simulation-4} follows. 

Putting it all together, we have
\begin{align*}
k \: \E[|R_i|] &\geq \E[|N'_i|] - k \: \E[|N_i|] \\
&\geq \prn{\frac{p}{1-p} - k} \E[|N_i|] \\
&\geq \prn{\frac{p}{1-p} - k} \prn{\frac{1-p}{p} (\E[|R'_i|] + \E[|N'_i|]) - \E[|R_i|]},
\end{align*}
yielding
\begin{align*}
\E[|R_i|] &\geq \prn{\frac{p}{1-p} - k} \frac{(1-p)^2}{p^2} (\E[|R'_i|] + \E[|N'_i|]) \\
 &= \frac{(p - (1-p) k) (1-p)}{p^2} (\E[|R'_i|] + \E[|N'_i|]).
\end{align*}
Thus, by Claims~\ref{clm:simulation-r-n-prime} and~\ref{clm:simulation-r}, we have
\[
\E\brk{\abs{\Core(\Rel(Y)) \cap E_i}} \geq \frac{(p - (1-p) k) (1-p)}{p^2} \E\brk{\abs{\Greedy(Y_i)}}.
\]
Since this is true for every $i$, applying linearity of expectation and the definition of the Core function yields
\[
\E\brk{w\prn{\Core(\Rel(Y))}} \geq \frac{(p - (1 - p) k) (1 - p)}{p^2} \E\brk{w\prn{\Greedy(Y)}}.
\]
For the particular case when all components are matroids, we use \eqref{eq:n-prime-upper-bound-2} instead of \eqref{eq:n-prime-upper-bound-1} to get
\begin{align*}
k \: \E[|R_i|] &\geq \E[|N'_i|] - (k-1) \: \E[|N_i|] \\
&\geq \prn{\frac{p}{1-p} - (k-1)} \E[|N_i|] \\
&\geq \prn{\frac{p}{1-p} - (k-1)} \prn{\frac{1-p}{p} (\E[|R'_i|] + \E[|N'_i|]) - \E[|R_i|]}.
\end{align*}
After rearranging the terms, we get
\begin{align*}
\E[|R_i|] &\geq \prn{p - (1-p)(k-1)} \frac{1-p}{p} (\E[|R'_i|] + \E[|N'_i|]).
\end{align*}
Similarly as before, by Claims~\ref{clm:simulation-r-n-prime} and~\ref{clm:simulation-r}, we have
\[
\E\brk{\abs{\Core(\Rel(Y)) \cap E_i}} \geq \frac{(p - (1-p) (k-1)) (1-p)}{p} \E\brk{\abs{\Greedy(Y_i)}}.
\]
Since this is true for every $i$, applying linearity of expectation and the definition of the Core function yields
\[
\E\brk{w\prn{\Core(\Rel(Y))}} \geq \frac{(p - (1 - p) (k-1)) (1 - p)}{p} \E\brk{w\prn{\Greedy(Y)}}.\qedhere
\]
\end{proof}

\subsection{Free-Order Secretary Algorithm for Edge Arrivals}

Next, we summarize the guarantees of our Core lemma with respect to $w(\OPT)$. Note that our guarantee for $k$-matchoids is exactly the same as the previous guarantee for $k$-matroid intersection, see Lemma \ref{lem:fsz-original-lemma}.

\begin{theorem}\label{thm:core-lemma-vs-opt}
Let $\calF = (S,\cI)$ be a combination of systems $\set{(S_j,\cI_j)}_{j = 1}^M$ where each $(S_j,\cI_j)$ is a $k_j$-growth system. Let $\Core_j(X)$ denote the set of critical elements with respect to $X$ in $(S_j, \cI_j)$. Also, let $k = \max_{e \in S} \sum_{j \colon e \in S_j} k_j$, and $Y = Y(p)$ be a random set that contains each element of $S$ independently with probability $p$ for some $p\in[0,1]$. Then,
for $p =\sqrt{1 - {1}/{(k+1)}}$, we have
\[
\E\brk{w\prn{\Core\prn{\Rel(Y)}}} 
\geq \prn{2 + \frac{1}{k} - 2\sqrt{1 + \frac{1}{k}}} w(\OPT(S)) > \prn{\frac{1}{4k^2} - \frac{1}{8k^3}} w(\OPT(S)).
\]
Furthermore, if all components are matroids, i.e., $\calF$ is a $k$-matchoid, setting $p=0$ for $k=1$ and $p=1-1/(2k)$ for $k\ge2$ yields
\[
\E\brk{w\prn{\Core\prn{\Rel(Y)}}}
\geq \begin{cases}
w(\OPT(S))                      & \text{if } k = 1, \\
\frac{1}{4k^2} \: w(\OPT(S))    & \text{if } k \geq 2. \\
\end{cases}
\]
\end{theorem}
\begin{proof}
Since the system $(S,\cI)$ is a $k$-system the weighted-greedy algorithm yields a $1/k$-approximation with respect to the weight of $\OPT(S)$. Therefore,
\[
\E[w(\Greedy(Y))] \geq \frac{1}{k} \E[w(\OPT(Y))] \geq \frac{1}{k} \E[w(\OPT(S)\cap Y)] = \frac{p}{k} w(\OPT(S)).
\]
Thus, by Lemma~\ref{lem:core-k-growth-general},
\[
\E\brk{w\prn{\Core\prn{\Rel(Y)}}} \geq \frac{(p - (1 - p) k) (1 - p)}{p k} w(\OPT(S)).
\]
Optimizing $p$ yields $p = \sqrt{1 - \frac{1}{k+1}}$ and
\[
\E\brk{w\prn{\Core\prn{\Rel(Y)}}} \geq \prn{2 + \frac{1}{k} - 2\sqrt{1 + \frac{1}{k}}} w(\OPT(S)).
\]
Now consider the case where all components are matroids. If $k = 1$, then the components of $\calF$ are pairwise disjoint and thus the system is a direct sum of matroids. In this case, for $p = 0$ we have $\Rel(Y) = S$ and, for all $j \in [M]$, $\Core_j(S) = \OPT(S)$. Thus,
\[
\E\brk{w\prn{\Core\prn{\Rel(Y)}}} = w\prn{\Core(S)} = w(\OPT(S)).
\]
If $k \geq 2$, by Lemma~\ref{lem:core-k-growth-general},
\[
\E\brk{w\prn{\Core\prn{\Rel(Y)}}} \geq \frac{(p - (1 - p) (k-1)) (1 - p)}{k} w(\OPT(S)).
\]
Optimizing $p$ yields $p = 1 - \frac{1}{2k}$ and
\[
\E\brk{w\prn{\Core\prn{\Rel(Y)}}} \geq \frac{1}{4k^2} w(\OPT(S)).\qedhere
\]
\end{proof}

Finally, using Algorithm~\ref{alg:free-general} and Theorem~\ref{thm:alg-general}, we obtain our main result for the free-order secretary problem on combinations of growth systems.

\begin{theorem}\label{thm:growth-systems-alg}
Let $\calF = (S,\cI)$ be a combination of systems $\set{(S_j,\cI_j)}_{j = 1}^M$ where each $(S_j,\cI_j)$ is a $k_j$-growth system. Let $k = \max_{e \in S} \sum_{j \colon e \in S_j} k_j$. Then, for $p = \sqrt{1 - {1}/{(k+1)}}$, Algorithm~\ref{alg:free-general} has a utility-competitive ratio at least
    \[
    \frac{1}{2} + \frac{1}{4k} - \frac{1}{2}\sqrt{1 + \frac{1}{k}} > \frac{1}{16k^2} - \frac{1}{32k^3}.
    \]
Furthermore, if all components are matroids, i.e. $\calF$ is a $k$-matchoid, Algorithm~\ref{alg:free-general} has a utility-competitive ratio at least $\frac{1}{4}$ if $k = 1$, and at least $\frac{1}{16k^2}$ if $k\geq 2$.
\end{theorem}
\section{The Agent-Arrival Setting}\label{sec:agent-arrival}

In this section, we turn our attention to the agent-arrival setting. As discussed in the introduction, the Core lemma (Lemma~\ref{lem:core-k-growth-general}) and thus our competitive ratio guarantees (Theorem~\ref{thm:growth-systems-alg}) unfortunately do not generalize directly from edge arrivals to this setting.

First, we need to ``transfer'' the constraints from the sets of agents and items to the set of edges. Let $\cF_A = (A, \cJ_A)$ be a combination of systems $\set{(A_j, \cJ_j)}_{j = 1}^{M_A}$ such that, for all $j \in [M_A]$, $(A_j,\cJ_j)$ is a $k_j$-growth system. Furthermore, let $\cF_B = (B, \cL_B)$ be a combination of systems $\set{(B_j,\cL_j)}_{j = 1}^{M_B}$ such that, for all $j \in [M_B]$, $(B_j,\cL_j)$ is an $\ell_j$-growth system.

Before we ``transfer'' $\cF_A$ and $\cF_B$ from the set of agents and items, respectively, to the edges, we need to take care of the matching constraint inherent in our model. This is easy to do however, since $k$-growth systems are closed under parallel-extension (for a proof and more information see Lemma~\ref{lem:extendible-contractions}). Therefore, for each system $(A_j, \cJ_j)$ for $j \in [M_A]$ (resp. $(B_j, \cL_j)$ for $j \in [M_B]$) and each agent $a \in A$ (resp. item $b \in B$) we create $|\delta(a)|$ (resp. $|\delta(b)|$) many copies of $a$ (resp. $b$). Since all copies are parallel, we associate each edge in $\delta(a)$ (resp. $\delta(b)$) with a unique copy of $a$ (resp. $b$). Since $(A_j, \cJ_j)$ (resp. $(B_j, \cL_j)$) is a $k_j$-growth (resp. $\ell_j$-growth) system, by Lemma~\ref{lem:extendible-contractions}, the resulting set system is still a $k_j$-growth (resp. $\ell_j$-growth) system. In addition, for every $(A_j, \cJ_j)$ (resp. $(B_j, \cL_j)$), no feasible solution includes two edges incident to the same original agent $a$ (resp. item $b$), and is thus still a matching.

Next, for each $(A_j, \cJ_j) \in \cF_A$, we create an independence system $(E_{A,j}, \cI_{A,j})$ on the set of edges. Then, the ground set $E_{A,j} = \delta(A_j)$ are the edges adjacent to the agents in $A_j$. Next, for a set $X \subseteq E_{A,j}$ of edges, let $V_A(X) = \set{a \in A_j \midd \delta(a) \cap X \neq \emptyset}$ denote the set of agents incident to some edge in $X$. Then, for the feasible sets, we have $\cI_{A,j} = \set{X \subseteq E_{A,j} \midd V_A(X) \in \cJ_j}$. Let $\cF_{A, E} = (E, \cI_A)$ denote the combination of $\set{(E_{A,j}, \cI_{A,j})}_{j = 1}^{M_A}$. Similarly, for each $(B_j, \cL_j) \in \cF_B$, we create an independence system $(E_{B,j}, \cI_{B,j})$ on the set of edges. Its ground set $E_{B,j} = \delta(B_j)$ are the edges adjacent to the items in $B_j$. Next, for a set $X \subseteq E_{B,j}$ of edges, let $V_B(X) = \set{b \in B_j \midd \delta(b) \cap X \neq \emptyset}$ denote the set of items incident to some edge in $X$. Then, for the feasible sets, we have $\cI_{B,j} = \set{X \subseteq E_{B,j} \midd V_B(X) \in \cL_j}$. Let $\cF_{B, E} = (E, \cI_B)$ denote the combination of $\set{(E_{B,j}, \cI_{B,j})}_{j = 1}^{M_B}$. Our global constraint $\cF = (E, \cI)$ on the set of edges can now be obtained as a combination of the systems $\set{(E_{A,j}, \cI_{A,j})}_{j = 1}^{M_A}$ and $\set{(E_{B,j}, \cI_{B,j})}_{j = 1}^{M_B}$. In other words, a set $X$ of edges is feasible for our global constraint if and only if $X \in \cI_A \cap \cI_B$.

\paragraph{Top-Relevant Edges.}

Similar to the edge-arrival setting, we begin by observing the edges adjacent to a sampled fraction of \emph{agents}. We then use this sampled set of agents to define a set of \emph{greedy-relevant edges} in subsequent phases of our algorithm.

Given a set $Y \subseteq A$ of agents, we apply the weighted-greedy algorithm for the system $(E, \cI)$ to the set of edges incident to $Y$ to obtain a set $\Greedy(\delta(Y)) \in \cI$. Recall that, by Definition~\ref{def:greedy-relevant-elements}, for every set $X$ of edges, we say that an edge $e \in E\setminus X$ is greedy-relevant with respect to $X$ if $e$ belongs to the output of the weighted-greedy algorithm on $X + e$. Then, for every agent $a \in A \setminus Y$, we associate the edge of highest weight in $\delta(a)$ that is greedy-relevant with respect to $\delta(Y)$, call it the \emph{top-relevant edge of $a$ with respect to $Y$}, and denote it by $\toprel(a,Y)$. In other words, $\toprel(a,Y) \coloneqq \argmax\set{ w(e) \midd e \in \delta(a), \text{ and } e \text{ is greedy-relevant with respect to } \delta(Y)}$. If $a$ has no greedy-relevant edges, we write $\toprel(a,Y) = \bot$.

In what follows, our aim is to show that, for a sufficiently large set of sampled agents $Y$, the \emph{set of top-relevant edges with respect to $Y$} which contains every top-relevant edge for every agent $a \in A \setminus Y$, retains a decent fraction of the weight of a maximum-weight basis. To keep our notation consistent with the Core lemma for edge arrivals, we abuse notation and use $\Rel(Y)$ to denote the set of top-relevant edges with respect to $Y$, i.e.
\[
\Rel(Y) \coloneqq \set{\toprel(a,Y) \in E \midd a \in A \setminus Y}.
\]

\subsection{The Core Lemma for Agent Arrivals}

We are now ready to describe the main result of this section: our generalization of the Core Lemma for combinations of growth systems, from the edge-arrival setting to the agent-arrival setting.

\begin{lemma}\label{lem:core-agent-arrival}
Let $\cF_{A, E} = (E, \cI_A)$ be a combination of $\set{(E_{A,j}, \cI_{A,j})}_{j = 1}^{M_A}$, where each $(E_{A,j}, \cI_{A,j})$ is a $k_j$-growth system, and $\cF_{B, E} = (E, \cI_B)$ be a combination of $\set{(E_{B,j}, \cI_{B,j})}_{j = 1}^{M_B}$, where each $(E_{B,j}, \cI_{B,j})$ is an $\ell_j$-growth system. Also, let $\cF = (E, \cI)$ denote the intersection of $\cF_{A,E}$ and $\cF_{B,E}$, let $\Core(X)$ denote the set of critical elements with respect to $X$ in $\cF$, i.e.,
\[
\Core(X) = \bigcap_{j=1}^{M_A}( (X\setminus E_{A,j})\cup\PCore_{j,A}(X\cap E_{A,j})) \cap \bigcap_{j=1}^{M_B}((X\setminus E_{B,j}) \cup \PCore_{j,B}(X\cap E_{B,j})),
\]
where $\PCore_{j,A}$ is the primitive core function of $(E_{A,j},\cI_{A,j})$ and $\PCore_{j,B}$ is the primitive core function of $(E_{B,j},\cI_{B,j})$.
Let
\[
k \coloneqq \max_{e \in E} \sum_{j \colon e \in E_{A,j}} k_j + \sum_{j \colon e \in E_{B,j}} \ell_j,
\]
and $Y = Y(p)$ be a random set where every agent $a \in A$ is in $Y$ independently with probability $p$ for some $p\in[0,1]$. Then,
\[
\E\brk{w\prn{\Core\prn{\Rel(Y)}}} \geq \frac{(p - (1 - p) k) (1 - p)}{p^2} \: \; \E\brk{w\prn{\Greedy(\delta(Y))}}.
\]
Furthermore, if all components are matroids, i.e. $\calF$ is a $k$-matchoid, then 
\[
\E\brk{w\prn{\Core\prn{\Rel(Y)}}} \geq \frac{(p - (1 - p) (k-1)) (1 - p)}{p} \: \; \E\brk{w\prn{\Greedy(\delta(Y))}}.
\]
\end{lemma}

\begin{proof}
Let $e_1, \dots, e_m$ be the elements of $E$ sorted in decreasing order of weights. We will define disjoint sets $R$, $N$, $R'$, and $N'$, where $R' \cup N' = \Greedy(\delta(Y))$ and $R \cup N = \Rel(Y)$, such that $R' = \Greedy(\delta(Y)) \cap \Core(\Rel(Y))$, $R = \Core(\Rel(Y))$ and $N' = \Greedy(\delta(Y)) \setminus R'$, $N = \Rel(Y) \setminus \Core(\Rel(Y))$. 
To construct $R$, $N$, $R'$ and $N'$, we use Algorithm~\ref{alg:simulation-agents} which is an offline simulation algorithm. Notice that Algorithm~\ref{alg:simulation-agents} is very similar to Algorithm~\ref{alg:simulation}, with the exception that, when we consider an edge $e_i$, if we have already processed the agent incident to $e_i$, this implies that $e_i$ is not the top-relative edge of its incident agent, and we ignore it. In Line~\ref{alg3:line18}, we have added the instruction ``Do nothing'' for clarity. Let us start by proving that the sets $R, N, R'$ and $N'$ satisfy the claimed conditions.

\begin{algorithm}[ht!]
\DontPrintSemicolon
$R, N, R', N' \gets \emptyset$ \\
$\overline{A}\gets \emptyset$ \tcp*{The set of \emph{processed agents}}
\For{$a \in A$} {\tcp*{If heads, $a$ is in the sampling phase}
    Toss a coin $c_a$ that is heads with probability $p$
}
\For{$i = 1$ to $m$} {
    Let $V_A(e_i)$ be the agent incident to $e_i$ \\
    \If{$(R'\cup N') + e_i \in \cI$} {
        \eIf{$V_A(e_i) \notin \overline{A}$} {
            $\overline{A} \gets \overline{A} + V_A(e_i)$ \;
            \eIf{$e_i \in \Core( (R\cup N) + e_i)$} {
                \leIf{$c_{V_A(e_i)}$ is heads} {
                    $R' \gets R' + e_i$
                }{
                    $R \gets R + e_i$
                }
            }{
                \leIf{$c_{V_A(e_i)}$ is heads} {
                    $N' \gets N' + e_i$
                }{
                    $N \gets N + e_i$
                }
            }
        }{ \tcp*{$e_i$ wasn't the first edge of $V_A(e_i)$ that could enter $\Greedy(E_i)$}
            Do nothing \label{alg3:line18}
        }
    }
}       
\caption{Simulation algorithm for agent arrivals}\label{alg:simulation-agents}
\end{algorithm}

Let $R_i, N_i, R'_i, N'_i, \overline{A}_i$ be the sets $R, N, R', N', \overline{A}$ respectively, at the end of the $i$-th iteration, where $R_0=N_0=R'_0=N'_0=\overline{A}_0=\emptyset$. Let $E_i = \set{e_1,\dots, e_i}$ and also say that an edge $e_i$ is a \emph{candidate for greedy} if, on iteration $i$, $R'_{i-1} \cup N'_{i-1} + e_i \in \cI$. Finally, let $Y$ be the set of agents $a\in A$ for which the coin $c_a$ is heads. Thus, we assume that they are the same set, and furthermore, we let $Y_i = \overline{A}_i \cap Y$. Also note $N_i$, $R_i$, $N'_i$ and $R'_i$ are pairwise disjoint. Let $I_i \coloneqq R'_i \cup N'_i$ and $X_i \coloneqq R_i \cup N_i$. 

\begin{claim}\label{clm:simulation-r-n-prime-agents}
For all $i \in [m]$, $\: R'_i \cup N'_i = \Greedy(\delta(Y_i))$.
\end{claim}
\begin{proof}
Consider the $i$-th step of our algorithm and let $a = V_A(e_i)$. If $e_i$ is not a candidate for greedy or if $c_a$ is tails, then in our algorithm $Y_{i+1} = Y_i$, $R'_i = R'_{i-1}$ and $N'_i = N'_{i-1}$ so the claim holds. Thus, assume that $e_i$ is a candidate for greedy and that $c_a$ is heads. Note that in this case, it must be that $a \notin \overline{A}_{i-1}$; if not, there exists an index $j < i$ such that $e_j$ is the first edge incident to $a$ that is a candidate for greedy. But since $c_a$ is heads, we have $a \in Y_j \subseteq Y_{i-1} \subseteq \overline{A}_{i-1}$, which is a contradiction. Thus, $a \notin \overline{A}_{i-1}$.

In that case, $Y_i = Y_i + a$ and since $e_i$ is a candidate for greedy,
\[
\Greedy(\delta(Y_i)) = \Greedy(\delta(Y_{i-1}) \cup \delta(a)) = \Greedy(\delta(Y_{i-1}) + e_i) = \Greedy(\delta(Y_{i-1})) + e_i,
\]
since $e_i$ is the first edge incident to $a$ that is a candidate for greedy. Indeed, all edges in $\delta(a)$ with higher weight than $e_i$ were discarded on previous iterations (i.e. they were not candidates for greedy) and all edges in $\delta(a)$ with lower weight than $e_i$ won't enter $\Greedy(\delta(Y_j))$ for any $j \geq i$, since $e_i$ is has entered $\Greedy(\delta(Y_i))$ and $\Greedy$ cannot contain two edges incident to the same agent.

Furthermore, by the description of the algorithm, $e_i$ will enter either $R'$ or $N'$ so, 
\[
R'_i \cup N'_i = R'_{i-1} \cup N'_{i-1} + e_i. \qedhere
\]
\end{proof}

\begin{claim}\label{clm:simulation-r-n-agents}
For all $i \in [m]$, $\: R_i \cup N_i = \Rel(Y) \cap E_i$.
\end{claim}
\begin{proof}
Consider the $i$-th step of our algorithm and let $a = V_A(e_i)$. First, note that if $e_i$ is not a candidate for greedy or if $c_a$ is heads, then in our algorithm $R_i \cup N_i = R_{i-1} \cup N_{i-1}$ so the claim holds. Thus, assume that $e_i$ is a candidate for greedy and that $c_a$ is tails.

Now, if $e_i$ is the top-relevant edge of $a$ with respect to $Y$, i.e. if $e_i = \toprel(a,Y)$, then we have that $e_i \in \Rel(Y) \cap E_i$ and also that $a \notin \overline{A}_{i-1}$, and thus, by the description of the algorithm, $e_i$ will enter either $R$ or $N$ so, 
\[
R_i \cup N_i = R_{i-1} \cup N_{i-1} + e_i.
\]
On the other hand, if $e_i \neq \toprel(a,Y)$, then $e_i \notin \Rel(Y) \cap E_i$. This, however, implies that there exists a $j < i$ such that $e_j = \toprel(a,Y)$ and, since $e_j$ has higher weight than $e_i$, this means that we have processed $a$ earlier and $a \in \overline{A}_{i-1}$. In either case, the claim holds.
\end{proof}

\begin{claim}\label{clm:simulation-r-agents}
For all $i\in[m]$, $\: R_i = \Core(\Rel(Y)) \cap E_i$.
\end{claim}
\begin{proof}
We show both inclusions. 
Let $e \in R_i$. Then, there exists $j < i$ such that $e = e_j$, $e_j \in R_j$ and 
\[
e_j \in \Core(R_{j-1} \cup N_{j-1} + e_j) = \Core(R_j \cup N_j) = \Core(\Rel(Y) \cap E_j).
\]
Since all the elements in $\Rel(Y) \setminus E_j$ have weight smaller than that of $e_j$ we also get that $e_j \in \Core(\Rel(Y))$, and therefore $e_j \in \Core(\Rel(Y)) \cap E_i$.

For the other direction, let $e \in \Core(\Rel(Y)) \cap E_i \subseteq \Rel(Y) \cap E_i$. This means that $e = e_j$ for some $j \leq i$ and $e_j \in \Core(\Rel(Y)) \cap E_j$. Since $\Core(\Rel(Y)) \subseteq \Rel(Y)$, by Claim~\ref{clm:simulation-r-n-agents}, we know that $e_j \in R_j \cup N_j$. However, since $e_j \in \Core(\Rel(Y)) \cap E_j$ we also have that $e_j \in \Core(\Rel(Y_{j-1})) \cap E_j$ and thus $e_j \in \Core((R_{j-1} \cup N_{j-1}) + e_j)$, which implies that $e_j \in R_j \subseteq R_i$.
\end{proof}

\begin{claim}\label{clm:simulation-r-prime-agents}
For all $i\in[m]$, $\: R'_i = \Greedy(\delta(Y_i)) \cap \Core(\Rel(Y))$.
\end{claim}
\begin{proof}
By Claim~\ref{clm:simulation-r-n-prime-agents}, we have that $R'_i \subseteq \Greedy(\delta(Y_i))$. In particular, for $j \leq i$, $R'_i$ corresponds to the set of edges $e_j$ from $\Greedy(\delta(Y_i))$ which also satisfied $e_j \in \Core((R_{j-1} \cup N_{j-1}) + e_j)$ and for which the agent $a$ incident to $e_j$ is not processed before step $j$. By Claim~\ref{clm:simulation-r-n-agents}, the former is equivalent to $e_j \in \Core((\Rel(Y) \cap E_{j-1}) + e_j)$, which implies that $e_j \in \Core(\Rel(Y) \cap E_j) \subseteq \Core(\Rel(Y))$. For the latter, every edge $e_j \in \Greedy(\delta(Y_i)) \cap \Core(\Rel(Y))$ is the top-relevant edge of its incident agent $a$ and thus we have that $V_A(e_j) \notin \overline{A}_{j-1}$. Thus, the claim holds.
\end{proof}

Now we will prove the following four inequalities about the sizes of the constructed sets.
\begin{itemize}\itemsep0em
    \item $\E\bigl[|N'|\bigr] \geq \frac{p}{1-p} \: \E[|N|], \hfill \refstepcounter{equation}(\theequation) \label{eq:agent-simulation-1}$
    \item $\E[|R|] + \E[|N|] \geq \frac{1-p}{p} (\E[|R'|] + \E[|N'|]), \hfill \refstepcounter{equation}(\theequation) \label{eq:agent-simulation-2}$
    \item $\E[|N'|] \leq k \: (\E[|R|] + \E[|N|]), \hfill \refstepcounter{equation}(\theequation) \label{eq:agent-simulation-3}$
    \item $\E[|N'|] \leq k \: \E[|R|] + (k-1) \: \E[|N|]$ if all components of $\calF$ are matroids. $\hfill \refstepcounter{equation}(\theequation) \label{eq:agent-simulation-4}$
\end{itemize}

The proofs of \eqref{eq:agent-simulation-1} and \eqref{eq:agent-simulation-2} are a bit more subtle in the agent arrival case than in the edge arrival case.
\begin{claim}
For all $i \in [m]$,
\[
(1-p) \E[|R'_i|] = p \: \E[|R_i|] \quad \text{ and } \quad (1-p) \E[|N'_i|] = p \: \E[|N_i|].
\]
\end{claim}
\begin{proof}
Fix an agent $a$. Observe that the first time an edge $e_i \in \delta(a)$ is about to enter $R' \cup R$ (resp. $N'\cup N$), the result of $c_a$ is used for the first time. If the coin is heads, which happens with probability $p$, then $e_i$ is added to $R'$ (resp. $N'$), otherwise, to $R$ (resp. $N$). Since $a \in \overline{A}_i$, no other edge $e_j \in \delta(a)$ with $j > i$ will ever be about to enter $R' \cup R$ (resp. $N'\cup N$) and thus $c_a$ will not be used again.
\end{proof}

With \eqref{eq:agent-simulation-1}, \eqref{eq:agent-simulation-2} and Claims~\ref{clm:simulation-r-n-prime-agents},~\ref{clm:simulation-r-n-agents},~\ref{clm:simulation-r-agents} and~\ref{clm:simulation-r-prime-agents} established, the rest of the proof follows exactly like the proof of Lemma~\ref{lem:core-k-growth-general}. We first use the fact that $\cF$ is a $k$-growth system to establish \eqref{eq:agent-simulation-3} and \eqref{eq:agent-simulation-4}, and then we manipulate the four established inequalities in the same way as in Lemma~\ref{lem:core-k-growth-general} to obtain the desired conclusion.
\end{proof}

We summarize the guarantees of our Core lemma with respect to $w(\OPT)$. Note that the guarantees we obtain are identical to those of Theorem~\ref{thm:core-lemma-vs-opt} for the edge-arrival setting.

\begin{theorem}\label{thm:core-lemma-agents-vs-opt}
Let $\cF_{A, E} = (E, \cI_A)$ be a combination of $\set{(E_{A,j}, \cI_{A,j})}_{j = 1}^{M_A}$, where each $(E_{A,j}, \cI_{A,j})$ is a $k_j$-growth system, and $\cF_{B, E} = (E, \cI_B)$ be a combination of $\set{(E_{B,j}, \cI_{B,j})}_{j = 1}^{M_B}$, where each $(E_{B,j}, \cI_{B,j})$ is an $\ell_j$-growth system. Also, let $\cF = (E, \cI)$ denote the combination of $\cF_{A,E}$ and $\cF_{B,E}$, let $\Core(X)$ denote the set of critical elements with respect to $X$ in $\cF$, i.e.,
\[
\Core(X) = \bigcap_{j=1}^{M_A}( (X\setminus E_{A,j})\cup\PCore_{j,A}(X\cap E_{A,j})) \cap \bigcap_{j=1}^{M_B}((X\setminus E_{B,j}) \cup \PCore_{j,B}(X\cap E_{B,j})),
\]
where $\PCore_{j,A}$ is the primitive core function of $(E_{A,j},\cI_{A,j})$ and $\PCore_{j,B}$ is the primitive core function of $(E_{B,j},\cI_{B,j})$,
and let $Y = Y(p)$ be a random set where every agent $a \in A$ is in $Y$ independently with probability $p$ for some $p\in[0,1]$. Then,
for $p =\sqrt{1 - \frac{1}{k+1}}$, we have
\[
\E\brk{w\prn{\Core\prn{\Rel(Y)}}} 
\geq \prn{2 + \frac{1}{k} - 2\sqrt{1 + \frac{1}{k}}} w(\OPT(E)) > \prn{\frac{1}{4k^2} - \frac{1}{8k^3}} w(\OPT(E)).
\]
Furthermore, if all components are matroids, i.e., $\calF$ is a $k$-matchoid, setting $p=0$ for $k=1$ and $p=1-1/(2k)$ for $k\ge2$ yields
\[
\E\brk{w\prn{\Core\prn{\Rel(Y)}}}
\geq \begin{cases}
w(\OPT(E))                      & \text{if } k = 1, \\
\frac{1}{4k^2} \: w(\OPT(E))    & \text{if } k \geq 2. \\
\end{cases}
\]
\end{theorem}
\begin{proof}
Since the system $(E, \cI)$ is a $k$-system the weighted-greedy algorithm yields a $1/k$-approximation with respect to the weight of $\OPT(E)$. Therefore,
\[
\E[w(\Greedy(\delta(Y)))] \geq \frac{1}{k} \E[w(\OPT(\delta(Y)))] \geq \frac{1}{k} \E[w(\OPT(E)\cap \delta(Y))] = \frac{p}{k} w(\OPT(E)).
\]
Thus, by Lemma~\ref{lem:core-agent-arrival},
\[
\E\brk{w\prn{\Core\prn{\Rel(Y)}}} \geq \frac{(p - (1 - p) k) (1 - p)}{p k} w(\OPT(E)).
\]
Optimizing $p$ yields $p = \sqrt{1 - \frac{1}{k+1}}$ and
\[
\E\brk{w\prn{\Core\prn{\Rel(Y)}}} \geq \prn{2 + \frac{1}{k} - 2\sqrt{1 + \frac{1}{k}}} w(\OPT(E)).
\]
Now, consider the case where all components are matroids. If $k = 1$, then the components of $\calF$ are pairwise disjoint and thus the system is a direct sum of matroids, so it is itself a matroid. In this case, for $p = 0$ we have $\Rel(Y) = E$ and, for all $j \in [M]$, $\Core_j(E) = \OPT(E)$. Thus,
\[
\E\brk{w\prn{\Core\prn{\Rel(Y)}}} = w\prn{\Core(E)} = w(\OPT(E)).
\]
If $k \geq 2$, by Lemma~\ref{lem:core-agent-arrival},
\[
\E\brk{w\prn{\Core\prn{\Rel(Y)}}} \geq \frac{(p - (1 - p) (k-1)) (1 - p)}{k} w(\OPT(E)).
\]
Optimizing $p$ yields $p = 1 - \frac{1}{2k}$ and
\[
\E\brk{w\prn{\Core\prn{\Rel(Y)}}} \geq \frac{1}{4k^2} w(\OPT(E)).\qedhere
\]
\end{proof}

\subsection{Order-Oblivious Core-Selecting Algorithms}

As discussed previously, for edge arrivals, our guarantees follow directly by using the Core lemma in conjunction with our black-box algorithm from Theorem~\ref{thm:alg-general}. For agent arrivals, this is no longer the case, for two main reasons. 

\paragraph{Problems with the Edge-Arrival Approach.}
First, suppose we try to repeat the approach of our black-box algorithm from Theorem~\ref{thm:alg-general} for edge-arrival: split the remaining agents into two parts $S_1$ and $S_2$, use $S_1$ to define a calling order, and then process $S_2$ in that order. Of course, we cannot ``call'' individual edges and when we ``call'' an agent, we observe all their incident edges simultaneously. The problem arises from the fact that we cannot ``call'' all edges incident to a single item. This asymmetry implies that the order in which our algorithm should ``call'' the agents should depend on the $\Hull$ of $\cF_{A,E}$, i.e. the system from the agent side, to guarantee that any agent whose top-relevant edge lies in the core of $\cF_{A,E}$ has a constant probability of being considered for selection.

However, even if we call the agents in the order specified by the $\Hull$ of $\cF_{A,E}$, this still does not resolve all our problems, as we still might not be able to select a decent fraction of the common core of $\cF_{A,E}$ and $\cF_{B,E}$. This is because the edges that our algorithm selects (that are in the core of the top-relevant edges of the sample with respect to $\cF_{A,E}$) may be dependent in the item-side system $\cF_{B,E}$, since we do not have any guarantees that these edges are also in the core of the top-relevant edges of the sample with respect to $\cF_{B,E}$. Ideally, we would want to select edges that are also in the item-side core, but the arrival order defined by the agent side does not ensure we pick them; indeed, each time we select an edge outside the item core, we may block one or many later edges in the item core.

To circumvent this issue, we focus our attention to systems $\cF_{B,E}$ on the item side that admit an \emph{order-oblivious core-selecting algorithm}; one that, after a uniformly random sample phase, ensures that every element in the core of $\cF_{B,E}$ that was not sampled is selected with constant probability, regardless of the arrival order of the elements in the sampling phase.

\begin{definition}\label{def:oocs}
Given an independence system $(S,\cI)$ with hidden weights on the elements, we say that an algorithm $\calA_r$ that uses $r$ internal random bits and selects a set of elements $\ALG \in \cI$ is $(q,\alpha)$ \emph{order-oblivious core-selecting} (OOCS) if, for any $X \subseteq S$ and any $e \in \Core(X)$, if $Y$ is a random set where every element of $X - e$ appears independently with probability $q$, then
\[
\E_{Y \sim X - e}\brk{ \min_{\sigma} \Pr_r\brk{e \in \ALG \midd e \notin Y}} \geq \alpha,
\]
where the minimum is taken over all  orderings $\sigma$ of the elements in $X \setminus Y$.
\end{definition}
To better explain the above definition, we describe how an $(q, \alpha)$-OOCS algorithm $\calA_r$ works.
\begin{enumerate}\itemsep0em
    \item Before the process begins, $\calA_r$ is presented with $(S, \cI)$.

    \item An adversary is then allowed to delete an arbitrary set of elements from $S$ and restrict the true ground set of the process to an arbitrary $X \subseteq S$ that is \emph{unknown} to $\calA_r$.
    
    \item The process begins and $\calA_r$ observes a random subset $Y \subseteq X$, which we call the sample, that includes each element of $X$ independently with probability $q$. $\calA_r$ learns the identities and weights of the elements in $Y$, and does not select any of them.
  
    \item $\calA_r$ then observes the remaining elements $\overline{Y} = X \setminus Y$ in an adversarial order that may depend on the sample $Y$, but not on $\calA_r$'s internal random bits $r$.
  
    \item Upon observing an element of $\overline{Y}$, $\calA_r$ must immediately and irrevocably decide whether to select it -- adding it to $\ALG$ -- subject to maintaining $\ALG \in \cI$.

    \item For any arbitrary restriction $X \subseteq S$, and conditioning only on the event that an element $e \in \Core(X)$ is not sampled, i.e. on $e \in \overline{Y}$, we require that the expected probability that $\calA_r$ selects $e$, regardless of the arrival order of the elements in $\overline{Y}$ in the selection phase, is at least $\alpha$. The expectation is over the random sample $Y$ and the probability is over the algorithm’s internal random bits $r$.
\end{enumerate}

Notice that Definition~\ref{def:oocs} is similar to the notion of a $c$-OPT competitive algorithm for random-order matroid intersection secretary problems, introduced by \cite{fsz-secretary-framework-matroid-intersection}, except that here we require selecting \emph{core elements} rather than \emph{elements in the optimal basis}. However, it is actually a slightly stronger requirement since an order-oblivious core-selecting algorithm has to work for any restriction of the original system, without knowing the true ground set, and give the same guarantees. This stronger requirement is necessary for our algorithm, in order to resolve our final obstacle, which is that our agent-side sample does not induce a uniformly random sample of items or even of item-adjacent edges. We describe our idea to circumvent this last problem in Section~\ref{sec:agents-algorithm}.

Given the above, one may wonder whether $(q, \alpha)$-OOCS algorithms exist for constant $q$ and $\alpha$ for any feasibility constraints. As it turns out, several matroid classes have known algorithms that are $(q, \alpha)$-OOCS for constant $q$ and $\alpha$, including uniform, partition, laminar and co-graphic matroids. In addition, even though no order-oblivious core-selecting algorithm with good guarantees is known for transversal matroids, graphic matroids and knapsack systems, we discuss in Section~\ref{sec:agent-remarks} how to transform them into systems that we can use for our algorithms with minimal losses in the competitive ratio guarantee.

Setting aside briefly the requirement for the algorithm to work for arbitrary restrictions, notice that, for matroids, every order-oblivious algorithm also selects elements in the core of any ground set, since the elements of the core are exactly the elements of the optimal basis. However, this need not be the case for general independence systems.

Next, we present an example of a simple $(1/2, 1/2)$-OOCS algorithm for rank-$1$ matroids.
\begin{example}
Let $(S,\cI)$ be a uniform rank-$1$ matroid. The algorithm takes the sample, computes the maximum-weight sampled element $e^*$ and, in the non-sampled phase, selects the first observed element $f$ with $w(f) > w(e^*)$. If the sample is empty, it selects the first non-sampled element it observes.
Notice that this algorithm does not require knowing $|S|$ a priori. To see that it is $(1/2,1/2)$-OOCS, fix any $X \subseteq S$ and let $e \in \Core(X) = \OPT(X)$ be the maximum-weight element of $X$. For $Y = S(q)$ and $q = 1/2$, we condition on $e \notin Y$. Feed $Y \cap X$ to the algorithm as the sample, and then present $X \setminus Y$
adversarially. Let $e_2$ be the element in $X$ with second-highest weight. With probability $1/2$, we have $e_2 \in Y$. In this case, $w(e^*) \ge w(e_2)$ and the only element $f$ that satisfies $w(f) > w(e^*)$ is $e$, so $e$ is selected in the second phase. Hence $\Pr\brk{e \in \ALG \midd e \in \overline{Y}} \geq 1/2$.
\end{example}

Next, we show that, for matroids, random-order constant-probability-competitive algorithms can be transformed to order-oblivious core-selecting algorithms.
\begin{observation}\label{obs:oocs-to-random-order-matroids}
For any matroid $M = (S,\cI)$, if $\calA$ is a $(q,\alpha)$-OOCS algorithm for $M$, there exists a $(1-q) \alpha$-probability competitive algorithm for $M$, when the elements in $S$ arrive in uniformly random order
\end{observation}
\begin{proof}
Let $\ALG$ denote the set of elements selected by $\calA$ in the random-order setting, $Y$ denote the sample and $\overline{Y} = S \setminus Y$. For every $e\in \OPT$,
\begin{align*}
\Pr\brk{e\in \ALG} &= \Pr\brk{e \in \overline{Y}} \: \E_{Y - e}\brk{ \Pr\brk{e \in \ALG \midd e \in \overline{Y}} } \\
&\geq (1-q) \: \E_{Y - e}\brk{ \min_\sigma \Pr\brk{e \in \ALG \midd e \in \overline{Y} , \sigma} } \geq (1-q) \alpha.\qedhere
\end{align*}
\end{proof}

With regard to our model, when the given item-side constraint is a combination of individual systems, having an OOCS algorithm for the individual components suffices to obtain an OOCS algorithm for the overall constraint, as the following lemma shows. Recall that, for a combination of $M$ systems, $\PCore_j$ denotes the primitive core of the $j$-th individual system for $j \in [M]$, and by Definition~\ref{def:core}, for any $X \subseteq S$,
\[
\Core(X) \coloneqq \bigcap_{j = 1}^M \prn{(X \setminus S_j) \cup \PCore_j(X \cap S_j)}.
\]

\begin{lemma}\label{lemma:order-oblivious}
Let $\cF = (S,\cI)$ be a combination of systems $\set{(S_j,\cI_j)}_{j = 1}^M$ where each $(S_j,\cI_j)$ admits a $(q_j,\alpha_j)$-OOCS algorithm $\calA_j$. Let $k = \max_{e \in S} \abs{\set{j \colon e \in S_j}}$ denote the maximum number of individual systems any element participates in and $q^{\star} = \max_{j = 1}^M q_j$. Then, for
\[
q \coloneqq 1 - \prn{1 - q^{\star}}^k \quad \text{ and } \quad \alpha \coloneqq \min_{e \in S} \prod_{j \colon e \in S_j} \alpha_j \: \:,
\]
there exists a $(q, \alpha)$-OOCS algorithm $\calA$ for the combination $(S, \cI)$ of $\set{(S_j, \cI_j)}_{j = 1}^M$.
\end{lemma}
\begin{proof}
For every element $e$ and every $j \in [M]$, we independently draw $t'_j(e) \sim U[0,1]$ from the uniform distribution in the interval $[0,1]$. Then, for every $e \in S$ we denote by $j_{\ell}(e)$ the index of the $\ell$-th individual system in the ordering $(S_1, \cI_1), (S_2, \cI_2), \dots, (S_M, \cI_M)$ that $e$ participates in. For $\ell = 1, 2, \dots, k$, let $t_\ell(e)$ be equal to $t'_{j_{\ell}(e)}$. If $e$ participates in $h < k$ systems, we let $t_{h+1}(e), \dots, t_k(e) \sim U[0,1]$ be drawn independently from all other random variables.

Consider the following sets:
\[
\forall j \in [M], \: \: Y_j \coloneqq \set{e \in S_j \midd t_j(e) \leq q_j}, \quad \text{ and } \quad Y \coloneqq \set{e \in S \midd \exists \ell \in [k] \: \text{ s.t. } \: t_\ell(e) \leq q^{\star}}.
\]
Notice that for every $j \in [M]$ we have $Y \supseteq Y_j$, since $q_j \leq q^{\star}$. Furthermore, for every element $e \in S$, we have $\Pr[e \in Y] = 1 - \prn{1 - q^{\star}}^k = q$, since even if $e$ participates in less than $k$ systems, we still toss $k$ coins for $e$ and observe whether or not one of them is less than $q^{\star}$ to decide whether $e \in Y$ or not. Furthermore, the events $(e\in Y)_{e\in S}$ are mutually independent by construction. Therefore, $Y$ is a truly random sample of $S$ with inclusion probability $q$. Similarly, the random sets $(Y_j)_{j\in [M]}$ are also mutually independent random variables since they depend on the result of different random coins.

Let $X \subseteq S$ denote the unknown arbitrary set that the adversary restricts our ground set to before the process begins. Our algorithm will receive a sample in which each element of $X$ is independently selected with probability $q$. Notice that $Y \cap X$ is a random set distributed exactly according to the sample of our distribution, so in what follows we assume that this set is really $Y \cap X$. Observe also that $Y_j$ is a random subset of $S_j$ such that for every element $e \in S_j$, we have $e \in Y_j$ independently with probability $q_j$. 

By the same logic as before, we have that $Y_j \cap X$ is a random set distributed exactly according to the sample expected by $\calA_j$, restricted to $X$. Since $Y_j \subseteq Y$ for all $j$, our algorithm $\calA$ proceeds as follows:
\begin{enumerate}
    \item $\calA$ receives the set of sampled elements which, as discussed, is identically distributed to $Y \cap X$ so we assume it is really that set.
    \item Next, $\calA$ will create simulated realizations of the random variables $t'_\ell(e)$ for $\ell = 1, 2, \dots, k$ and $e\in X\cap Y$ that are compatible with the definition of $Y$, so they are identically distributed. To do that each element $e \in Y \cap X$ and $j \in [M]$, $\calA$ will draw $t'_j(e) \sim U[0,1]$, as described above. It will repeat this process until one of the $t_\ell(e)$'s is less than $q$, for $\ell = 1, 2, \dots, k$. When this happens, the conditioning is correct and we assume that these are the correct draws of the $t_\ell(e)$'s. 
    \item  Using the previous step, for every $j \in [M]$, $\calA$ creates a simulated set by including every element $e\in Y\cap X\cap S_j$ for which the random variable $t_\ell(e)$, for the value of $\ell$ that corresponds to system $(S_j, \cI_j)$, is less than $q_j$. The resulting set is now distributed identically to $Y_j \cap X$, so in what follows we assume that this set is really $Y_j \cap X$.
    \item For every individual system $(S_j, \cI_j)$, $\calA$ feeds $Y_j \cap X$ to $\calA_j$ as its sample. Since $Y_j \cap X$ is identically distributed to the sample expected by $\calA_j$ restricted to $X$, we now start the second, selection phase for every system $j \in [M]$.
    \item Let $\overline{Y} = X \setminus Y$ denote the set of elements in the selection phase. Assume that $\calA$ observes them in an adversarial arrival order $\pi$.
    \item At every step $i$, $\calA$ observes an element $e_i$. For every $j$ such that $e_i \in S_j$, $\calA$ ``feeds'' $e_i$ to $\calA_j$ in the selection phase. If all such $\calA_j$ accept $e_i$, $\calA$ accepts $e_i$ as well; otherwise it rejects $e_i$.
    \item After all elements of $\overline{Y}$ end, $\calA$ feeds all the elements of $(Y\setminus Y_j) \cap S_j$ to $\calA_j$, for each $j$ in arbitrary order. Nothing is really selected here, but the process is done anyway so that each $\calA_j$  receives all the elements from $\overline{Y}\cap S_j$ on its selection phase. Then, the process ends.
\end{enumerate}
We now show that $\calA$ is $(q,\alpha)$-OOCS. Previously, we showed that the sample of $\calA$ is distributed identically to $Y \cap X$, which is a random subset of $X$ where each element appears independently with probability $q$.
Fix an element $e \in \Core(X)$ and let $\ALG$ denote the set selected by $\calA$. Next, we show that
\[
\E_{Y \sim X - e}\brk{ \min_{\sigma} \Pr_r\brk{e \in \ALG \midd e \notin Y}} \geq \alpha.
\]
We have
\begin{align*}
\E_{Y \sim X - e}&\brk{ \min_{\sigma} \Pr_r\brk{e \in \ALG \midd e \notin Y}} \\
&= \E_{Y \sim X - e}\brk{\min_{\sigma} \Pr_r\brk{\bigwedge_{j \colon e \in S_j} \calA_j \text{ accepts } e \midd e \notin Y}} \\
&= \E_{Y \sim X - e}\brk{\min_{\sigma} \Pr_r\brk{\bigwedge_{j \colon e \in S_j} \calA_j \text{ accepts } e \midd \forall j \text{ s.t. } e \in S_j, \: \: e \notin Y_j}} \\
&= \E_{Y_1 \overset{q_1}{\sim} (X - e) \cap S_1, \dots, Y_M \overset{q_M}{\sim} (X - e) \cap S_M}\brk{\min_{\sigma} \prod_{j \colon e \in S_j} \Pr_r\brk{\calA_j \text{ accepts } e \midd e \notin Y_j}} \\
&= \prod_{j \colon e \in S_j} \E_{Y_1 \overset{q_1}{\sim} (X - e) \cap S_1, \dots, Y_M \overset{q_M}{\sim} (X - e) \cap S_M}\brk{\min_{\sigma} \Pr_r\brk{\calA_j \text{ accepts } e \midd e \notin Y_j}} \\
&\geq \prod_{j \colon e \in S_j} \alpha_j \geq \min_{e' \in S} \prod_{j \colon e' \in S_j} \alpha_j = \alpha,
\end{align*}
where the second equality follows from the fact that, for any two systems $(S_j, \cI_j)$ and $(S_{j'}, \cI_{j'})$ for which $e \in S_j$ and $e \notin S_{j'}$, it is irrelevant to $\calA_j$ whether or not $e \in Y_{j'}$, given that $e \notin Y_j$, the third equality follows from the fact that the $Y_j$'s are independent subsets of $(X - e) \cap S_j$, conditioned on $e \notin Y$, the fourth equality follows from the independence of the decisions of all $\calA_j$'s and the fifth inequality follows from the fact that every $\calA_j$ is a $(q_j, \alpha_j)$-OOCS algorithm.
\end{proof}

\subsection{Free-Order Secretary Algorithm for Agent Arrivals}\label{sec:agents-algorithm}

As discussed earlier, we are now faced with our final obstacle. Even if we have order-oblivious core-selecting algorithms for the constraints on the item side, they still require a uniformly random sample to work. However, when we sample agents, we simultaneously observe all edges incident to those agents. From the perspective of the item side, this set is neither a uniformly random sample of item-adjacent edges, nor does it induce a uniformly random sample of the corresponding items. This is easy to see: consider an extreme case where, originally, there are only two items, one connected to all agents and the other connected only to a single agent. After our parallel extension construction and our transfer of $\cF_B$ from the items to the edges, every edge is connected to a unique agent and a unique item, but we have many copies of the first item and only a single copy of the second one. So a uniformly random sample of edges does not correspond to a uniformly random sample of items, and thus no such sample can be used an an input to an order-oblivious core-selecting algorithm.

However, we are able to get around this issue by observing that, after our very first sampling phase, each agent contributes at most one greedy-relevant edge -- their top greedy-relevant edge, if one exists --  to the set of greedy-relevant edges. Now the idea is to do the following: after the first sample, we restrict our ground set for the item side to the (unknown) set of top-relevant edges. Let's call this set $X$. Essentially, $X$ forms the true ground set of our order-oblivious core-selecting algorithms for the item side. Now, even though we cannot ``call'' items, we can perform another sampling phase, independently from our first sampling phase, to ensure that the sampled top-relevant edges in this second sampling phase forms a uniformly random subset of the new ground set $X$.

This insight motivates the stronger requirement we impose on order-oblivious core-selecting algorithms: they must maintain their guarantees even if, at the beginning of the process, an arbitrary and unknown subset of the ground set is removed (corresponding here to the edges that are not top-relevant from the agent side). In other words, the algorithm must perform equally well under any adversarial restriction of the ground set applied before the process begins.

We finally have all the ingredients we need to design an algorithm for the agent-arrival setting. Note that, for the $\Hull$ function, we use the $\Hull$ of $\cF_A$, the system on the agent side.
Algorithm \ref{alg:free-agent} receives as input the following; see Figure~\ref{fig:sampling-order-oblivious} for an illustration:
\begin{enumerate}\itemsep0em
    \item $\cF_{A,E} = (E, \cI_A)$ is the (edge-induced) agent-system.
    \item $\cF_{B,E} = (E, \cI_B)$ is the (edge-induced) item-system.
    \item $\calA$ is a $(q, \alpha)$-OOCS algorithm for $\cF_{B,E}$.
    \item A parameter $p$.
\end{enumerate}

\begin{algorithm}[th!]
\DontPrintSemicolon
\KwResult{An independent set $\ALG$}
    $\ALG \gets \emptyset$ \\
    $p_1 \gets p$ \\
    $p_2 \gets p_1 + (1 - p_1)q$ \\
    $p_3 \gets (1+p_2)/2$ \;
    \lFor{each $a \in A$} {
        Choose $t_a$ independently and uniformly from $(0,1)$
    }
    $Y \gets \set{a \midd t_a \in (0, p_1)}$ \;
    $S_0 \gets \set{a \midd t_a \in [p_1,p_2)}$ \;
    $S_1 \gets \set{a \midd t_a \in [p_2, p_3)}$ \;
    $S_2 \gets \set{a \midd t_a \in [p_3, 1]}$ \;
    Observe (without accepting) all agents in $Y \cup S_0\cup S_1$

    $\textbf{unseen} \gets S_2$ \tcp*{The set of agents not yet revealed}
    Compute the edge-sets $E_0 = \Rel(Y) \cap \delta(S_0)$ and $E_1 = \Rel(Y)\cap \delta(S_1)$\;
    Run $\calA$ on $\cF_B|_{\Rel(Y)} = (E, \cI_B)|_{\Rel(Y)}$ using $E_0$ as the sample \;
    
    Sort the elements of $E_1$ in decreasing order of weights as $e_1,\dots, e_m$. Let $a_1, \dots, a_m$ be the respective agents such that $e_i$ is incident to $a_i$ \;
    \For{$j = 1$ to $m$}{
        Let $Q = \textbf{unseen} \cap \Hull_A(\set{a_1,\dots, a_j})$ \;
        \For{each $a \in Q$ in uniformly random order}{
            $\textbf{unseen} \gets \textbf{unseen} - a$ \tcp*{Reveal $a$}
            \lIf{$\toprel(a,Y) = \bot$}{skip to the next $a\in Q$}
            $f \gets \toprel(a,Y)$
            
            \If{$w(f) > w(e_j)$ and $\calA$ selects $f$ and $\ALG + f \in \cI_A\cap \cI_B$}{
                $\ALG \gets \ALG + f$
            }
        }
      }

    $Q \gets \textbf{unseen}$ \\
    \For{each $a \in Q$ in uniformly random order}{
            \lIf{$\toprel(a,Y) = \bot$}{skip to the next $a\in Q$}
            $f \gets \toprel(a,Y)$ 
            
            \If{$\calA$ selects $f$ and $\ALG + f \in \cI_A\cap \cI_B$}{
                $\ALG \gets \ALG + f$
            }
        }
    \Return $\ALG$
\caption{Agent-arrival Free-Order Secretary}\label{alg:free-agent}
\end{algorithm}

\begin{theorem}\label{thm:alg-general-agent}
Let $\cF_{A, E} = (E, \cI_A)$ be a combination of $\set{(E_{A,j}, \cI_{A,j})}_{j = 1}^{M_A}$, where each $(E_{A,j}, \cI_{A,j})$ is a $k_j$-growth system, and $\cF_{B, E} = (E, \cI_B)$ be a combination of $\set{(E_{B,j}, \cI_{B,j})}_{j = 1}^{M_B}$, where each $(E_{B,j}, \cI_{B,j})$ is an $\ell_j$-growth system. Also, let $\cF = (E, \cI)$ denote the intersection of $\cF_{A,E}$ and $\cF_{B,E}$ and $\calA$ be a $(q, \alpha)$-OOCS algorithm for $\cF_{B,E}$. Then, for $p = \sqrt{1 - {1}/{(k+1)}}$, Algorithm~\ref{alg:free-agent} has a utility-competitive ratio at least
\[
\frac{\alpha(1-q)}{4}\prn{2 - 2\sqrt{1 + \frac{1}{k}} + \frac{1}{k}} > \alpha(1-q) \prn{\frac{1}{16k^2} - \frac{1}{32k^3}}.
\]
Furthermore, if all components are matroids, i.e. if $(E,\cI_A \cap \cI_B)$ is a $k$-matchoid with $k_{j,A} = 1$ and $\ell_{j,B} = 1$ for every $j$ and $\abs{\set{j \colon e\in E_{A,j}}} + \abs{\set{j \colon e \in E_{B,j}}} \leq k$ for every $e \in E$, for $p = 1 - \frac{1}{2k}$, Algorithm \ref{alg:free-agent} has a utility-competitive ratio at least
\[
\frac{\alpha (1-q)}{16k^2}.
\]
\end{theorem}
\begin{proof}
First, it is clear that the algorithm returns an independent set since, before adding any edge $e$ to $\ALG$, the algorithm checks whether $\ALG + e \in \cI$.
We first condition on a fixed realization of the sampled set $Y$ of agents. Let
\[
\Core_{\text{Common}}(\Rel(Y)) = \Core_A(\Rel(Y)) \cap \Core_B(\Rel(Y)),
\]
where
\[
\Core_A(\Rel(Y)) = \bigcap_{j = 1}^{M_A} \Core_{A,j}(\Rel(Y)) \text{ and } \Core_B(\Rel(Y)) = \bigcap_{j = 1}^{M_B} \Core_{B,j}(\Rel(Y)),
\]
where $\Core_{A,j}$ ($\Core_{B,j}$ respectively) corresponds to the individual core function of the system component $(E_{A,j},\cI_{A,j})$ (resp. $(E_{B,j}, \cI_{B,j})$).
We show that each element in $\Core_{\text{Common}}(\Rel(Y))$ is selected with probability at least $\alpha(1-q)/4$. Let $e \in \Core_{\text{Common}}(\Rel(Y)) \subseteq \Rel(Y)$. Condition now on the event that $e \notin \delta(S_0)$ (i.e. $e \in \delta(S_1) \cup \delta(S_2)$).

Recall that, $e \in \Core_B(\Rel(Y))$, and observe that $E_0 = \Rel(Y) \cap \delta(S_0)$ is a random subset of $\Rel(Y)$ obtained in such a way that  every element of $\Rel(Y) - e$ is in $E_0$ independently with probability
\[
\frac{p_2 - p_1}{1 - p_1} = q.
\]
Since $\calA$ is $(q,\alpha)$-OOCS for $\cF_{B,E}$, if we ``feed'' $E_0$ as the sample of $\calA$, we conclude that, for every possible ordering $\sigma$ of $\Rel(Y) \setminus E_0$, the event
\[
\cE_0: e \text{ is accepted by } \calA \text{  when revealed},
\]
which depends only on the internal random coins of $\calA$ (recall that the random coins $r$ of $\calA$ are independent from the internal arrival order of the elements), satisfies 
\[
\E_{S_0 - e}\brk{\Pr_r\brk{\cE_0} \midd e \notin \delta(S_0)} \geq \alpha.
\]
Here we are crucially using the fact that $\calA$ is order-oblivious, since the order in which $\Rel(Y) \setminus E_0$ is presented to $\calA$ is chosen by the instructions of our algorithm.
Next, let $a_e$ be the agent in $S_0$ to which $e$ is incident. Condition on the realization of $S_0 - a_e$ -- in other words, $S_1 - a_e$ and $S_2-a_e$ are still undecided -- and let $E_2 = \Rel(Y) \cap \delta(S_2)$.

Let $h_1$ be the heaviest element in $E_1 - e$ such that $e$ is in the agent-hull of the elements with a weight greater than $h_1$ in $E_1 - e$. In other words,
\[
h_1 \coloneqq \argmax \set{w(h) \midd h \in E_1 - e, \text{ and } e \in \Hull_A(\set{g \in E_1 - e \midd w(g) \geq w(h)})}.
\]
Similarly for $h_2$ in $E_2-e$:
\[
h_2 \coloneqq \argmax \set{w(h) \midd h \in E_2 - e, \text{ and } e \in \Hull_A(\set{g \in E_2 - e \midd w(g) \geq w(h)})}.
\]
If there are no such elements, we set $h_1 = \bot$ and $h_2 = \bot$, respectively, with $w(\bot) = 0$.

Since, conditioned on $Y$ and $S_0-a_e$, the sets $(E_1 - e)$ and $(E_2 - e)$ are identically distributed, we have that $\Pr[w(h_1) \geq w(h_2)] = 1/2$, regardless of whether one or both of $h_1$ and $h_2$ are $\bot$. Next, we condition on the independent event that $a_e \in S_2$, and conclude that both $\mathcal{E}_1: w(h_1) \geq w(h_2)$ and $\mathcal{E}_2: a_e \in S_2$ occur together with probability at least $1/4$ given the previous conditions on $Y$ and $S_0-a_e$. So, in what follows, we also suppose that the two events $\mathcal{E}_1$ and $\mathcal{E}_2$ occur.

Next, recall that, $E_1 = \Rel(Y) \cap \delta(S_1)$ is ordered in decreasing order of weights as $e_1, \dots, e_m$. Let $j^*$ be such that $h_1=e_{j^*}$ (where $j^* = m+1$ if $h_1 = \bot$). By construction, $e$ is revealed by the algorithm precisely in iteration $j^*$, when it first enters the set $Q$. Furthermore, since $e \in \Core_A(\Rel(Y)) = \bigcap_{j = 1}^{M_A} \Core_{A,j}(\Rel(Y))$ we conclude, by the relations between $\Hull_A$ and $\Core_A$ and the monotonicity of $\Hull_A$ that
\[
e \notin \Hull_A(\set{g \in \Rel(Y) \midd w(g) > w(e)}) \supseteq \Hull_A(\set{g\in \Rel(Y)\cap \delta(S_1) \midd w(g) > w(e)}).
\]
Thus, $w(e) > w(h_1) \geq w(h_2)$, so upon the arrival of $e$, we satisfy that $w(e) > w(e_{j^*})$.

Let $\ALG'$ denote the solution immediately before $e$ is revealed and note that all elements in $\ALG'$ are in $\Rel(Y) \cap \delta(S_2) = E_2$ and have weights larger than $w(e_{j^*})$. There are two possibilities.
First, if $h_2 \neq \bot$, since $w(e_{j^*}) = w(h_1) \geq w(h_2)$, by the definition of $h_1$ and $h_2$, we have
\[
e \notin \Hull_A(\set{ g \in E_2 \midd w(g) \geq w(e_{j^*}) }) \supseteq \Hull_A(\ALG'),
\]
where the last inequality follows by the monotonicity of $\Hull$, and since all elements considered by $\ALG'$ for addition had weights larger than $w(e_{j^*})$.

Second, if $h_2 = \bot$ then, again by the monotonicity of $\Hull_A$, we have
\[
e \notin \Hull_A(E_2) \supseteq \Hull_A(\ALG').
\]
It follows that at the moment in which $e$ is revealed, $e \notin \Hull_A(\ALG')$, and thus $\ALG'+e\in \cI_A$ by the properties of $\Hull_A$.

Finally, let $\calA_0$ be the set of elements that $\calA$ selects before $e$'s arrival. We know that $\ALG'\subseteq \calA_0$ by the description of our algorithm. So, if $\calA$ selects $e$ upon arrival we would immediately have that $\ALG'+e \in \cI_A \cap \cI_B$ since the set of elements that $\calA$ selects are independent in $\cI_B$.
Therefore, so long as $\calA$ selects $e$ (i.e., as long as $\cE_0$ occurs), $e$ will be added to $\ALG$.
Putting it all together:
\begin{align*}
\Pr\brk{e \in \ALG \midd Y} &=\Pr[a_e \notin S_0] \cdot \E_{S_0 - a_e} \Pr\brk{e \in \ALG \midd Y, a_e \notin S_0} \\
&\geq \frac{1 - p_2}{1 - p_1} \E_{S_0-a_e} \Pr\brk{\cE_0 \wedge \cE_1 \wedge \cE_2 \midd Y, a_e \notin S_0} \\
&= (1-q) \E_{S_0-a_e}\brk{\Pr\brk{\cE_0 \midd Y, a_e \notin S_0} \cdot \Pr\brk{\cE_1 \wedge \cE_2 \midd Y, a_e \notin S_0}} \\
&\geq (1-q) \E_{S_0-a_e}\brk{\Pr\brk{\cE_0 \midd Y, a_e \notin S_0}} \cdot  \frac{1}{4} \\
&\geq \frac{\alpha (1-q)}{4}.
\end{align*}
Since this holds for every realization of $Y$ and every $e \in \Core_{\text{Common}}(\Rel(Y))$, for the case of general combinations of growth-systems and for $p = \sqrt{1-\frac{1}{k+1}}$, we conclude that
\begin{align*}
\E[w(\ALG)] &\geq \frac{\alpha (1-q)}{4} \E\brk{w(\Core_{\text{Common}}(\Rel(Y)))} \\
&\geq \frac{\alpha(1-q)}{4} \prn{2 - 2\sqrt{1 + \frac{1}{k}} +\frac{1}{k}} w(\OPT(E)),
\end{align*}
where the last inequality follows from Lemma~\ref{lem:core-agent-arrival}.

Similarly, for the case of $k$-matchoids and $p = 1 - \frac{1}{2k}$, we conclude that
\begin{align*}
\E[w(\ALG)] &\geq \frac{\alpha (1-q)}{4} \E\brk{w(\Core_{\text{Common}}(\Rel(Y)))} \\
&\geq \frac{\alpha(1-q)}{16k^2} \: w(\OPT(E)),
\end{align*}
where again the last inequality follows from Lemma~\ref{lem:core-agent-arrival}.
\end{proof}

\subsection{Remarks}\label{sec:agent-remarks}

It is worth noting that our free-order agent-arrival technique also applies, even when some item-side components do not admit an order-oblivious core-selecting (OOCS) algorithm. In particular, by using techniques similar to the reduce-and-solve construction of Feldman, Svensson, and Zenklusen~\cite{fsz-secretary-framework-matroid-intersection} for random-order matroid intersection, we can reduce those components to instances that do admit an OOCS and then apply our method unchanged, incurring only the loss dictated by the reduction. Without going into full generality, we briefly explain how to deal with some specific examples.

\paragraph{Graphic Matroids.}\;
We can reduce graphic matroids to a (random) partition matroid using an idea of Korula and P\'al \cite{korula-pal-graphic}. Given a graph $G = (V, E)$, before the process begins, let $\pi$ denote a uniformly random permutation of the vertices $V$. We create a partition matroid where every vertex $u$ is associated with a partition class. For any edge $e = \set{u,v}$, we place $e$ in the partition class of $u$ if and only if $\pi(u) < \pi(v)$. Now we can use an $(1/2, 1/2)$-OOCS for the resulting partition matroid, incurring only an extra loss of $1/2$ by the reduction. What is actually done in the agent-arrival model is to replace the graphic matroid component by the partition matroid just defined, even before starting the agent-arrival algorithm (in particular, greedy, and the greedy-relevant edges are computed with respect to this new combination of systems). Note that in the transformed instance we are not enforcing selection from the
core of the original graphic; instead, we select from the core of this new partition matroid.  Nevertheless, the expected weight guarantee remains sufficiently high because the analysis applies to the greedy-relevant edges of this new combined system. 

Interestingly, reductions like the one described above can be designed for any matroid class that has the $\alpha$-partition property \cite{matroid-partition}, so long as the randomized partition can be created a priori, before observing any elements. For example, this is the case for $k$-column sparse matroids (see \cite{soto-secretary}).

\paragraph{Transversal Matroids}
The idea is a little different here. Assume the item-side constraint is a combination in which one component is a transversal matroid~$T$. Represent $T$ by its standard bipartite presentation $H=(L,R;E_H)$: the left vertices $L$ are the original
ground elements, the right vertices $R$ are auxiliaries, and a set $F\subseteq L$ is independent in $T$
iff there exists a matching in $H$ that saturates $F$.

\smallskip
\noindent\emph{Duplication and partition intersection.}
For each $e\in L$ and each neighbor $r\in N_H(e)$, create a copy $e_r$ (think of it as the edge $(e,r)$),
and let $U=\{e_r : e\in L,\ r\in N_H(e)\}$ be the new ground set. Define two partition matroids on $U$:
\begin{align*}
\mathcal{P}_{\text{left}}&:\ \text{at most one element from each block } \{e_r : r\in N_H(e)\}\ \text{for every } e\in L,\\
\mathcal{P}_{\text{right}}&:\ \text{at most one element from each block } \{e_r : e\in L\}\ \text{for every } r\in R.\end{align*}
Then $\mathcal{P}_{\text{left}} \cap \mathcal{P}_{\text{right}}$ has independent sets exactly the matchings of $H$.
Projecting an independent set $I\subseteq U$ to $\{e\in L : \exists r\ \text{with}\ e_r\in I\}$ yields precisely the
$L$-vertices that can be simultaneously matched in $H$, i.e., an independent set of $T$.

\smallskip
\noindent\emph{Incorporation to the agent-arrival algorithm.}
Similar to what we did for graphic matroids, we will replace the component $T$ on the original combination of systems by the two components $\mathcal{P}_{\text{left}}$ and $\mathcal{P}_{\text{right}}$.
This is done at the very beginning, so both greedy and the greedy-relevant elements are computed with respect to the new combinations. Whenever the agent-arrival  algorithm decides to reveal an agent that owns a \emph{copy} of an original element $e\in L$ of the transversal in its neighborhood, actually all of its copies 
$\{e_r : r\in N_H(e)\}$ are simultaneously revealed (the weight that is revealed is actually $w(e)$ on all of them; break ties arbitrarily to maintain
injectivity if needed.) Accepting any copy $e_r$ is interpreted as accepting $e$. In the transformed system we run an
order-oblivious core-selecting (OOCS) procedure, since $\mathcal{P}$ is an intersection of partition matroids.

\smallskip
\noindent\emph{Matchoid parameter.}
This reduction is not preserving for the growth parameter: one original matroidal component is replaced by two. In particular,
if the original item-side constraint was a $k$-matchoid, the transformed constraint is a $(k+1)$-matchoid, so the competitive ratio  by this reduction we obtain is slightly decreased..

\smallskip
\noindent\emph{Core and guarantee.}
Similar to the graphic case in the transformed instance, we are \emph{not} enforcing selection from the core of the original
transversal matroid~$T$, but from the core of the new components 
$\mathcal{P}_{\text{left}}\cap\mathcal{P}_{\text{right}}$. This is still fine for the analysis to work.

\paragraph{Knapsack with integer sizes.}
We give a small construction that transforms knapsack constraints into a hypergraph matching, for the particular case in which all item sizes are integer numbers between $1$ and $k$.

Essentially, for a knapsack of total capacity $C$ and $n$ elements, we create a hypergraph $H=(V,E)$ with $C+n$ vertices. The first $C$ vertices correspond to unit slots in the knapsack, while the remaining $n$ vertices each correspond to an original element of the knapsack constraint. Order the elements of the original knapsack constraint as $e_1,\dots,e_n$. For every $e_i$ with size $s_i\in\mathbb{N}$, we create $C-s_i+1$ hyperedges. Each “copy” hyperedge $e_{i,j}$ (for $j=1,\dots,C-s_i+1$) contains the vertices $\set{j,\dots,j+s_i-1}\cup\set{C+i}$. The last vertex ensures we pick at most one copy of $e_i$; the chosen copy encodes the position of $e_i$ in the knapsack, ensuring that the final set of selected hyperedges is feasible. Furthermore, this hypergraph matching constraint is represented by a $(k+1)$-matchoid in which each component is a rank-1 uniform matroid, and each \emph{copy} hyperedge participates in at most $k+1$ of these matroids.

We replace the knapsack constraint with this $(k+1)$-matchoid (each piece of this matchoid admits an OOCS algorithm since it is uniform), and by using the same idea described for transversals, we can effectively replace the original elements with their copies in the agent-arrival model.
If instead of having sizes $1$ to $k$ we have a $k$-ratio-bounded knapsack (so the sizes range between some base size $m$ and $k\cdot m$), and if the capacity of the knapsack is sufficiently large compared to $m$, then we can still carry out a similar construction with only a slight loss in the approximation factor by choosing a sufficiently small $\varepsilon$ and rounding up each item size to the nearest integer multiple of $m\varepsilon$. By standard techniques, every independent set feasible for the original constraint contains a subset that is feasible in the new, rounded knapsack with almost the same total weight, so we use this rounded system instead and apply the previous construction. This incurs a constant blow-up (on the order of $1/\varepsilon$) in the number of matroids used to describe the system.
\section{Multiple Item Selection}\label{sec:multiple-items}

Our techniques for the free-order agent-arrival model on bipartite graphs require that each agent selects at most one incident edge. 
In this section, we consider a different extension of the bipartite-graph model, which we call the \emph{unrelated-agent} model.

\begin{definition}[Unrelated-agent model]
Let $G=(A\cup B, E)$ be a bipartite graph, and for $a\in A$ let $\delta(a)\subseteq E$ denote the set of edges incident to $a$. Each agent $a\in A$ has an independence system $\mathcal{F}_a=(\delta(a), \mathcal{I}_a)$ over its incident edges. There is also an independence system $\mathcal{F}_B=(B,\mathcal{I}_B)$ over the set of items $B$. 
A subset $F\subseteq E$ of edges is \emph{feasible} if:
\begin{enumerate}[label=(\roman*)]\itemsep0em
    \item for every agent $a$, the set of edges of $F$ incident to $a$ is independent, i.e., $\delta(a)\cap F \in \mathcal{I}_a$, \label{def:uama}
    \item the set of assigned items is independent, i.e., 
    $\{\, b\in B \colon \delta(b)\cap F \neq \emptyset \,\} \in \mathcal{I}_B$. \label{def:uamb}
\end{enumerate}
\end{definition}
We say that agents are \emph{unrelated} in this model because each has its own independence system, and there is no global system that constrains whether one agent can select elements based on the actions of another.
Our goal is to study the agent-arrival model in this setting, namely, upon arrival, each agent $a$ reveals the weights of all edges incident to it, and we must irrevocably select a subset $F_a\subseteq \delta(a)$ such that the union $\bigcup_{a\in A} F_a$ is feasible.

Although, in general, the entire problem can be encoded within the item system by treating each edge as an item -- so that conditions \ref{def:uama} and \ref{def:uamb} are captured by a single system $(B,\mathcal{I}_B)$ -- our purpose is to understand what can be achieved in the agent-arrival setting when the item system is simple. 
A canonical example is the case where the only constraint is that each item is assigned to at most one agent (i.e., the item system is a partition matroid).

This setting can be used to model problems such as the \emph{secretary problem with groups} proposed by Korula and P\'al \cite{korula-pal-graphic}. 
In this problem, there is a unique independence system $(B,\cI_B)$ whose ground set is partitioned (arbitrarily) into a known number $m$ of \emph{groups} $B_1,\dots, B_m$ (but the partition is unknown).
Once the groups have been formed, the groups arrive in random order, and upon arrival, the identity and weight of the elements are revealed and the algorithm must select a subset from each arriving group in such a way that the total set of selected elements is feasible for the system, and it has maximum weight.
We note that this model falls within our framework, as it can be represented by a bipartite graph between a set $A=\{a_1,\dots,a_m\}$ of $m$ agents and the set $B$ of items. The adversary will assign each group $B_i$ to agent $a_i$ by assigning $w(\{a_i, x\})=w(x)$ for all $x\in B_i$ and $w(\{a_i,x\})=0$ (or an infinitesimal small number) if $x\not\in B_i$. 

\subsection{The Case of Hypergraph Matching}

Consider the special case of the unrelated-agent model in which the item constraint is a hypergraph matching constraint.
That is, the items are hyperedges of an auxiliary hypergraph $H$, and the set of items (i.e., hyperedges) that can be used simultaneously must form a matching in $H$. 
In the {\it combinatorial assignment problem with agent arrivals}, we are given a bipartite hypergraph $H=(A\cup R,E)$ where each hyperedge $e$ contains exactly one node $A(e)$ in $A$ and a subset $R(e)\subseteq R$ of nodes in $R$, so we can always denote them as $e=(a,R(e))$, with $a=A(e)\in A$. The set $A$ is regarded as the agents. And the nodes in $R$ are the \emph{goods}. Each agent $a\in A$ owns a private independence system $\mathcal{F}_a=(\delta(a),\mathcal{I}_a)$ over the set $\delta(a)$ of hyperedges that contain $a$. Also, there is a global function $w\colon E\to \R_+$ where $w(e)=w(a,R(e))$ is interpreted as the weight that agent $a$ assigns to $R(e)$. 
We say that a set of hyperedges $F\subseteq E$ is feasible if:
\begin{enumerate}[label=(\roman*)]\itemsep0em
    \item for all $a\in A$, $\delta(a)\cap F \in \mathcal{I}_a$;
    \item every individual node of $b$ is in at most one hyperedge of $F$, i.e. for all $b\in R$, $|\delta(b)\cap F|\le 1$.
\end{enumerate}
An interpretation of this model comes from the combinatorial auction setting, where each hyperedge can be understood as a {\it bundle} of goods.
Each agent $a$ has a hidden valuation $w(e)=w(a,R(e))$ over bundles $e$ of goods,
and each agent $a$ can receive not just one bundle, but a collection of them, under the conditions that the assignment to $a$ must be feasible for the private independence system of agent $a$, i.e., $(\delta(a),\mathcal{I}_a)$, and each node in $R$ must be in at most one of the assigned hyperedges $F$.
The simplest example is when the hyperedges are actually edges: i.e., they consist of a single agent and a single item.
In this example, upon arrival, agent $a$ would reveal the value of all edges in $\delta(a)$, and then an algorithm must immediately assign a subset of it to $a$ in such a way that this subset is feasible for $a$'s own independence system, and each item can be assigned to at most one edge.
When the bundles have size at most $k$, we call this setting the $k$-combinatorial assignment with agent arrivals.

\paragraph{Relation to Previous Bipartite Models.}
We remark that the combinatorial assignment problem with agent arrivals is a direct extension of the bipartite vertex-at-a-time hypergraph matching problem studied by Korula and P\'{a}l \cite{korula-pal-graphic} and Kesselheim, Radke, T\"{o}nnis, and V\"{o}cking\cite{kesselheim-secretary-matching}, and the combinatorial assignment model by Marinkovic, Soto and Verdugo~\cite{marinkovic-soto-verdugo}.
The main difference is that, rather than limiting each agent to a single hyperedge, agent $a$ may receive any subset of $\delta(a)$ that is feasible in $\cF_a$. In the prior works, this corresponds to the special case where each $\cF_a$ is the rank-1 uniform matroid on $\delta(a)$.

\paragraph{The Algorithm.} Each agent selects a random arrival time uniformly and independently on the interval $(0,1)$ and uses it to represent the random-order arrival of the agents. We also assume that, for any set of agents $A'$, we can compute the optimum allocation $\OPT(A')$, i.e., the maximum weight feasible set of hyperedges in the hypergraph induced by $A'$ and $R$.
We need some notation. For any set $S\subseteq A$ of agents, any agent $a\in S$, let $F(S,a)=\OPT(S)\cap \delta(a)$ be the set of hyperedges of $\OPT(S)$ that are incident to $a$, and let $N(S,a)=\bigcup_{e\in F(S,a)}R(e)$ be the goods in $R$ that are assigned to $a$. 
For every $t>p$, let $X(S,a,t)$ be the subset of goods in $N(S,a)$ that are still unassigned by the algorithm by time $t$ (or more precisely, immediately before time $t$).

\begin{algorithm}[H]
\DontPrintSemicolon
    \( \ALG \gets \emptyset \);
     Ignore all agents arriving before time \( p \in (0,1) \)\;
    \For{each agent \( a \) arriving at time \( t_a \in [p,1] \), in arrival order}{
    Let \( A_a \) be the set of agents that have arrived up to time \( t_a \)\;
    Compute \( \OPT(A_a) \), and let \( F(a) = \OPT(A_a)\cap \delta(a)\) denote the edges of $\OPT(A_a)$ incident to the new agent $a$.\;
    Let $F'(a) = \{e\in F(a)\colon \forall f\in \ALG, R(e)\cap R(f)=\emptyset\}$, i.e., the subset of edges $e$ of $F(a)$ such that $R(e)$ is composed of goods that have not been assigned by the algorithm yet.\;    
    Assign to $a$ all edges in $F'(a)$, that is:
        \( \ALG \gets \ALG \cup F'(a) \)
    }
    \Return $\ALG$.
\caption{\textsc{Basic Online Assignment Algorithm}}
\label{alg:standard}
\end{algorithm}

\begin{lemma}\label{lem:XOS-condition} Let $S\subseteq A$ be a fixed set of at least $k+1$ agents. Let $i=|S|$ and suppose that $S$ is the set of the first $i$ arriving agents on the algorithm. 
Let $a$ be the last arriving agent in $S$, and $t=t(a)>p$ be its time of arrival. 
Let also  $D\in F(S,a)$ be a subset of goods among the ones assigned to $a$ in $\OPT(S)$, with $|D|\leq k$. 
Then the probability that $D\subseteq X(S,a,t)$, i.e., that all goods in $D$ are unassigned at time $t$, is at least $(p/t)^k$. 
\end{lemma}

We defer the proof of this lemma and show how to get the following theorem.

\begin{theorem}\label{thm:multiple-items}
$\ALG$ is $\int_{p}^1 (p/t)^k dt$-competitive for the $k$-combinatorial assignment with agent arrivals. By choosing appropriate values of $p$ for each $k$, the competitive ratio is $1/e$ for $k=1$ and 
$$k^{-k/(k-1)}=\frac{1}{k}-\frac{\ln(k)}{k^2}+\bigO{\frac{\ln^2(k)}{k^3}}\text{ for $k\geq 2$}$$
\end{theorem}
\begin{proof}
Let $n$ be the number of agents. 
Let \( \ALG_i \) be the contribution to the solution at the \( i \)-th arrival, that is, if $F\subseteq E$ is the final assignment and $a$ is the $i$-th arriving agent, then  
\( \ALG_i \coloneqq \sum_{e\in \delta(a)\cap \ALG} w(e) \).
For any set $S$ of agents, let \( \ell(S) \) be the last one to arrive inside $S$ in the random order. 
Let \( A_i \) be the set of the first \( i \) agents to arrive and let \( q_{i}(t) \) be the probability density function for the arrival time of the \( i \)-th agent among \( n \) uniformly arriving agents. Then:

\begin{align}
\mathbb{E}[\ALG_i] 
&= \frac{1}{\binom{n}{i}} \sum_{\substack{S \subseteq A \\ |S| = i}} \frac{1}{i}\sum_{a \in S} \int_0^1 \mathbb{E}[\ALG_i | A_i = S,\, a = \ell(S), t(a)=t] q_{i}(t)\, dt \notag  \\
&=\frac{1}{\binom{n}{i}} \sum_{\substack{S \subseteq A \\ |S| = i}} \frac{1}{i}\sum_{a \in S} \int_p^1 \mathbb{E}[\ALG_i | A_i = S,\, a = \ell(S), t(a)=t] q_{i}(t)\, dt, \label{eq:ALGi-contribution}
\end{align}
where the last equality holds since, before time $p$, we do not select any hyperedge.
Note that for fixed $S$, $a=\ell(S)$ and $t=t(a)$, Lemma \ref{lem:XOS-condition}
implies that
\begin{align*}
\mathbb{E}[\ALG_i | A_i = S,\, a = \ell(S), t(a)=t]&=  \sum_{e\in F(S, a)} w(e)\Pr(R(e)\subseteq X(S,a,t))\\
&\geq \sum_{e\in F(S,a)}w(e) (p/t)^k,
\end{align*}
since for all $e$, we assumed that $|R(e)|\leq k$. The assumption that $S$ has at least $k+1$ agents can be easily satisfied by adding to our original instance enough dummy agents without incident edges.
Therefore, summing \eqref{eq:ALGi-contribution} over all $i$, and changing the summation order, we have
\begin{align*}
\sum_{i=1}^n \mathbb{E}[\ALG_i] 
&\geq \int_p^1 (p/t)^k \sum_{i=1}^n\frac{q_i(t)}{i}  \frac{1}{\binom{n}{i}}  \sum_{\substack{S \subseteq A \\ |S| = i}} \sum_{a \in S}\sum_{e\in F(S,a)}w(e) \, dt \\
&= \int_p^1 (p/t)^k  \sum_{i=1}^n\frac{q_i(t)}{i}  \frac{1}{\binom{n}{i}} \sum_{\substack{S \subseteq A \\ |S| = i}} w(\OPT(S))\, dt\\
&= \int_p^1 (p/t)^k \sum_{i=1}^n\frac{q_i(t)}{i}  \EE_{\{S\colon |S|=i\}}[w(\OPT(S))]\, dt
\end{align*}
where the expectation is taken over all sets $S$ of exactly $i$ agents (uniformly). Now observe that:

\[
 \mathbb{E}_{\{S\colon |S|=i\}}[w(\OPT(S))]\geq \mathbb{E}_{\{S\colon |S|=i\}}[w(\OPT(A)\cap \delta(S))] \geq \frac{i}{n} \cdot w(\OPT(A)),
\]
since each agent in $A$ is included in \( S \) with probability \( i/n \), and the restriction \( \OPT(A) \cap \delta(S) \) is a feasible solution over \( S \). Putting it all together:
\begin{align*}
\mathbb{E}[w(\ALG)] &= \sum_{i=1}^n \mathbb{E}[\ALG_i] \\
&\geq \int_p^1 (p/t)^k \sum_{i=1}^n \frac{q_{i}(t)}{i}  \frac{i}{n} w(\OPT(A))\, dt \\
&= w(\OPT(A)) \int_p^1 (p/t)^k \cdot \frac{1}{n} \sum_{i=1}^n q_{i}(t)\, dt \\
&= w(\OPT(A)) \int_p^1 (p/t)^k\, dt = \begin{cases}
    w(\OPT(A)) \cdot p \ln(1/p) &\text{ if $k=1$},\\
    w(\OPT(A)) \cdot \frac{p-p^k}{k-1}&\text{ if $k\geq 2$}.
\end{cases}
\end{align*}

For $k=1$, we set  \( p = 1/e \), to obtain a competitive ratio of 
\( p \ln(1/p) = 1/e \).  
For $k\geq 2$, we set \(p =k^{-1/(k-1)}\) to obtain a competitive ratio of \(k^{-k/(k-1)}\)\end{proof}

We now show the proof of the missing lemma above.

\begin{proof}[Proof of Lemma \ref{lem:XOS-condition}]

Let $\ell$ be the number of agents of $S$ arriving before time $p$. Label all the arriving times of the agents increasingly as 
$\tau_1<\dots <\tau_\ell <p\;<\tau_{\ell+1}<\tau_{\ell+2}<\cdots<\tau_{i-1}< t = \tau_{i}.$
For convenience, denote $a^{q}$ to the $q$-th arriving agent so that $a^i=a$ and let $A_q=\{a^1,\dots, a^q\}$ be the set of the $q$ first agents. By the statement of the lemma, $A_{i-1}=S\setminus \{a\}$. 
Recall that $D\subseteq F(S,a)$ is a fixed set of items.

Note that the optimum assignment $\OPT(A_{i-1})$ at time $\tau_{i-1}$ assigns each good $b\in D$ to at most one agent in $A_{i-1}$; it follows that the number of agents that could get a good from $D$ at time $i-1$ is at most $|D|\leq k$. 
Since the identity of the agent arriving at time $\tau_{i-1}$ is chosen uniformly at random from the set $A_{i-1}$, with probability at least $(1-{k}/{(i-1)})_+ = {(i-1-k)_+}/{(i-1)}$, agent $a^{i-1}$ does not receive any good from $D$ at time $\tau_{i-1}$. 
Here $x_+=\max(x,0)$ is the positive part of $x$.
Now, suppose that we have already revealed the identities of $a_{q+1}, a_{q+2}, \dots, a_{i-1}, a_i=a$ and that none of them received goods from $D$. Then $A_q=S\setminus \{a_{q+1},\dots, a_{i-1},a_i\}$. 
Once again, we note that the optimum assignment at time $\tau_{q}$ assigns each good $b\in D$ to at most one agent in $A_q$ and that the identity of agent $a_{q}$ is uniformly chosen at random from $A_q$. 
Then, with probability at least $(1-k/q)_+ ={(q-k)_+}/{q}$, agent $a_q$ does not receive any good from $D$.  
By repeating this until $q=\ell+1$, we conclude that
\begin{align*}
\Pr[D\subseteq X(S,a,t)\, |\, \ell, \tau_{\ell+1},\dots, \tau_{i-1}, t=\tau_i ] &\geq \frac{(i-1-k)_+}{i-1} \cdot \frac{(i-2-k)_+}{i-2} \cdot \ldots \cdot \frac{(\ell+1-k)_+}{\ell+1}\\
&=\frac{\binom{\ell}{k}}{\binom{i - 1}{ k}}.
\end{align*}
Therefore, 

\[\Pr[D\in X(S,a,t) | t=\tau_i] = \frac{\E[\binom{\ell}{k}]}{\binom{i - 1}{ k}},\]
where the expectation is over the random variable $\ell$ that counts the number of the  agents in $S\setminus \{a\}$ that arrive in the interval $[0,p]$. 
Since all agents in $S\setminus \{a\}$ chose their arrival time in $[0,t]$ uniformly at random, we get $\ell\sim \Bin(i-1,p/t)$, and therefore, using known properties of the factorial moments of binomial variables we get $\Pr[D\in X(S,a,t) | t=\tau_i] = (p/t)^k.$
\end{proof}
\paragraph{Extensions.}
We note that the same analysis used for the hypergraph matching case can be applied to more general independence systems on the item side. 
In fact, following Marinkovic, Soto and Verdugo\cite{marinkovic-soto-verdugo}, we can replace the hypergraph matching system (which is the intersection of $k$ unitary partition matroids) by the combination of matroids admitting a so-called {\it directed certifier}~\cite[Definition 3]{marinkovic-soto-verdugo}.
Our approach can also be extended to the case of more general weight functions for the agents, e.g., {\it fractionally subadditive}.\footnote{$v$ is fractionally subadditive (a.k.a., XOS) if 
there is a fixed family $w_1,\dots, w_k$ of additive weight functions over $R$ such that for all $S$ we have $v(S)=\max_{i\in [k]}w_i(S).$
}
Namely, for each agent $a$, consider the function $v_a(F)=\max_{I\subseteq F,I \in \mathcal{I}_a}\sum_{e\in I} w(e),$
that is, if $F$ is a set of edges incident to $a$, then $v_a(F)$ is the maximum weight of a subset of $F$ that can be assigned to $a$.
For each $a$, the function $v_a$ is fractionally subadditive. 
In this variant, the combinatorial assignment problem with agent arrival becomes an online combinatorial allocation problem with XOS valuations over \emph{hyperedges}. 
Kesselheim, Radke, T\"{o}nnis, and V\"{o}cking\cite{kesselheim-secretary-matching} gave an algorithm that works for the specific problem of XOS valuations when all $R(e)$ are singletons, achieving a $1/e$ competitive algorithm.\footnote{In \cite{kesselheim-secretary-matching} the algorithm is stated for submodular valuations, but the proof also works for XOS valuations.}
Our approach applies to the case of fractionally subadditive weights over hyperedges as well, generalizing the result by Kesselheim et al~\cite{kesselheim-secretary-matching}.
\section{Properties of the \texorpdfstring{$k$}{k}-systems Hierarchy}\label{sec:higher-order-systems}

The study of $k$-systems and their subclasses plays a key role in understanding the structural and algorithmic limits of independence systems beyond matroids. These classes capture different levels of combinatorial flexibility and have been useful for understanding the trade-offs between structural generality and algorithmic tractability in submodular optimization, secretary problems, and related areas. The introduction of $k$-growth systems offers a promising new class that likely refines this hierarchy, although its exact relation to previously studied systems remains to be fully understood. We believe that the proposed class is of combinatorial interest on its own, and one of the implicit goals of this paper is to initiate its study.

In this section, we give an overview of $k$-systems, present examples of $k$-growth systems, and explain how $k$-growth systems fit in the hierarchy. We also study the closure properties of several subclasses of $k$-systems under fundamental operations such as restriction and contraction. In particular, we give a characterization of $k$-extendible systems in terms of contractions of $k$-systems and a characterization of $k$-circuit bounded systems as the systems that satisfy a natural generalization of the classical circuit axiom for matroids. 

\subsection{Overview of Existing Classes}

Matroids are characterized by their defining exchange properties. Since their introduction, several related classes have been proposed, with each subsequent class assuming a weaker form of these properties. We give a brief overview of the most important subclasses of $k$-systems.

\begin{definition}[$k$-matroid intersection]\label{def:k-matroid-intersection}
An independence system $(S,\cI)$ is a \emph{$k$-matroid intersection} system if there exist matroids $\mathcal{M}_i = (S, \mathcal{I}_i)$ for $i\in[k]$, such that
\[
\cI = \{ F \subseteq S \colon F \in \mathcal{I}_i~ \text{for all $i\in[k]$} \}.
\]
\end{definition}

\begin{definition}[$k$-matchoid]\label{def:k-matchoid}
An independence system $(S,\cI)$ is a \emph{$k$-matchoid} if there exist matroids $\mathcal{M}_i = (S_i, \mathcal{I}_i)$ for $i\in[N]$, where the ground sets $S_i$ may intersect, such that $S=\bigcup_{i=1}^N S_i$, each $e\in S$ belongs to at most $k$ of the ground sets $(S_i)_{i=1}^N$,  and
\[
\cI = \{ F \subseteq S \colon F \cap S_i \in \mathcal{I}_i\ \text{for all $i\in[N]$} \}.
\]
\end{definition}

\begin{definition}[$k$-circuit bounded]\label{def:k-circuit-bounded}
An independence system $(S,\cI)$ is \emph{$k$-circuit bounded} if for every independent $I\in \cI$ and every $e \not \in I$, $I+e$ contains at most $k$ distinct circuits.
\end{definition}

\begin{definition}[$k$-extendible]\label{def:k-extendible}
An independence system $(S,\cI)$ is \emph{$k$-extendible} if, for every $A\subseteq B\in \cI$ and every $e\in S\setminus B$, we have the following  $k$-extension property:
\[
A+e\in \cI \implies \exists Z\subseteq B\setminus A, \text{such that } |Z|\leq k \text{ and } (B+e)\setminus Z \in \cI.
\]
\end{definition}

\begin{definition}[$k$-system]\label{def:k-system}
An independence system $\cF=(S,\cI)$ is a $k$-system if, for every subset $X\subseteq S$ and every pair $P, Q$ of bases of $X$, we have $|P| \leq k \: |Q|$.
\end{definition}

The following hierarchy of independence systems for $k \in \N$ is well-known, see e.g.~\cite{feldman2011improved}.
\begin{align*}
\text{$k$-matroid intersection} \subseteq \text{$k$-matchoid} \subseteq \text{$k$-circuit bounded} \subseteq \text{$k$-extendible} \subseteq \text{$k$-system}.
\end{align*}
It is easy to see that $1$-systems coincide with matroids, and therefore, for $k = 1$, all classes are the same. For $k \geq 2$, however, all containments in the hierarchy are known to be strict. Furthermore, as Proposition \ref{prop:greedy-matroid-ksystem} states, an independence system is a $k$-system if and only if, for every nonnegative weight function, the weighted greedy algorithm returns a $1/k$-approximation to the maximum-weight basis.

\subsection{Growth Systems: A New Class of Independence Systems}

Recall the definition of $k$-growth systems, presented here again for convenience.
\kGrowthDefn*
It is easy to see that condition ($\star1$) can be replaced by any of the following equivalent ones.
\begin{enumerate}[label=($\star$\arabic*)]\itemsep0em
\setcounter{enumi}{1}
\item For every independent set $P\in \cI$ with $P\subseteq X$, $P\cup Q \in \cI$.
\item $Q$ is a subset of every basis of $X\cup Q$.
\item For every circuit $C$, if $C\subseteq X\cup Q$ then $C\subseteq X$.
\end{enumerate}
We remark that to check whether an independence system satisfies the \ref{def:kbg} axiom, it is enough to check it for all sets $X$ and bases $B$ of the independence system.

\paragraph{Relations to the $k$-systems hierarchy.}
Next, we describe how $k$-growth systems are related to other classes in the hierarchy of independence systems.

\begin{theorem}
The 1-growth systems are precisely matroids. Moreover, in this case the set $Z$ guaranteed by the axiom can be chosen to satisfy the stronger bound $|Z|\leq r(X)-|X\cap I|$.
\end{theorem}
\begin{proof}
Let $\cF=(S,\cI)$ be a $1$-growth system, and let $J$ and $I$ be two independent sets with $|J|<|I|$. Apply the $1$-growth axiom \hyperlink{kbg}{($1\mathrm{BG}$)} with $X=J$ to obtain a partition $(Q,Z)$ of $I\setminus J$ such that $|Z|\le |J\setminus I|$. Since $|I|>|J|$, we have $|I\setminus J|>|J\setminus I|$, hence $|Z|<|I\setminus J|$ and so $Q=(I\setminus J)\setminus Z\neq\emptyset$. For any $e\in Q$, the axiom guarantees $J+e\in\cI$. Thus the augmentation property holds, and $\cF$ is a matroid.  

Now let $\cM=(S,\cI)$ be a matroid, and let $X \subseteq S$ and $I\in \cI$. Take $Q$ to be any basis of $I\setminus X$ in the contracted matroid $\cM/X$. Then, by the properties of contraction, we know that for every basis $P$ of $X$ in $\cM$, the set $P\cup Q$ forms a basis of $X\cup I$, and therefore is independent. Define $Z=(I\setminus (X\cup Q))$. Then $(Q,Z)$ is a partition of $I\setminus X$, and it only remains to show that $|Z|\leq |X\setminus I|$. To see this, let $R$ be a basis of $X\setminus I$ in $\cM/(X\cap I)$. Since $X\cap I\in \cI$, this means that $R\cup (X\cap I)$ is a basis of $X$ in $\cM$. By the properties of $Q$ discussed above, $R \cup (X \cap I) \cup Q$ is a basis of $X \cup I$ in $\cM$.  Moreover, since $I \in \mathcal{I}$, it follows that $|R\cup (X\cap I) \cup Q| \geq |I|=|X\cap I|+|Q|+|Z|$. In particular, $|Z| = |R| = r(X)-|X \cap I|\leq |X|-|X\cap I|=|X\setminus I|$, as needed.
\end{proof}

As it turns out, $k$-growth systems fit nicely in the hierarchy between $k$-matchoids and $k$-extendible systems, for any $k \in \N$. In fact, for the former containment, we show a stronger result.

\kGrowthComb*
\begin{proof}
Let $X \subseteq S$ and $I \in \cI$. Fix a $j \in [M]$. Since $(S_j,\cI_j)$ is a $k_j$-growth system, there exist sets $Z_j \subseteq (I\setminus X) \cap S_j$ and $Q_j = ((I\setminus X) \cap S_j) \setminus Z_j$ with $|Z_j| \leq k_j |(X\setminus I) \cap S_j|$ such that for every independent set $P_j$ of $X \cap S_j$ in $(S_j,\cI_j)$, we have $P_j \cup Q_j \in \cI_j$. If $k_j = 1$, the $j$-th component is a matroid and the same holds with the stronger bound $|Z_j| \leq r_j((X\setminus I) \cap S_j) - |(X \cap S_j) \cap I|$.

Now let
\[
Z \coloneqq \bigcup_{j = 1}^M Z_j, \quad  Q \coloneqq (I \setminus X) \setminus Z.
\]
For every $P \in \mathcal{I}$ and $j\in [M]$, we have $(P \cap S_j) \cup Q_j  \in \mathcal{I}_j$. Since $Q\cap S_j \subseteq Q_j$, it follows that $(P\cup Q)\cap S_j = (P\cap S_j) \cup (Q\cap S_j) \subseteq (P\cap S_j) \cup Q_j$. Therefore, $(P \cup Q) \cap S_j \in \mathcal{I}_j$, and thus we have $P\cup Q\in \cI$. 

Finally, note that
\[
\abs{Z} = \abs{\bigcup_{j = 1}^M Z_j} \leq \sum_{j=1}^M \abs{Z_j} \leq \sum_{j=1}^M k_j |(X \setminus I) \cap S_j| = \sum_{e \in X \setminus I} \sum_{j \colon e \in S_j} k_j \leq k |X\setminus I|,
\]
so the system is a $k$-growth system.
\end{proof}

\begin{theorem}\label{thm:relations}
The following statements hold for any $k \in \N$.
\begin{enumerate}[label=(\alph*)]\itemsep0em
    \item Every $k$-growth system is $k$-extendible. \label{kGiskE}
    \item Every $k$-matchoid is a $k$-growth system. \label{kMiskG}
\end{enumerate}
\end{theorem}
\begin{proof}
For \ref{kGiskE}, let $(S,\cI)$ be a $k$-growth system, and let $A\subseteq B\in \cI$, $e\in S\setminus B$ such that $A+e\in \cI$. Applying \ref{def:kbg} to $X=A+e$ and $I=B$ we deduce that there exists a set $Z\subseteq B\setminus (A+e)=B\setminus A$ with $|Z|\leq k  |(A+e)\setminus B|=k$ such that $(A + e) \cup (B\setminus (A\cup Z)) = (B \setminus Z) + e \in \cI$. Thus, $(S, \cI)$ satisfies the $k$-extension property.

Finally, notice that \ref{kMiskG} follows directly from Lemma \ref{lem:k-growth-combinations} since, for matroids, we have $k_j = 1$ for all $j \in [M]$.
\end{proof}

Next, we show that the first containment is strict.

\begin{theorem}\label{thm:extendible-growth-separation}
There exist a $2$-extendible system that is not $2$-growth.
\end{theorem}
\begin{proof}
    Let $\cF = (S, \cI)$ be an independence system on the ground set $S = \{x,1,2,3,4,5,6\}$ whose circuits are
    \[
    C_1 = \{x, 1, 2, 3\},  \quad  C_2 = \{x, 2, 3, 4\},  \quad  C_3 = \{x, 1, 3, 5\},  \quad  C_4 = \{x, 1, 2, 6\}.
    \]
We first check $2$-extendibility using a hitting-set argument, and then exhibit a counterexample showing that the $2$-growth property does not hold.

\paragraph{\normalfont\itshape Proof of the 2-extendible property.}
Let $A \subseteq B \in \cI$ and $e \in S \setminus B$ such that $A+e \in \cI$. We must find $Z \subseteq B \setminus A$ with $|Z| \leq 2$ such that $(B+e) \setminus Z$ contains no circuit.
\medskip

\noindent \textbf{Case 1: $e \in \{1, 2, 3, 4, 5, 6\}$.}
Every circuit in $\mathcal{C}$ contains $x$. If $x \notin B$, then $B+e$ contains no circuit, so $Z = \emptyset$ suffices. Suppose $x \in B$. Since $B \in \cI$, $B$ does not contain the remaining elements $C_i \setminus \{x\}$ of any circuit. Adding $e$ can only complete circuits $C_i$ for which $e \in C_i$.
\begin{itemize}\itemsep0em
    \item If $e=1$, the potential circuits are $C_1, C_3, C_4$. We must choose $Z \subseteq B \setminus A$ to hit the sets $\{x, 2, 3\}, \{x, 3, 5\}, \{x, 2, 6\}$. By the assumption $A+e \in \cI$, $A$ does not contain any of these sets. Furthermore, since $x \in B$, if $x \notin A$, then $x \in B \setminus A$, and removing $Z=\{x\}$ breaks all circuits. If $x \in A$, then $Z$ must be chosen from $\{2, 3, 5, 6\}$. A set $Z$ of size 2 (e.g., $Z = \{2, 3\}$ or $Z = \{3, 6\}$) is sufficient to hit all three configurations. 
    \item By symmetry, adding $e=2$ or $e=3$ creates at most three circuits, which can similarly be broken by a hitting set of size at most 2. Adding $e \in \{4, 5, 6\}$ completes at most one circuit, requiring $|Z| \leq 1$.
\end{itemize}

\noindent \textbf{Case 2: $e = x$.}
If $e=x$, then $B \subseteq \{1, 2, 3, 4, 5, 6\}$. The neighborhoods of $x$ inside the circuits are
\[ N_1 = \{1, 2, 3\}, \quad N_2 = \{2, 3, 4\}, \quad N_3 = \{1, 3, 5\}, \quad N_4 = \{1, 2, 6\}. \]
The set $B+x$ is dependent if $N_i \subseteq B$ for some $i$. Let $\mathcal{N} = \{ N_i \colon N_i \subseteq B \}$. Because $A+x \in \cI$, no $N_i$ is fully contained in $A$. Thus, for every $N_i \in \mathcal{N}$, the intersection $N_i \cap (B \setminus A)$ is nonempty. We must find a hitting set $Z \subseteq B \setminus A$ for $\mathcal{N}$ such that $|Z| \leq 2$.
\begin{itemize}
    \item If $N_1 \in \cN$, then $(B \setminus A) \cap \{1, 2, 3\} \neq \emptyset$. Let $v$ be such an element. Since $v$ hits three of the four possible sets in $\cN$, at most one set $N$ remains in $\cN \setminus \{N_i \colon v \in N_i\}$. We can pick one element from $N\cap (B \setminus A)$ to hit $N$, yielding $|Z| \leq 2$.
    \item If $N_1 \notin \mathcal{N}$, then at least one element $v \in \{1,2,3\}$ is not in $B$. Because each such $v$ is contained in three of the $N_i$ sets, the absence of $v$ from $B$ implies that those three sets are not in $\mathcal{N}$. Thus, $\mathcal{N}$ contains at most one set $N$, which can be hit by any element in $N\cap (B\setminus A)$.
\end{itemize}

\paragraph{\normalfont\itshape Failure of the 2-growth property.} 
Let $X \coloneqq  \{x,1,2,3\}$ and $I \coloneqq  \{1,2,3,4,5,6\}$. The set $I$ is independent, since it contains no circuit, as all circuits contain $x$. The bases of $X$ are
\[ P_1 = \{1, 2, 3\}, \quad P_2 = \{x, 1, 2\}, \quad P_3 = \{x, 2, 3\}, \quad P_4 = \{x, 1, 3\}. \]
To satisfy \hyperlink{kbg}{($2\mathrm{BG}$)}, there must exist $Q \subseteq I \setminus X = \{4,5,6\}$ such that $|Z| = |\{4,5,6\} \setminus Q| \leq 2$, and $P_j \cup Q \in \cI$ for all $j=1,2,3,4$. We now show that no choice of $Q$ satisfies this requirement.
\begin{enumerate}\itemsep0em
    \item If $4 \in Q$, then $P_3 \cup \{4\} = \{x, 2, 3, 4\} = C_2 \notin \cI$. Thus, $4 \notin Q$.
    \item If $5 \in Q$, then $P_4 \cup \{5\} = \{x, 1, 3, 5\} = C_3 \notin \cI$. Thus, $5 \notin Q$.
    \item If $6 \in Q$, then $P_2 \cup \{6\} = \{x, 1, 2, 6\} = C_4 \notin \cI$. Thus, $6 \notin Q$.
\end{enumerate}
Hence $Q = \emptyset$, and therefore $Z = \{4,5,6\}$, which violates $|Z| \leq 2$. Thus the axiom \hyperlink{kbg}{($2\mathrm{BG}$)} fails.
\end{proof}

We also show that the second containment is strict; in fact, we separate $k$-growth systems not only from $k$-matchoids but also from the larger class of $k$-circuit bounded systems.

\begin{theorem}\label{thm:growth-circuit-separation}
For any $k \geq 2$, there exist $k$-growth systems that are not $k$-circuit bounded.
\end{theorem}
\begin{proof}
Let $S=\set{e_0,e_1,\dots,e_{k+1}}$ and consider the family of bases $\cB=\set{\set{e_0,e_i}\colon i\in[k+1]}\cup\set{\set{e_1,\dots,e_{k+1}}}$. Let $\cF = (S, \cI)$ where $\cI = \set{F \subseteq B \midd B \in \cB}$.

\begin{claim}\label{clm:growth}
$\cF$ is a $k$-growth system.
\end{claim}
\begin{proof}
By our previous observation, it suffices to check condition \ref{def:kbg} for all sets $X$ and bases $B$ of the independence system. If $e_0 \in B$, then $|B| = 2$, and the $k$-growth condition \ref{def:kbg} is trivially satisfied for any $X$. For $B = \set{e_1,\dots,e_{k+1}}$, $X \setminus B \neq \emptyset$ only holds if $X \setminus B = \set{e_0}$. In this case, regardless of $B \setminus X$, we can always find a $Z$ with $|Z| \leq k$ such that $|B \setminus (X \cup Z)| \leq 1$, in which case adding it to any basis of $X$ will form an independent set.
\end{proof}

\begin{claim}\label{clm:not-circuit-bounded}
$\cF$ is not $\bigl(\binom{k+1}{2}-1\bigr)$-circuit bounded.
\end{claim}
\begin{proof}
For $I = \set{e_1,\dots,e_{k+1}}$, the set $I+e_0$ contains $\binom{k+1}{2}$ circuits of the form $\set{e_0,e_i,e_j}$ for distinct $i,j\in[k+1]$.
\end{proof}

The theorem follows by Claims~\ref{clm:growth} and~\ref{clm:not-circuit-bounded}.
\end{proof}

Finally, we show that $k$-circuit bounded systems are not $k$-growth. In fact, the circuit boundedness can be made to be as small as $2$.

\begin{theorem}\label{thm:circuit-not-under-growth}
For any $k \geq 2$, there exist $2$-circuit bounded systems that are not $k$-growth.
\end{theorem}
\begin{proof}
Let $n = 2k+1, A = \set{a_1, \dots, a_n}, B = \set{b_1, \dots, b_n}$ and $Y = \set{y_{i,j} \midd 1 \leq i,j \leq n}$. We define a set system $\cF$ on the ground set $S = A \cup B \cup Y$ via its circuits, which are
\begin{itemize}
    \item All pairs $\set{a_i, a_j}$ for $i,j \in [n]$,
    \item All pairs $\set{b_i, b_j}$ for $i,j \in [n]$,
    \item All triples $\set{a_i, b_j, y_{i,j}}$ for $i,j \in [n]$.
\end{itemize}

It is easy to see that $\cF$ is $2$-circuit bounded. There are two cases for any independent set:
\begin{itemize}
    \item It contains some $a_i$, some $b_j$ and a subset of $Y \setminus \set{y_{i,j}}$, or
    \item It does not contain both some $a_i$ and some $b_j$, and contains a subset of $Y$.
\end{itemize}
In the first case, adding an element $a_{i'}$ creates at most two circuits: $\set{a_i, a_{i'}}$ and $\set{a_{i'}, b_j, y_{i', j}}$. Adding an element $b_{j'}$ is similar by symmetry. Finally, adding $y_{i,j}$ creates one circuit, namely $\set{a_i, b_j, y_{i,j}}$ and no other elements create a circuit. In the second case, if the independent set is a subset of $Y$, then no added element can create a circuit. So suppose, without loss of generality, that it contains some $a_i$ and no elements from $B$; the case for $b_j$ and no elements of $A$ is similar. Then, adding some $a_{i'}$ creates at most one circuit, namely $\set{a_i, a_{i'}}$. Adding some $b_j$ creates at most one circuit, namely $\set{a_i, b_j, y_{i,j}}$, and finally adding any element of $Y$ does not create a circuit.

We next show that $\cF$ is not $k$-growth. Let $X = A \cup B$ and $I = Y$. The bases of $X$ are the sets $\set{a_i, b_j}$ for $1 \leq i,j \leq n$. If $\cF$ is $k$-growth, there exists a partition of $I = Z \cup Q$ such that $|Z| \leq k |X| = 2kn$ and, for every basis $\set{a+i, b_j}$ of $X$, the set $Q \cup \set{a_i, b_j}$ is independent. To see why this cannot happen, notice that $Q \neq \emptyset$. Indeed, $|I| = n^2 = (2k+1)^2 = 4k^2 + 4k +1$, whereas $|Z| \leq 2kn = 2k(2k+1) = 4k^2 + 2k$. Therefore, $|Q| = |I| - |Z| \geq 2k + 1$. Let $y_{i,j} \in Q$, and take the basis $\set{a_i, b_i}$ of $X$. The set $\set{a_i, b_j, y_{i,j}}$ forms a circuit, and thus $Q \cup \set{a_i, b_j}$ contains a circuit, contradicting the $k$-growth axiom. Therefore, $\cF$ is not $k$-growth.
\end{proof}

\subsection{Examples of \texorpdfstring{$k$}{k}-growth Systems}

Theorems~\ref{thm:relations}, \ref{thm:extendible-growth-separation}, and~\ref{thm:growth-circuit-separation} show that $k$-growth systems lie properly between $k$-matchoids and $k$-extendible systems in the $k$-system hierarchy, while they do not form a subclass of $k$-circuit-bounded systems. 
To gain a better understanding of the class, we next present some fundamental examples.

\paragraph{Knapsack systems.} We also show that the class of $k$-growth systems includes specific classes of knapsack systems. Let $S$ be a finite set and $s\colon S\to \RR_{\geq 0}$ be a nonnegative function. The associated \textit{knapsack system} is $(S,\cI)$ where $\cI=\{I\subseteq S \colon \sum_{e\in I} s(e) \leq 1\}$. We say that the knapsack system has \textit{ratio bound} $\rho$ if $\max_{e\in S} s(e) = \rho \cdot \min_{e\in S} s(e)$. 

\begin{theorem}
Every knapsack system $(S,\cI)$ with ratio bound $\rho$ is a $\lceil \rho\rceil$-growth system.
\end{theorem}
\begin{proof}
Let $X\subseteq S$ and $I\in \cI$. Let $\alpha=\min_{e\in S}s(e)$, $\beta=\max_{e\in S}s(e)$ and set $k\coloneqq\lceil \beta/\alpha\rceil$. We distinguish two cases. 

\paragraph{Case 1.} If $s(X\setminus I)\ge s(I\setminus X)$, then we have \[\alpha|I\setminus X|\leq s(I\setminus X)\leq s(X\setminus I)\leq \beta |X\setminus I|, \]
yielding $|I\setminus X|\leq k|X\setminus I|$. Then we can set $Z\coloneqq I\setminus X$ and $Q\coloneqq\emptyset$.
\medskip

\paragraph{Case 2.} If $s(X\setminus I)<s(I\setminus X)$, then $s(X)<s(I)\leq 1$ and therefore $X\in \cI$.  If $|I\setminus X|\leq k|X\setminus I|$, then we can simply set $Z\coloneqq I\setminus X$ and $Q\coloneqq\emptyset$. Otherwise $|I\setminus X|>k|X\setminus I|$, and set $Z$ to be any set of the largest  $k|X\setminus I|$ elements of $I\setminus X$, and $Q\coloneqq I\setminus (X\cup Z)$. Note that the size bound for $Z$ is attained and
\[
s(Z)\geq \alpha |Z| = \alpha k |X\setminus I| \geq \frac{\alpha}{\beta}ks(X\setminus I) \geq s(X\setminus I).
\]
Therefore
\[s(X\cup Q)=s(X\setminus I)+s(X\cap I)+s(Q) \leq s(Z)+s(X\cap I)+s(Q) = s(I)\leq 1,\]
thus $X\cup Q\in \cI$.
\end{proof}

Notice that the above can be extended to $d$-sparse packing constraints. For a packing constraint $A x \leq B$ where each element of $x$ appears in at most $d$ rows of $A$, let $\rho = \max_j \rho_j$ be the maximum ratio bound across all knapsack constraints corresponding to rows of $A$. Since the packing constraint can be seen as a combination of the knapsack constraints corresponding to rows of $A$, Lemma~\ref{lem:k-growth-combinations} shows that the packing constraint is $(d \cdot \rho)$-growth.

\paragraph{Stable Sets of Intersection Graphs.}
Let $D$ be a ground set and $V\subseteq 2^D$ be a family of subsets; we refer to elements of $D$ and $V$ as \emph{points} and \emph{figures}, respectively. The \textit{intersection graph} of the family is an undirected graph $G=(V,E)$ with
\[
E=\bigl\{\,uv\in\tbinom{V}{2}\colon u\cap v\neq\emptyset\,\bigr\}.
\]
A set $W\subseteq V$ is considered as \emph{stable} or \textit{independent} in $G$ if and only if its figures are pairwise disjoint. Let $(V,\cI)$ be the independence system of stable sets of $G$.

When $D$ is finite, computing a stable set is equivalent to choosing, for each point $p\in D$, \emph{at most one} figure that contains $p$. Thus $(V,\cI)$ is the intersection of $|D|$ partition matroids, one per point. If every figure contains at most $k$ points, then each figure participates in at most $k$ of those matroids and the system is a $k$-matchoid. Moreover, the circuits all have size $2$: if a set is not stable, it contains two figures meeting at some point, and any proper subset is stable. 

We focus on a special family of figures defined by intervals on a line. Let $D=[N]$, and for integers $1\le a\le b\leq N$, let $V_{a,b}$ denote the collection of all intervals whose sizes lie between $a$ and $b$, where the size of an interval is the number of its points. That is,
\[
V_{a,b}=\{[i,j]\colon 1\le i\le j\le N,\ a\le j-i+1\le b\}.
\]
By the above, $(V_{a,b},\cI)$ is a $b$-matchoid. We show it is also $c$-circuit-bounded and a $c$-growth system with $c=\bigl\lfloor\tfrac{b-2}{a}\bigr\rfloor+2$.

This construction can be further generalized: suppose that we have a capacity $\ell$ such that at most $\ell$ intervals can be selected at each point in time. Then, the corresponding system is $(\ell \cdot c)$-growth, where $c = \bigl\lfloor\tfrac{b-2}{a}\bigr\rfloor+2$. This capacitated interval constraint captures \emph{unsplittable flow on a path (UFP)}, by viewing the path edges as the ``points'' of the interval system. Specifically, in UFP, each request/task $i$ corresponds to a subpath $P_i$ of a given path and a demand $d_i$, and each path edge $e$ has capacity $u_e$. A set of accepted requests is feasible if and only if, for every path edge $e$, $\sum_{i: e\in P_i} d_i \leq u_e$. This is exactly a capacitated interval constraint, since selecting a request consumes $d_i$ units of capacity at every point/edge in the corresponding interval.
\medskip

\noindent \emph{Circuit boundedness.} Let $R=\{v_1,\dots,v_r\}\subseteq V_{a,b}$ be stable and take any interval $v'=[\ell,r]\in V_{a,b}$ such that $R\cup\{v'\}\notin\cI$. Let
\[
Z(v')=\{\,v\in R\colon v\cap v'\neq\emptyset\,\}
\]
be the intervals in $R$ intersecting $v'$. The number of circuits created by adding $v'$ is exactly $|Z(v')|$; our goal is to bound $|Z(v')|$.

Let $Z_0(v')\subseteq Z(v')$ be those intervals that do not contain either endpoint $\ell$ or $r$ of $v'$. Then every $v\in Z_0(v')$ is a subinterval of the interior of $v'$, so the total number of points they cover is at most $|v'|-2\le b-2$. Since each interval has length at least $a$, we get
\[
a\,|Z_0(v')|\ \le\ |v'|-2\ \le\ b-2,
\]
yielding $|Z_0(v')|\ \le\ \bigl\lfloor\tfrac{b-2}{a}\bigr\rfloor$. Since at most two additional intervals in $Z(v')$ can contain the endpoints $\ell$ and $r$, we obtain
\[
|Z(v')|\ \le\ |Z_0(v')|+2\ \le\ \bigl\lfloor\tfrac{b-2}{a}\bigr\rfloor+2\ =\ c,
\]
proving $c$-circuit-boundedness.
\medskip

\noindent \emph{Being $c$-growth.} Let $X\subseteq V_{a,b}$ be arbitrary, $I\in\cI$ be stable, and set $R\coloneqq I\setminus X$. For each $v'\in X\setminus I$, define $Z(v')\subseteq R$ as above; by the previous bound, we have $|Z(v')|\le c$. Let
\[
Z\coloneqq \bigcup_{v'\in X\setminus I} Z(v'),\qquad Q\coloneqq R\setminus Z.
\]
Every interval in $Q$ is disjoint from every interval in $X\setminus I$ by construction, and also from every interval in $X\cap I$ by the independence of $I$. Therefore, for any stable set $P\subseteq X$, we have $P\cup Q\in\cI$. Moreover,
\[
|Z|\ \le\ \sum_{v'\in X\setminus I} |Z(v')|\ \le\ c\,|X\setminus I|\ \le\ c\,|X|,
\]
showing the $c$-growth property \hyperlink{kbg}{($c\mathrm{BG}$)}.
\medskip

\noindent \emph{Not a $c$-matchoid.} Finally, let us show that for certain values of $a$ and $b$,  the system above is not a $c$-matchoid. We give an example below with $N=8$, $a=b=3$ and $c=\lfloor (b-2)/a\rfloor + 2 = 2$.

Denote $v_j=\{j,j+1,j+2\}$ for $j\in[6]$ the six intervals of size 3. Suppose that the system of stable sets $(S,\cI)$ is a $2$-matchoid given by a family of matroids. Since every singleton is independent and every circuit of $(S,\cI)$ has size $2$, we conclude that each circuit $\{x,y\}$ of $(S,\cI)$ is a matroid-circuit in at least one of the matroids.

The four pairs $\{v_3,v_1\}$, $\{v_3,v_2\}$, $\{v_3,v_4\}$, and $\{v_3,v_5\}$ are circuits. Since $v_3$ belongs to exactly two matroids, these circuits must be assigned into two groups, one for each of the matroids, say $M_A$ and $M_B$, not necessarily disjoint. Not every assignment is feasible: by circuit elimination, the pairs $\{v_3,v_1\}$ and $\{v_3,v_4\}$ cannot lie in the same matroid, as this would force $\{v_1,v_4\}$ to contain a circuit, contradicting its independence. The same argument excludes placing $\{v_3,v_1\}$ together with $\{v_3,v_5\}$, and $\{v_3,v_2\}$ together with $\{v_3,v_5\}$. Consequently, the only possibility is that $\{v_3,v_1\}$ and $\{v_3,v_2\}$ are circuits in $M_A$, while $\{v_3,v_4\}$ and $\{v_3,v_5\}$ are circuits in $M_B$.

Applying the same reasoning to the four circuits $\{v_4,v_2\}$, $\{v_4,v_3\}$, $\{v_4,v_5\}$, and $\{v_4,v_6\}$, we deduce that $v_4$ is contained in two distinct matroids, say $M_C$ and $M_D$. Moreover, $\{v_4,v_2\}$ and $\{v_4,v_3\}$ must be circuits in $M_C$, while $\{v_4,v_5\}$ and $\{v_4,v_6\}$ are circuits in $M_D$.

Furthermore, by circuit elimination, $\{v_1,v_2\}$ is also a circuit in $M_A$, and $\{v_4,v_5\}$ is a circuit in $M_B$. Observe that $v_3$ belongs to the ground sets of $M_A$, $M_B$, and $M_C$, and since $M_A\neq M_B$, it follows that $M_C\in\{M_A,M_B\}$. If $M_C=M_A$, then $M_A$ contains the circuits $\{v_1,v_2\}$ and $\{v_2,v_4\}$. By circuit elimination, $\{v_1,v_4\}$ would also be a circuit, contradicting that $\{v_1,v_4\}\in\cI$ is independent. Similarly, if $M_C=M_B$, then $M_B$ contains the circuits $\{v_4,v_5\}$ and $\{v_2,v_4\}$, and circuit elimination would imply that $\{v_2,v_5\}$ is a circuit, again contradicting independence. Hence $M_A$, $M_B$, and $M_C$ must be pairwise distinct matroids containing $v_3$, which contradicts the assumption that $(S,\cI)$ is a $2$-matchoid.

\subsection{Parallel Extension, Restriction, and Contraction in Independence Systems}

In matroid theory, the most important classes are those closed under basic operations, such as restriction, contraction, and parallel extension. Determining whether analogous closure properties hold for more general independence systems is also of particular interest. In this section, we investigate this problem and related questions. We first define the counterpart of these two operations for independence systems.

\begin{definition}[Parallel extension]
Let $\cF=(S,\cI)$ be an independence system and $s\in S$ be a non-loop element. Set $S_{\circ s}\coloneqq (S-s)\cup\{s_1,s_2\}$, and define the \emph{projection} of a set $X\subseteq S_{\circ s}$ as $\pi_s(X)=X$ if $X\cap\{s_1,s_2\}=\emptyset$ and $\pi_s(X)=(X\setminus\{s_1,s_2\})+s$ otherwise. Then, the \textit{parallel extension} of $\cF$ along $s$ is the system $\cF_{\circ s} = (S_{\circ s}, \cI_{\circ s})$ with \[\cI_{\circ s}=\{X\subseteq S_{\circ s}\colon |X\cap\{s_1,s_2\}|\leq 1,\pi_s(X)\in\cI\}.\]   

\end{definition}

\begin{definition}[Restriction]
Let $\cF=(S,\cI)$ be an independence system and $Y\subseteq S$. We define the \textit{restriction} of $\cF$ to $Y$ as the system $\cF|_Y=(Y, \cI|_Y)$ with
\[
\cI|_Y = \set{X\subseteq Y\colon X\in \cI}.
\]    
\end{definition}

The two operations above agree with the corresponding ones on matroids.

\begin{definition}[Contraction]
Let $\cF=(S,\cI)$ be an independence system and $I\in \cI$ be an independent set. We define the \textit{contraction} of $I$ in $\cF$ as the system $\cF/I = (S\setminus I, \cI/I)$ with
\[\cI/I=\{X\subseteq S\setminus I\colon X\cup I \in \cI\}.\]    
\end{definition}

The contraction of an independent set coincides with the notion of contraction in matroids when we contract an independent set, but it is undefined when we try to contract a set that is not independent. 

The next result shows that most of the classes behave well with respect to these operations.

\begin{lemma}\label{lem:extendible-contractions}
If $\cF=(S,\cI)$ is $k$-matroid intersection/$k$-matchoid/$k$-circuit bounded/$k$-growth/$k$-extendible, then so are its parallel extensions, restrictions and contractions.
\end{lemma}
\begin{proof}
The claim for restrictions is direct. For parallel extensions and contractions, we prove the statement for each class separately. We first consider parallel extensions. Let $s\in S$ be an arbitrary element. 
\medskip

\noindent \textit{$k$-matchoids.} If $\cF$ is defined by matroids $(S_i,\cI_i)$ for $i\in[N]$, then its parallel extension $\cF_{\circ s}$ is the $k$-matchoid defined by the matroids obtained by performing the same parallel extension on each of them.
\medskip

\noindent\textit{$k$-circuit-bounded.} Take an independent set $I\in\cI_{\circ s}$ and let $e\in S_{\circ s}\setminus I$. If $I\cap\{s_1,s_2\}=\emptyset$ or $e\notin\{s_1,s_2\}$, then there is a one to one correspondence between the circuits of $\cF_{\circ s}$ in $I+e$ and of $\cF$ in $\pi_s(I+e)$. If $I\cap\{s_1,s_2\}\neq\emptyset$ and $e\in\{s_1,s_2\}$, then $I+e$ contains the unique circuit $\{s_1,s_2\}$. Since $\cF$ is $k$-circuit bounded, $I+e$ contains at most $k$ such circuits in both cases. 
\medskip

\noindent\textit{$k$-growth.} Take an arbitrary set $X\subseteq S_{\circ s}$ and an independent set $I\in\cI_{\circ s}$. If $|(I\cup X)\cap \{s_1,s_2\}|\leq 1$ or $\set{s_1,s_2}\subseteq X$, then the $k$-growth axiom \ref{def:kbg} simply follows from that for the projections in the original system. Otherwise, since $I\in\cI_{\circ s}$, we have $|(X\setminus I)\cap\set{s_1,s_2}|=|(I\setminus X)\cap\{s_1,s_2\}|=1$; we may assume without loss of generality that $s_1\in S\setminus I$ and $s_2\in I\setminus X$. Then \ref{def:kbg} for $\pi_s(X)$ and $\pi_s(I)$ in the original system yields a set $Z\subseteq \pi_s(I)\setminus \pi_s(X)$ such that $P\cup (\pi_s(I)\setminus Z)\in\cI$ for every basis $P$ of $\pi_s(X)$, and $|Z|\leq k|\pi_s(X)\setminus \pi_s(I)|=k|X\setminus I|-k$. Here, the last inequality follows from $s\in\pi_s(X)\cap \pi_s(I)$. Then, choosing $Z'\coloneqq Z+s_2$ satisfies $P'\cup(I\setminus Z')\in\cI_{\circ s}$ for every basis $P'$ of $X$ and $|Z'|=|Z|+1\leq k|X\setminus I|-k+1\leq k|X\setminus I|$.
\medskip

\noindent\textit{$k$-extendible.} Take any triple $(A,B,e)$ with $A\subseteq B$, $B\in \cI_{\circ s}$, $e\in S_{\circ s}\setminus B$, and $A+e\in \cI_{\circ s}$. If $B\cap\{s_1,s_2\}=\emptyset$ or $e\notin\{s_1,s_2\}$, then the $k$-extension property of $\cF$ implies that there exists $Z\subseteq \pi_s(B)\setminus \pi_s(A)$ with $|Z|\leq k$ such that $(\pi_s(B)+\pi_s(e))\setminus Z \in \cI$. Let $Z'\coloneqq Z$ if $s\notin Z$ and $Z'\coloneqq (Z-s)\cup(B\cap \{s_1,s_2\})$ otherwise. Then $Z'$ satisfies $(B+e)\setminus Z'\in \cI_{\circ s}$ and $|Z'|\leq k$. If $B\cap\{s_1,s_2\}\neq\emptyset$ and $e\in\{s_1,s_2\}$, then $B+e$ contains the unique circuit $\{s_1,s_2\}$, thus $Z=\{s_1,s_2\}-e$ is a proper choice.

\bigskip

Now we consider contractions. Let $I$ be an independent set in $\cI$ and $J\subseteq S\setminus I$ an independent set of $\cF/I$, that is, $I\cup J\in \cI$.
\medskip

\noindent \textit{$k$-matchoids.} If $\cF$ is defined by matroids $(S_i,\cI_i)$ for $i\in[N]$, then the contraction $\cF/I$ is the $k$-matchoid defined by the matroids $(S_i\setminus I,\cI_i/(I\cap S_i))$ for $i\in[N]$.
\medskip

\noindent\textit{$k$-circuit-bounded.} Let $e\in (S\setminus I)\setminus J$ and $C$ be a circuit of $\cF/I$ contained in $J+e$. Then $I\cup (J+e)\notin \cI$, but $I\cup (J+e-f)\in \cI$ for every $f\in C$. These together imply that $C=C'\setminus I$ for some circuit $C'\subseteq (I\cup J)+e$ of $\cF$. Since $(I\cup J)+e$ contains at most $k$ circuits of $\cF$, we have that $J+e$ contains at most $k$ circuits of $\cF/I$.
\medskip

\noindent\textit{$k$-growth.} Take an arbitrary set $X \subseteq S \setminus I$ and consider the sets $Y \coloneqq X \cup I$ and $F \coloneqq J \cup I$. Note that $F \in \cI$ by definition and, since $\cF$ is a $k$-growth system, there exists $Z \subseteq F \setminus Y$ such that $|Z| \leq k |Y \setminus F|$ and for any basis $P$ of $Y$, we have $P \cup (F \setminus (Y \cup Z)) \in \cI$. However, we have that $I \subseteq Y \cap F$, and thus $Y \setminus F = X \setminus J$, $F \setminus Y = J \setminus X$, and $F \setminus (Y \cup Z) = J \setminus (X \cup Z)$. Let $Q$ be a basis of $X$ in $\cF/I$. Then
$Q \cup I$ is a basis of $Y$: indeed, $Q \cup I \in \cI$ by definition of contraction, and if there
existed $e \in Y \setminus (Q \cup I) = X \setminus Q$ such that $(Q \cup I)+e \in \cI$, then
$Q+e \in \cI/I$, contradicting that $Q$ is a basis of $X$ in $\cF/I$.
Therefore, applying the
$k$-growth axiom to the basis $Q \cup I$ of $Y$, we get $(Q \cup I) \cup (F \setminus (Y \cup Z)) \in \cI$.
Equivalently, $(Q \cup I) \cup (J \setminus (X \cup Z)) \in \cI$, and thus
$Q \cup (J \setminus (X \cup Z)) \in \cI/I$. Therefore, $\cF/I$ is a $k$-growth system.
\medskip
\medskip

\noindent\textit{$k$-extendible.} We show the $k$-extension property for a triple $(A,B,e)$ with $A\subseteq B$, $B\in \cI/I$, $e\in (S\setminus I)\setminus B$, and $A+e\in \cI/I$. By the definition of contraction, we have $I\cup A\subseteq I\cup B\in \cI$, and $I\cup (A+e)\in \cI$. Since the original system is $k$-extendible, then there exists $Z\subseteq (B\cup I)\setminus (A\cup I)=B\setminus A$ with $|Z|\leq k$ such that $(I\cup (B+ e))\setminus Z \in \cI$. Therefore $(B+e)\setminus Z \in \cI/I$.
\end{proof}

Interestingly, the case for $k$-systems is different.

\begin{lemma} \label{lem:ksys}
Every parallel extension and restriction of a $k$-system is a $k$-system. However, for $k\geq 2$, the contraction of a $k$-system is not necessarily a $k$-system.
\end{lemma}
\begin{proof} 
Closure under parallel extension follows immediately, since duplicating an element into parallel copies does not affect the relative sizes of maximal independent sets in any subset, while closure under restriction is trivial by definition. For closure under contraction, consider the system with $S=\{x,a,b,c,d\}$ and bases $\{a,b,c,d\}$, $\{x,a\}$, $\{x,b\}$, $\{x,c\}$, and $\{x,d\}$. Since every set of size two is independent, for every $X\subseteq S$ with $|X|\le 4$ and for every pair $P,Q$ of bases of $X$, we have $|P|\ge|X|/2\ge|Q|/2$. If $X=S$, then all bases of $X$ have size either $2$ or $4$, so the system is a $2$-system. However, after contracting $a$, both $P=\{b,c,d\}$ and $Q=\{x\}$ are bases, with $|P|/|Q|=3$, and the resulting independence system is not a $2$-system.
\end{proof}

The following theorem reveals a key structural link between $k$-systems and $k$-extendible systems.

\begin{theorem} 
An independence system $\cF=(S,\cI)$ is $k$-extendible if and only if all of its contractions are $k$-systems.
\end{theorem}
\begin{proof} 
If $\cF$ is $k$-extendible, then by Lemma~\ref{lem:extendible-contractions}, all its contractions are $k$-extendible, and therefore also $k$-systems.

Now suppose that $\cF$ is a $k$-system whose contractions are $k$-systems. We show that it is $k$-extendible. Let $A\subseteq B\in\cI$ and $e\in S\setminus B$ with $A+e\in\cI$. Since $A+e\subseteq B+e$, choose $Z\subseteq B\setminus A$ of minimum size such that $(B+e)\setminus Z\in\cI$. If $B+e\in\cI$, then $Z=\emptyset$ works and we are done; otherwise assume $B+e\notin\cI$. Then $\{e\}\in\cI/(B\setminus Z)$ and $Z\in\cI/(B\setminus Z)$. We claim that both $\{e\}$ and $Z$ are bases of $Z+e$ in the contracted system $\cF/(B\setminus Z)$. Indeed, if $Z$ were not a basis, then $Z+e$ would be independent in $\cI/(B\setminus Z)$, implying $(Z+e)\cup(B\setminus Z)=B+e\in\cI$, a contradiction. Similarly, if there existed $f\in Z$ such that $\{e,f\}\in\cI/(B\setminus Z)$, then $(B\setminus Z)\cup\{e,f\}=(B+e)\setminus(Z-f)\in\cI$, contradicting the minimality of $Z$ as $Z-f$ would have smaller size. Therefore, both $\{e\}$ and $Z$ are bases of $Z+e$ in $\cI/(B\setminus Z)$. As $\cI/(B\setminus Z)$ is a $k$-system, we must have $|Z|\le k|\{e\}|=k$, and hence $(S,\cI)$ is $k$-extendible.
\end{proof}

Finally, we characterize $k$-circuit bounded systems. Recall that the circuit axiom for matroids says that an independence system is a matroid if and only if for every two circuits $C_1, C_2$ such that $e \in C_1 \cap C_2$ for some element $e$, we have that $(C_1 \cup C_2) \setminus \{e\}$ contains a circuit. We show that the natural generalization of the above axiom captures exactly the class of $k$-circuit bounded systems.

\begin{theorem}\label{thm:circuit-bounded-characterization}
An independence system $\cF=(S,\cI)$ is $k$-circuit bounded if and only if, for any $k+1$ distinct circuits $C_1, \dots, C_{k+1}$ of $\cF$ such that $e \in \bigcap_{i = 1}^{k+1} C_i$ for some $e \in S$, we have that $\bigcup_{i = 1}^{k+1} C_i \setminus \set{e}$ contains a circuit.
\end{theorem}
\begin{proof}
If $\cF$ is $k$-circuit bounded, we know that for every independent set $I \in \cI$ and element $e \in S$, $I+e$ contains at most $k$ circuits. Let $C_1, \dots, C_{k+1}$ of $\cF$ such that $e \in \bigcap_{i = 1}^{k+1} C_i$ for some $e \in S$, and let $I = \bigcup_{i = 1}^{k+1} C_i \setminus \set{e}$. Assume, towards contradiction, that $I$ does not contain a circuit; thus $I$ is independent. Consider adding $e$ to $I$. The resulting set $I+e$ contains $k+1$ circuits, contradicting the fact that $\cF$ is $k$-circuit bounded.

Now suppose that, for any $k+1$ distinct circuits $C_1, \dots, C_{k+1}$ of $\cF$ such that $e \in \bigcap_{i = 1}^{k+1} C_i$ for some $e \in S$, we have that $\bigcup_{i = 1}^{k+1} C_i \setminus \set{e}$ contains a circuit. For an arbitrary independent set $I \in \cI$ and element $e \in S$, assume, towards contradiction, that $I+e$ contains at least $k+1$ distinct circuits $C_1, \dots, C_{k+1}$. Since $I$ is independent, all such circuits contain $e$. Let $J = \bigcup_{i = 1}^{k+1} C_i \setminus \set{e}$ and notice that $J \subseteq I$. By our axiom, $J$ contains a circuit, which contradicts the fact that $I$ is independent.
\end{proof}
\section{Conclusion and Future Directions}\label{sec:discussion}

In this paper, we study the free-order secretary problem in a two-sided bipartite setting, where both agents and items are subject to independent feasibility constraints. This model generalizes standard bipartite matching secretary problems and captures matroid, matroid intersection, and other combinatorial constraints in a unified framework. We also introduce $k$-growth systems, a new class of independence systems that lie properly between $k$-matchoids and $k$-extendible systems and may be of independent combinatorial interest. Building on these structural insights, we design $\Omega(1/k^2)$-competitive algorithms for edge- and agent-arrival settings over $k$-growth systems and their combinations, as well as constant-competitive algorithms for settings where agents may select multiple items and the item-side constraint is sufficiently simple.

Our work leaves several natural questions open. 
\begin{enumerate}[itemsep=0em, left=0pt]
    \item We show a generalization of the Feldman--Svensson--Zenklusen core lemma for $k$-growth systems. Whether or not one can obtain a core lemma for $k$-extendible systems, for suitable definitions of primitive core and primitive hull, is a very interesting open problem.
    \item While we obtain $\bigOm{1/k^2}$-competitive algorithms for $k$-growth systems, it remains an intriguing question whether this dependence on $k$ is optimal, or whether constant competitiveness can be achieved for broader classes such as $k$-extendible systems. 
    \item Our results rely on the free-order model; extending these guarantees to the classical random-order setting, even for restricted subclasses of two-sided systems, appears to require fundamentally new ideas. 
    \item In the agent-arrival model, our approach relies on order-oblivious algorithms on the item side; it would be interesting to see if this requirement can be relaxed, and more generally, to identify further classes of independence systems that admit a core lemma and clarify the boundary between tractable and intractable two-sided online selection problems.
\end{enumerate}

\bibliographystyle{abbrv}
\bibliography{references}

\end{document}